\documentclass[11pt, a4paper]{article}
\usepackage{jheparxiv}
\usepackage[utf8]{inputenc}
\usepackage{amsmath}
\usepackage{bbold}
\usepackage{chessboard} 
\usepackage{amsfonts}
\usepackage{array, caption, floatrow, tabularx, makecell, booktabs}
\usepackage{amssymb}
\usepackage{latexsym}
\usepackage{mathrsfs}
\usepackage{braket}		%Bra Ket Notation
\usepackage{wasysym}
\usepackage{oplotsymbl}
\usepackage{graphicx}
\usepackage{simpler-wick}
\usetikzlibrary{tikzmark}

\newcommand{\cc}{\mathsf{c}}

\newcommand{\cK}{\mathcal{K}}

\newcommand{\Coh}{\mathtt{Coh}}

\newcommand{\nor}{\boldsymbol{:}}

\newcommand{\cAmp}{\mathtt{Amp}}

\newcommand{\dd}{\mathsf{d}}

\newcommand{\ca}{\mathfrak{ca}}
\newcommand{\Res}{\mathtt{Res}}

\newcommand{\Mat}{\mathtt{Mat}}

\newcommand{\tT}{\mathtt{T}}
\newcommand{\Ccurl}{\mathscr{C}}

\newcommand{\tU}{\mathtt{U}}

\newcommand{\Spec}{\mathtt{Spec}}
\newcommand{\tA}{\mathtt{A}}
\newcommand{\tV}{\mathtt{V}}

\newcommand{\msu}{\mathfrak{su}}

\newcommand{\nn}{\nonumber}

\newcommand{\sn}{\mathsf{n}}

\newcommand{\sA}{\mathsf{A}}

\usepackage{tabularray}

\usepackage{tikz-cd}
\usepackage{graphicx}
\usepackage{color}
\usepackage{xcolor}
\usepackage{slashed}
\usepackage{twistor}
\usepackage[all]{xy}
\usepackage{tikz-cd}
\usepackage{mathtools}
\usepackage{float}
\usetikzlibrary{arrows, matrix}
\tikzcdset{arrow style=math font}
\tikzset{>=latex}

\newcommand{\tJ}{\mathtt{J}}
\newcommand{\cI}{\mathcal{I}}

\newcommand{\tp}{\mathtt{p}}

\newcommand{\tH}{\mathtt{H}}

\usepackage{skak}

\newcommand{\Tr}{\text{Tr}}

\newcommand{\sh}{\mathsf{h}}

\newcommand{\mg}{\mathfrak{g}}

\newcommand{\B}{\mathbb{B}}

\newcommand{\eps}{\epsilon}

\newcommand{\msf}[1]{\mathsf{#1}}

\newcommand{\tC}{\mathtt{C}}
\newcommand{\Sym}{\mathrm{Sym}}

\title{\Large \center Chiral higher-spin symmetry of the celestial twistor sphere}

\author[]{Tung Tran}
\affiliation[]{Asia Pacific Center for Theoretical Physics, POSTECH,\\
77 Cheongamro, Nam-gu,
Pohang-si, Gyeongsangbuk-do, 37673, Korea}
\emailAdd{tran.tung@apctp.org}

\abstract{We study the chiral higher-spin symmetry algebras $\ca$ of various twistorial higher-spin theories. These symmetries play the roles of asymptotic symmetries on the celestial twistor sphere, which constrain the observables of twistorial theories. 
To first order in quantum correction, we show that the chiral algebras associated with anomaly-free holomorphic twistorial higher-spin theories are associative themselves. On the other hand, the chiral algebras associated with anomalous holomorphic twistorial higher-spin theories only become associative upon including  suitable axionic currents. When computing $4d$ form factors in terms of correlation functions between higher-spin currents on the celestial twistor sphere, we observe that there are some non-vanishing higher-spin form factors. % suggesting that there can be non-trivial higher-spin imprints on the celestial sphere. 
This observation, however, is only well justified for the case of theories with Yang–Mills–like interactions. We also give some brief comments on the case of %twistorial higher-spin theories with 
higher‑derivative interactions.%is also discussed
}

\begin{document}

\maketitle
\newpage
%%%%%%%%%%%%%%%%%%%%%%%%%%%%%%%%%%%%%%%%%%%%%%%%%%%%%%%%%%%%%%%%%%%%%%%%%%%%%%%%%%%%
\section{Introduction}\label{sec:1}
%Living in the infrared (IR) regime of physics has not stopped us from thinking about the physics at the ultraviolet (UV) counterpart. %This line of sophisticated thought is particularly interesting, 
%In particular, as irrelevant deformations, manifesting as higher-derivative interactions of a given QFT in the low-energy regime, become significant in the UV, they can lead to various issues such as causality violations and inconsistency. There is, however, a potentially simple remedy to this problem: introducing higher-spin fields. 

Higher-spin theory (see \cite{Bekaert:2022poo} for a review) is the study of theories involving interactions between fields with spin greater than two. When formulating these theories in terms of field theories, there are certain subtleties in interpreting their non-local interactions \cite{Bekaert:2015tva}. %See, however, [] for a proposal of a worldline model for higher-spin gravity in AdS${}_4$, which is based on the recent work \cite{Basile:2023vyg,Joung:2024akb}. 
Nonetheless, there are higher-spin theories which are sufficiently close in structures with the standard QFTs, despite being conformal \cite{Segal:2002gd,Tseytlin:2002gz,Bekaert:2010ky,Basile:2022nou}, chiral/self-dual \cite{Metsaev:1991mt,Metsaev:1991nb,Ponomarev:2016lrm,Metsaev:2018xip,Skvortsov:2018uru,Metsaev:2019dqt,Metsaev:2019aig,Tsulaia:2022csz}, or topological \cite{Blencowe:1988gj, Bergshoeff:1989ns, Campoleoni:2010zq, Henneaux:2010xg, Pope:1989vj, Fradkin:1989xt, Grigoriev:2019xmp, Grigoriev:2020lzu,Sharapov:2024euk}.\footnote{See \cite{Campoleoni:2024ced} for a review of $3d$ higher-spin gravity, see also \cite{Chen:2025xlo} and references therein for the studies of the edge modes associated with $3d$ higher-spin gravities.} Note, however, that these QFT-like higher-spin theories do not couple to particles in the usual way as in GR \cite{Basile:2024hjg,Ivanovskiy:2025kok}, making the notion of locality slightly subtle. %As a consequence, higher-spin gravities should ultimately be interpreted as some sort of tensionless string theories.

There has been some progress in constructing actions for higher-spin theories from top-down approaches. For instance, it is now well-known that $3d$ higher-spin gravities \cite{Campoleoni:2024ced} can be written in terms of Chern-Simons 3-forms. 
Similarly, the construction of chiral higher-spin gravity \cite{Ponomarev:2016lrm} and its closed subsectors \cite{Ponomarev:2017nrr,Krasnov:2021nsq} has also revealed somewhat similar structures with $3d$ higher-spin gravities. In particular, when the cosmological constant is zero, we can formulate these theories either as holomorphic BF theories or as holomorphic Chern-Simons theories on twistor space, cf. \cite{Tran:2021ukl,Tran:2022tft,Adamo:2022lah,Herfray:2022prf}.\footnote{See \cite{Sharapov:2022faa,Sharapov:2022wpz,Sharapov:2022awp} for the construction of the equations of motion for chiral higher-spin gravity with non-vanishing cosmological constant.} The story is, however, slightly different for conformal higher-spin gravity, cf. \cite{Segal:2002gd,Tseytlin:2002gz,Bekaert:2010ky,Basile:2022nou}. In particular, its action can be constructed in any even dimensions $d\geq 4$ via Feigin-Felder-Shoikhet cocycle \cite{feigin2005hochschild} with suitable symmetry constraints imposed on Fedosov's geometric data \cite{Fedosov:1994zz}. 

While top-down constructions are useful in controlling local structures, they, however, offer only limited insight on observables such as scattering amplitudes and correlation functions -- which are expected to be governed by the symmetries at asymptotic infinity. It has been observed that suitable deformations of chiral or self-dual theories can give rise to non-trivial higher-spin amplitudes \cite{Adamo:2022lah,Tran:2022amg}, which suggests an analysis of observables at asymptotic infinity may provide useful hints for inverse-bootstrapping some unknown theories. Motivated by this aspect, we aim to study the chiral higher-spin symmetries on the celestial twistor sphere $\P^1_{\tp}$ over some point $\tp$ in spacetime,\footnote{This celestial twistor sphere can be viewed as a complex codimension-2 defect in twistor space, and it can be shown to be equivalent with the usual celestial sphere in the literature \cite{Strominger:2017zoo} in the affine patch.} which are expected to govern observables of twistorial higher-spin theories \cite{Tran:2025uad,Mason:2025pbz} and certain deformations thereof. %This name is inspired from terminologies in \cite{Costello:2020jbh}.

To construct chiral higher-spin symmetry algebras $\ca$, we consider a stack of $N$ D${}_5$ branes filling twistor space, and another stack of D${}_1$ branes wrapping the celestial sphere (viewed as defect in twistor space). These branes, in turn, determine the group structures which fields in twistor space, and symmetry operators on the celestial twistor sphere take values in. %Note that this bulk/defect system can also be referred to as D${}_5$-D${}_1$ system.
What we will do in this work is to use Koszul duality/homomorphism to transfer the associative structures of the underlying symmetries which govern holomorphic higher-spin theories in twistor space to $\P^1_{\tp}$. 
%one can systematically construct the chiral algebras associated with holomorphic higher-spin theories in twistor space. 
%In particular, what Koszul homomorphism does is to transfer 
The chiral symmetry algebras $\ca$ %, which we will denote from now on as $\ca$, 
are then the vertex operator algebras, whose associative products are defined by the OPEs of the holomorphic higher-spin currents on the defect. 

Note that the above construction is perturbative in nature, so it does not guarantee $\ca$ to be associative when there are quantum corrections, a priori. 
Nonetheless, once we ensure $\ca$ is associative up to a given quantum order in the perturbation theory, then the correlation functions of higher-spin currents, which generate $\ca$, can be interpreted as form factors in some $4d$ spacetime theories \cite{Witten:2003nn}. Remarkably, this computation can be done algebraically by doing simple Wick contractions. As a result, one can mitigate the complexity of constructing explicit spacetime vertices or summing over many diagrams, which occurs when computing scattering amplitudes using traditional approaches \cite{Weinzierl:2022eaz}. Therefore, we may be able to bootstrap higher-spin amplitudes (in some cases), and search for non-trivial imprints on the celestial twistor sphere with the hope that we may detect some mysterious $4d$ theories that have not been constructed in the literature. Note that the indispensable criteria for this algorithm to work is to ensure gauge invariance of the bulk/defect (twistor space/celestial sphere) system and associativity of $\ca$ order by order in perturbation theory. 

%that preserves associativity structures on two sides of the above bulk/defect system in the framework of twisted holography \cite{Costello:2017dso} --

The organization of the paper is as follows:
\begin{itemize}
    \item[-] Section \ref{sec:2} introduces some basic notion of higher-spin symmetry and examples of twistorial higher-spin theories, together with the Green-Schwarz anomaly cancellation mechanism in twistor space. This section is mainly based on \cite{Tran:2025uad}. See also \cite{Costello:2021bah,Costello:2022upu,Bittleston:2022jeq} for relevant work.    
    \item[-] Section \ref{sec:3} studies $\ca$, which can be viewed as higher-spin extension of the chiral algebras studied in \cite{Costello:2022wso,Costello:2022upu,Costello:2023vyy,Bittleston:2022jeq}. We show, at classical level, that the chiral higher-spin algebras $\ca$ are associative and can be identified with the so-called color-kinematic algebras of chiral higher-spin theories \cite{Ponomarev:2017nrr,Monteiro:2022xwq}. We also show that the associativity of $\ca$ is \emph{not always} protected from quantum effects. In particular, the OPEs of higher-spin currents, which generate $\ca$, can receive non-trivial quantum corrections, leading to the failure of associativity. Note that to restore associativity at one loop, we can extend $\ca$ with suitable axionic currents.\footnote{There is also a possibility of introducing matter currents, cf. \cite{Costello:2023vyy}, which, however, is not the focus of this work.} %However, these axionic currents generally break higher-spin symmetry on the defect. These observation 
    These results are summarized in Theorem \ref{thm:quantum-associative} and Theorem \ref{thm:quantum-HS}. %-- the central results of this work.
    
    \item[-] Section \ref{sec:4} studies some simple correlation functions of the chiral algebra $\ca$, which can be identified with form factors %-- genuine scattering amplitudes in the rational sector (see Table \ref{tab:form factors}) -- 
    of certain $4d$ higher-spin theories in spacetime. We observe that there are some specific choices of the kinematic data that can lead to non-trivial higher-spin amplitudes on the celestial twistor sphere. Nonetheless, this is only well justified for theories with Yang-Mills-like interactions. This stems from the fact that the OPE of higher-spin currents, in many cases, encodes only the collinear limit of the higher-spin soft factors studied in \cite{Tran:2022amg}. Thus, adapting the approach of \cite{Costello:2022wso} to higher-derivative case will require further adjustments and clarifications.
    
    \item[-] Section \ref{sec:boson-fermi} proposes some simple chiral CFTs  on the celestial twistor sphere which can generate the chiral algebras studied in this paper. %Intriguingly, with ghosts included these CFTs have structures that closely resemble some worldsheet theories.

    \item[-] We wrap up the paper in Section \ref{sec:discuss} with some discussions. There are also three appendices that provide detailed computations related to Section \ref{sec:3} and Section \ref{sec:4}.
\end{itemize}

%%%%%%%%%%%%%%%%%%%%%%%%%%%%%%%%%%%%%%%%%%%%%%%%%%%%%%%%%%%%%%%%%%
\section{Review}\label{sec:2}
This section reviews some relevant material for the study of the chiral higher-spin algebra $\ca$ associated to various twistorial higher-spin theories in Section \ref{sec:3}. Note that we aim to be concise and refer the reader to \cite{Tran:2025uad} for more detail.

%%%%%%%%%%%%%%%%%%%%%%%%%%%%%%%%%%%%%%%%%%%%%%%%%
\subsection{Twistor and higher spins}\label{sec:twistor}
Denotes $\P^3$ as the 3-dimensional complex projective space with homogeneous coordinates 
\begin{align}
   Z^A=(Z^1,Z^2,Z^3,Z^4)=(\lambda^{\alpha},w^{\dot\alpha})\,,\qquad A=1,2,3,4\,,\quad \alpha=1,2\,,\quad \dot\alpha=\dot{1},\dot{2}\,.
\end{align}
Here, $(\lambda^{\alpha},w^{\dot\alpha})$ can be interpreted as the left- and right-handed commutative spinors of the Lorentz group $ SL(2,\C)\times SL(2,\C)$ when working in a complexified setting. (See e.g. \cite{Adamo:2017qyl} for an introduction to twistor theory.) The open subset $\PT:=\big\{Z^A\in \P^3\,\big|\,\lambda^{\alpha}\neq 0\big\}\subset \P^3$ where  $\lambda^{\alpha}$ is non-degenerate is referred to as \emph{undeformed} twistor space. There is a natural quaternionic conjugation, which acts on $Z^A$ as follows:
\begin{align}
    Z^A\mapsto \hat Z^A=(\hat\lambda^{\alpha},\hat w^{\dot\alpha})\,,\qquad \hat\lambda^{\alpha}=(-\overline{\lambda^2},\overline{\lambda^1})\,,\quad \hat w^{\dot\alpha}=(-\overline{w^{\dot 2}},\overline{w^{\dot 1}})\,.
\end{align}
Although working with $(\lambda,\hat\lambda)$ and $(w,\hat w)$ variables is sufficient to construct theories as well as their observables on twistor space, introducing a twistor correspondence between $\PT$ and the complexified Minkowski spacetime $\cM$ allows one to identify twistor cohomology classes with massless fields in spacetime \cite{Eastwood1979,Eastwood:1981jy,Mason:2007ct,Bullimore:2011ni,Bittleston:2020hfv}, thereby making twistor theory a viable physical framework. The twistor correspondence is expressed via the so-called incidence relations \cite{kodaira1962theorem,kodaira1963stability,Penrose:1967wn}:
\begin{align}\label{eq:incidence1}
    w^{\dot\alpha}=w^{\dot\alpha}(x,\lambda)\,, \qquad x^{\alpha\dot\alpha}:=x^{a}\sigma_{a}^{\alpha\dot\alpha}\,,\qquad a=1,2,3,4\,,
    \end{align}
where $\sigma_{a}^{\alpha\dot\alpha}$ are quaternions. Upon imposing these relations, all regular functions on twistor space can be referred to as elements of sheaves of local holomorphic functions $\cO(n)$ with weight $n\in \Z$. In studying twistor cohomology, cf. \cite{Eastwood1979,Eastwood:1981jy}, one finds that
\begin{align}
    \sA_{2h-2}\in H^{0,1}(\PT,\cO(2h-2))\leftrightarrow \{\text{$4d$ massless field with helicity $h\in \Z$}\}\,.
\end{align}
where $H^{0,1}$ denotes the Dolbeault cohomology group of $(0,1)$-form on $\PT$ twisted by the holomorphic line bundle $\cO(2h-2)$.\footnote{Note that the Dolbeault integrable complex structure on twistor space is given by \cite{Woodhouse:1985id}: 
\begin{align}
    \bar{\p}:=d\hat Z^A\frac{\p}{\p \hat{Z}^A}= d\hat\lambda^{\alpha}\hat{\p}_{\alpha}+ d\hat{w}^{\dot\alpha}\,\hat{\p}_{\dot\alpha}\,,\qquad \hat\p_{\alpha}=\frac{\p}{\p\hat\lambda^{\alpha}}\,,\quad \hat\p_{\dot\alpha}=\frac{\p}{\p\hat w^{\dot\alpha}}\,.
\end{align}}
Note that the $\lambda$-weight of a given twistor representative determines the helicity of the massless particle in spacetime. 

There is a neater way to organize the above result. It is well-known that $\P^3$ is a quotient of $S^7$ by a $U(1)$ factor \cite{Ward:1990vs}. Therefore, it is quite natural to construct twistor theories on the $U(1)$-bundle over $\PT$ since it allows us to handle all twistor expressions on an equal footing. In particular, we may associate a monodromy to each of twistor variable %as its loop around $U(1)\subset S^7$ fiber -- 
in terms of a phase $e^{i h\theta}$ with $h$ being the charge induced by the $U(1)$ fiber. Since twistor actions should have weight zero on $\PT$, or trivial $U(1)$ charge on $S^7$, they subject to trivial monodromy condition. This provides certain flexibility to construct a broad class of twistorial higher-spin actions, as shown in \cite{Tran:2025uad,Mason:2025pbz}.

To make the study of higher-spin chiral vertex algebra sufficiently general, we introduce the following holomorphic $\star$-product:
\begin{align}\label{eq:star-product}
    f_1(\lambda, w)\star f_2(\lambda,w)=\exp\Big([\p_1\,\p_2]\Big)f_1(\lambda,w_1)f_2(\lambda,w_2)\Big|_{
    w_{1,2}=w}\,.
\end{align}
Here, we adopt the convention:
\begin{align}\label{eq:convention}
 w^{\dot\alpha}=\epsilon^{\dot\alpha\dot\beta}w_{\dot\beta}\,,\quad w_{\dot\alpha}=w^{\dot\beta}\epsilon_{\dot\beta\dot\alpha}\,;\qquad  [\p_i\,\p_j]=\p_{i}^{\dot \alpha}\p_{j\dot \alpha}\,,\quad\p_{i\dot\alpha}\equiv\frac{\p}{\p w_i^{\dot\alpha}}\,.
\end{align}
where $\eps^{\dot\alpha\dot\beta}=\epsilon_{\dot\alpha\dot\beta}=-\epsilon^{\dot\beta\dot\alpha}$ with $\epsilon^{\dot1 \dot 2}=1$. (Similar expressions also apply to variables with undotted spinorial indices.) The higher-spin algebra $\mathrm{h}\hs$ associated with the $\star$-product \eqref{eq:star-product} is given by 
\begin{align}\label{eq:hhs}
    \mathrm{h}\hs:=\C[\lambda]\otimes \C\Bigg[\frac{\hat\lambda}{\langle \lambda\,\hat\lambda\rangle}\Bigg]\otimes\sA_1(w)\,,
\end{align}
where $\sA_1(w)$ denotes the Weyl algebra whose canonical pair are $(w^{\dot1},w^{\dot 2})$. As is well-known, the above h$\hs$ is a unique Moyal-Weyl deformation of the twistorial $ w_{1+\infty}$ algebra \cite{Adamo:2021lrv,Monteiro:2022xwq}. Note that there is also a possibility to tensor h$\hs$ with a %Lie algebra $\mg$, or a 
matrix algebra $\Mat(N,\C)$. This will be the non-commutative algebra, which our construction of the chiral higher-spin algebra $\ca$ on the celestial twistor sphere will largely be based on.

%%%%%%%%%%%%%%%%%%%%%%%%%%%%%%%%%%%%%%%%%%%%%%
\subsection{Holomorphic twistorial higher-spin theories}\label{sec:BV} 
The higher-spin symmetry introduced in previous subsection is known to govern various twistorial higher-spin theories in twistor space \cite{Tran:2025uad,Mason:2025pbz}. These theories can be nicely described through the scope of BV-BRST formalism for holomorphic theories \cite{williams2020renormalization,Costello:2021bah,Bittleston:2022nfr}. (See also \cite{Cattaneo:2023hxv} for a short summary and \cite{Henneaux:1992ig,Barnich:2018gdh,Cattaneo:2019jpn} for an introduction into this formalism.) 
%%%%%%%%%%%%%%%%%%%%%%%%%%%%%%%%%%%%%%%%%%%%%%%%
\subsubsection{Holomorphic Chern-Simons theories}
Let us first review the case of holomorphic Chern-Simons theories with the BV field
\begin{align}
    \A_{2h-2}\in \Omega^{0,\bullet}(\PT,\cO(2h-2)\otimes\mg)[1]\,,
\end{align}
whose fields components are:
\begin{subequations}
    \begin{align}
     \text{ghost}&:   &\cc_{2h-2}&\in \Omega^{0,0}(\PT,\cO(2h-2)\otimes\mg)[1]\,,\quad &|\cc|&=1\,,\\
      \text{field}&:  &\sA_{2h-2}&\in \Omega^{0,1}(\PT,\cO(2h-2)\otimes \mg)[1]\,,\quad &|\sA|&=0\,,\\
      \text{antifield}&:  &\sA^{\vee}_{2h-2}&\in \Omega^{0,2}(\PT,\cO(2h-2)\otimes \mg)[1]\,,\quad &|\sA^{\vee}|&=-1\,,\\
       \text{antifield of ghost}&: &\cc^{\vee}_{2h-2}&\in \Omega^{0,3}(\PT,\cO(2h-2)\otimes\mg)[1]\,,\quad &|\cc^{\vee}|&=-2\,.
    \end{align}
\end{subequations}
Here, $|x|\equiv \deg(x)$ refers to the cohomological degree (or ghost degree). For a given $(0,k)$-form field, its ghost degree is given by $1-k$, where the `1' comes from the degree-shift denoted by $[1]$. We will choose $\mg$ to be some Lie algebra such that it is isomorphic to its dual algebra $\mg^{\vee}$. The specific choice of $\mg$ will be important when we introduce the holomorphic Green-Schwarz anomaly cancellation mechanism in \cite{Costello:2021bah}. 

The BV twistor actions on the total space $S^7$, which lead to various holomorphic Chern-Simons theories on twistor space are (see also \cite{Mason:2025pbz})\footnote{as well as \cite{Neiman:2024vit} for another way of constructing twistor actions for self-dual higher-spin theories.}
\begin{subequations}\label{eq:parent-action-1}
    \begin{align}
    S_{BV_1}^{[,]_{\mg}}&=\int_{S^7}d\theta \,\Omega^{3,0} \,\Tr\Big(\A\wedge\bar{\p}\A+\frac{1}{3}\A\wedge [\A,\A] \Big)\,,\label{eq:BV-Lie}\\
    S_{BV_2}^{\{,\}}&=\int_{S^7}d\theta \,\Omega^{3,0} \,\Big(\A\wedge\bar{\p}\A+\frac{1}{3}\A\wedge\{\A, \A\} \Big)\,,\label{eq:BV-Poisson}\\
    S_{BV_3}^{\star}&=\int_{S^7}d\theta \,\Omega^{3,0} \,\Tr\Big(\A\star\bar{\p}\A+\frac{2}{3}\A\star \A \star \A \Big)\,,\label{eq:BV-star}
\end{align}
\end{subequations}
where
\begin{align}
    \Omega^{3,0}=e^{4i\theta}\langle\lambda \,d\lambda\rangle\wedge [d w\wedge dw]
\end{align}
is the canonical holomorphic measure on twistor space of charge $+4$, and 
\begin{align}\label{eq:A-generating-function}
    \A=\sum_{h\in \Z}e^{i\theta(2h-2)}\A_{2h-2}\,,\qquad \A_{2h-2}\in \Omega^{0,1}(\PT,\cO(2h-2)\otimes \mg)[1]\,,
\end{align}
are $\mg$-valued higher-spin generating connection one-forms. For $S_{BV_2}$, it is necessary that $\mg=\mathfrak{u}(1)$. In the above, $[-,-]$ denotes the usual Lie algebra commutator (with neutral $U(1)$ charge), $\{-,-\}$ denotes the Poisson bracket with $U(1)$ charge $-2$, and $\star$ is the Moyal-Weyl product defined in \eqref{eq:star-product}. Then, upon projecting to the trivial monodromy sectors, and pushing forward to the twistor base, one can obtain for instance
\begin{align}\label{eq:BV-action}
    S_{BV_3}^{\star}=\int_{\PTc}\Omega^{3,0}\Tr\Big(\sum_h\A_{-2|h|-2}\bar{\p}\A_{2|h|-2}+\frac{2}{3}\sum_{\{h_i\}}\frac{1}{k!}\A_{2h_1-2}\Pi^k(\A_{2h_2-2},\A_{2h_3-2})\Big)\,,
\end{align}
where $\Pi(-,-):=\epsilon^{\dot\alpha\dot\beta}\p_{\dot\beta}\wedge \p_{\dot\alpha}$. Note that the trivial monodromy constraint fixes the number of derivatives to be
\begin{align}\label{eq:k-constraint}
    k=h_1+h_2+h_3-1\geq 0\,.
\end{align}
This can be viewed as the helicity constraints for the vertices on the base manifold $\PTc$.\footnote{This is the twistor space associated with the deformed complex structure $\bar{\cD}:=\bar{\p}+\Pi(\sA_2,-)$ where $\sA_{2}\in \Omega^{0,2}(\PT,\cO(2))$. The spacetime dual of $\PTc$ is a self-dual spacetime by virtue of the non-linear graviton construction \cite{Penrose:1976js}.} We shall not delve on listing all the descendants of the parent actions \eqref{eq:parent-action-1}, and refer the reader to \cite{Tran:2025uad} for detail.

%%%%%%%%%%%%%%%%%%%%%%%%%%%%%%%%%%%%%%%%
\subsubsection{Holomorphic BF theories}
In the case of holomorphic BF theories, we have the following parent actions on $S^7$ \cite{Tran:2025uad}:
\begin{subequations}
    \begin{align}
    S_{BF_1}^{[,]}&=\int_{S^7}d\theta\,\Omega^{3,0}\Tr\Big(\B\wedge \bar{\p}\A+\frac{1}{2}\B[\A,\A]\Big)\,,\\
    S_{BF}^{\{,\}}&=\int_{S^7}d\theta\,\Omega^{3,0}\Big(\B\wedge \bar{\p}\A+\frac{1}{2}\B\{\A,\A\}\Big)\,,
\end{align}
\end{subequations}
where $\A$ stands for the master BV field that includes fields with non-negative $U(1)$ charges, while $\B$ contain those with strictly negative $U(1)$ charges. In particular,
\begin{subequations}\label{eq:BF-BV}
    \begin{align}
   \A&=\sum_{s\in \N}e^{i\theta(+2s-2)}\A_{+2s-2}\,,\qquad &\A_{+2s-2}&\in\Omega^{0,\bullet}(\PT,\cO(+2s-2)\otimes\mg)[1]\,,\\\B&=\sum_{s\in\N_0}e^{i\theta(-2s-2)}\B_{-2s-2}\,,\qquad &\B_{-2s-2}&\in\Omega^{0,\bullet}(\PT,\cO(-2s-2)\otimes\mg)[1]\,.
\end{align}
\end{subequations}
The field components of the above are organized as
\begin{subequations}
    \begin{align}
        \A_{2s-2}&=\cc_{2s-2}+\cA_{2s-2}+\cB^{\vee}_{2s-2}+\dd^{\vee}_{2s-2}\,,\\
        \B_{-2s-2}&=\dd_{-2s-2}+\cB_{-2s-2}+\cA^{\vee}_{-2s-2}+\cc^{\vee}_{-2s-2}\,,
    \end{align}
\end{subequations}
where
\begin{align}
    |\cc|=|\dd|=1\,,\qquad |\cA|=|\cB|=0\,,\qquad |\cA^{\vee}|=|\cB^{\vee}|=-1\,,\qquad |\cc^{\vee}|=|\dd^{\vee}|=-2\,.
\end{align}
Observe that we have two kind of ghost fields $\cc$ and $\dd$, which appear as the gauge parameters in the gauge transformations of the physical fields $(\cA,\cB)$:
\begin{subequations}
    \begin{align}
        S_{BF_1}&: &\delta_{\cc}\cA&=\bar{\p}\cc+[\cA,\cc]\,,\quad 
        &\delta_{\cc,\dd}\cB&=\bar{\p}\dd+[\cA,\dd]+[\cB,\cc]\,,\\
        S_{BF_2}&: &\delta_{\cc}\cA&=\bar{\p}\cc+\{\cA,\cc\}\,,\quad 
        &\delta_{\cc,\dd}\cB&=\bar{\p}\dd+\{\cA,\dd\}+\{\cB,\cc\}\,.
    \end{align}
\end{subequations}
This marks the difference with the holomorphic Chern-Simons theories discussed above. We refer the reader to \cite{Tran:2025uad} for the quantization of holomorphic BF and CS theories at one loop. (See also \cite{Bittleston:2022nfr} for an analysis of holomorphic BF-type theories with spin-$s\leq 2$.)

%%%%%%%%%%%%%%%%%%%%%%%%%%%%%%%%%%%%%%%%%%%%%%%%%
\subsection{Holomorphic Green-Schwarz anomaly cancellation mechanism}\label{sec:anomaly-cancellation}

It was shown in \cite{Tran:2025uad} that all holomorphic higher-spin Chern-Simons theories are anomaly-free at one loop. However, this does not always apply to holomorphic BF theories. Such obstruction for quantum consistency at one loop can be lifted by introducing appropriate couplings to the anomalous theories, cf. \cite{Costello:2015xsa,Costello:2021bah,Bittleston:2022nfr}.

\paragraph{Anomaly-free theories.} Let us first discuss the anomaly-free theories. In \cite{Tran:2025uad}, it was shown that the gauge anomalies associated to the wheel diagram of various holomorphic twistorial theories have the following simple form
\small
\begin{align}\label{eq:anomaly-term}
    \parbox{45pt}{\includegraphics[scale=0.1]{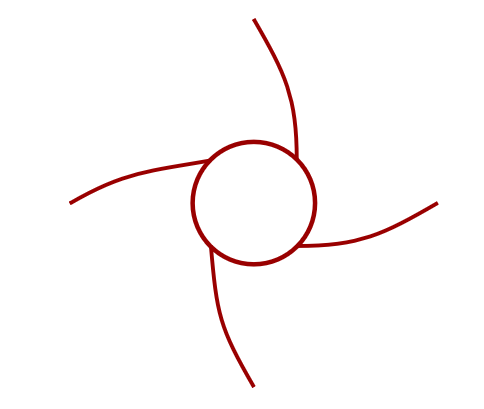}}=\sum_{h\in\Spec}\frac{i^3}{3!(2\pi)^3}\int_{\PT}\mho\,\Tr\Big(\cc_{h_1}\wedge(\p \sA_{h_2}(z_2))\wedge (\p\sA_{h_3}(z_3))\wedge(\p\sA_{h_4}(z_4))\Big)\Big|_{z_i=z}\,,
\end{align}
\normalsize
where 
\begin{align}
    \p=dz^a\frac{\p
    }{\p z^a}\,,\qquad z^a=(z,w^{\dot 1},w^{\dot 2})\,,\qquad a=1,2,3\,,
\end{align}
is the holomorphic differential. In the above, $h_i$ denote the helicities of the external fields entering the wheel diagram. Moreover, $\mho$ is a differential operator, which reads
\begin{align}
    \mho=\frac{\big([\p_2\,\p_3]+[\p_2\,\p_4]+[\p_3\,\p_4]\big)^{\tH_4-4}}{(\tH_4-4)!}\,,\qquad \tH_4=h_1+h_2+h_3+h_4\,,\quad \p_i=\frac{\p}{\p w_i^{\dot\alpha}}\,.
\end{align}
Here, $\p_i$ acts on $\sA_{h_i}$, and all derivatives are understood to be evaluated at some point $z\in\PTc$. Note that in evaluating the above anomalies, we have performed our computation in the patch $\C^3\subset\PTc$ parametrized by in-homogeneous coordinates $z^a$. (The easiest way to reach this coordinate patch is to parametrize $\lambda^{\alpha}=(1,z)$ with $z\in \C$.)

Since the integral \eqref{eq:anomaly-term} is finite, we are left with a sum over the spectrum that requires regularization. Although there is not yet a concrete proposal, we expect that higher-spin theories should be regarded as some string theories. It is therefore natural to employ some string-inspired regularization, when summing over Kaluza-Klein modes, as well as higher-spin modes. One such powerful scheme is the Riemann zeta-function regularization, see e.g. \cite{Beccaria:2015vaa,Beccaria:2023ujc} for the justification of such regularization. Using the fact that $\zeta(s)=\sum_{n=1}^{\infty}\frac{1}{n^s}$ in the $s\rightarrow 0$ limit, we find for instance
\begin{align}
    2\sum_{|h|\geq 1}1=-1\,,\qquad 2\!\!\!\!\sum_{|h|\in 2\N_0+1}\!\!\!\!\!\!\!=0\,,\qquad 2\sum_{|h|\geq 2}1=-3\,,\qquad 1+2\!\!\!\!\sum_{|h|\in 2\N}\!\!\!1=0\,.
\end{align} 
Therefore, only twistorial theories with 
\begin{align}\label{eq:quantum-protected-spectrum}
    \Spec=\Z\,,\  \text{or}\  2\Z\,,\ \text{or}\ 2\Z+1\,,
\end{align}
can be anomaly-free. 

\paragraph{Green-Schwarz anomaly cancellation on twistor space.} As is well known, in curing a theory with gauge anomaly, it typically requires extending the field content of the original theory, often leading to a larger theory where the anomaly can be cancelled off-shell. This, however, does not apply to the Green-Schwarz anomaly cancellation on twistor space proposed in \cite{Costello:2021bah} for non-supersymmetric theories, where the anomaly is cancelled only on-shell. This stems from the fact that the twistor dual of the spacetime axion field is a gauge field on $\PT$; necessitates $\bar{\p}\sA\approx 0$ for the anomaly cancellation on twistor space to be possible \cite{Costello:2021bah}. Note that beside the axion, one can also introduce suitable fermionic matter fields to cancel the anomaly occurs on twistor space as in \cite{Costello:2021bah,Bittleston:2022jeq}. 

In contrast with the common view point where gauge anomaly is fatal, the twistor gauge anomaly is somewhat interesting since it implies that the corresponding spacetime theory will also be ``anomalous'' in the sense of Bardeen, cf. \cite{Bardeen:1995gk} -- i.e. it will have non-trivial scattering amplitudes. Note that the axionic currents associated with the axionic field in twistor space are the key factors in rendering the chiral higher-spin symmetry algebras of anomalous higher-spin theories associative to first order in quantum correction via Koszul duality \cite{Costello:2022wso}. (This will be one of the subjects of Section \ref{sec:3}.)

Let us return to the anomaly \eqref{eq:anomaly-term} and complete our discussion. By doing integration by part, we can write
\begin{align}
    [\p_2\,\p_3]+[\p_2\,\p_4]+[\p_3\,\p_4]=[\p_2\,\p_3]+[\p_4\,\p_1]\,.
\end{align}
Then, employing Okubo's relations \cite{Okubo:1978qe}:
\begin{align}\label{eq:Okubo}
    \Tr(T^{(a_1}T^{a_2}T^{a_3} T^{a_4)})=C_{\mg}\tr(T^{(a_1}T^{a_2})\tr(T^{a_3}T^{a_4)})\,, \qquad C_{\mg}=\frac{10 h^{\vee}}{2+\dim(\mg)}\,,
\end{align}
with $\Tr$ the trace in the adjoint, $\tr$ the trace in the fundamental representations, and $h^{\vee}$ is the Coxeter number of the Lie algebra associated to either $SU(2),SU(3),SO(8)$ or $E_{6,7,8}$, we propose an anomaly cancellation for holomorphic theories with higher-derivative interactions by considering the following on-shell quantum corrected action, cf. \cite{Tran:2025uad},
\begin{align}\label{eq:axion-BVaction}
    S^{cor}_{\text{HS-BF}_{\mg}}=\int_{\PT}\p^{-1}\vartheta\bar{\p}\vartheta+c_{\mg}\int_{S^7}\vartheta \tr(\sA\star \p\sA)\,, \quad \vartheta\in \Omega^{2,1}(\PT,\cO(0))\,.
\end{align}
Here $\p^{-1}:\Omega^{p,\bullet}(\PT)\rightarrow \Omega^{p-1,\bullet}(\PT)$ is the formal inversion of the holomorphic differential $\p:=dz^a\p_a$, and $\vartheta^{2,1}$ is an axion field subjected to the constraint $\p\vartheta=0$. 
Note that $\vartheta$ transforms as $\delta \vartheta=\bar{\p}\varpi^{2,0}$, and its propagator $P_{\vartheta}$ is formally a $(4,2)$-form obeying 
\begin{align}\label{eq:pinch-axion}
    \bar{\p}P_{\vartheta}(z,z')=-\p\delta^{3,3}(z-z')\,,
\end{align}
where $\delta^{3,3}(z-z')$ is a $(3,3)$-form delta distribution. 

To see how the anomaly cancellation works, one can compute 4-pt tree-level amplitude with $\vartheta$ in the exchange, whose integration domain is $\PT\times_{\cM}\PT$ -- here, $\times_{\cM}$ denoted the fiberwise product over the same spacetime point.\footnote{It is useful to note that locally $\PT\simeq \P^1\times \cM$. Moreover, the realization of the curved twistor space $\PTc$ is slightly more complicated due to the fact that higher-derivative interactions can also deform the $\P^1$-fiber. Nevertheless, the analysis stays the same since we are working mainly on twistor space.} Although there may be higher-derivative terms in the vertices of the tree-level amplitudes, the pushforward to the twistor space will select for us the right couplings, which belong to the trivial monodromy sector.  For instance, we can recast \eqref{eq:anomaly-term} as \cite{Tran:2025uad}
\begin{align}\label{eq:anomaly-term-1}
    \eqref{eq:anomaly-term}=\sum_{\Spec}\frac{i^3C_{\mg}}{3!(2\pi)^3}\int_{S^7}\tr\big(\cc\wedge \p\sA\big)\tr\big(\p\sA\wedge\p\sA\big)\,,
\end{align}
for theories with gauge interactions with $\tH_4=4$. Then, upon considering a tree-level diagram whose gauge variation yields precisely \eqref{eq:anomaly-term-1}, but with an opposite sign,\footnote{Note that this is an on-shell statement, as stated above.} the cancellation of the gauge anomaly associated with the holomorphic higher-spin BF theory with $\Spec=\{|h|>1\}$ allows us to fix \cite{Tran:2025uad}
\begin{align}\label{eq:axion-coupling-constant}
    c_{\mg}=\sqrt{\frac{-iC_{\mg}}{3!(2\pi)^3}}\,,
\end{align}
to be the coupling constants entering the tree-level $\vartheta$-exchanged diagrams.

%(Note that the word ``\emph{anomalous}'' merely refers to theories with non-trivial amplitudes \cite{Bardeen:1995gk}, and does not imply those theories are pathological.) See also \cite{Monteiro:2022xwq,Monteiro:2022nqt, Doran:2023cmj}.

%%%%%%%%%%%%%%%%%%%%%%%%%%%%%%%%%%%%%%%%%%%%%%%%%%%%%%%%%%%%%%%%
\section{Chiral higher-spin algebras of the celestial twistor sphere}\label{sec:3}

As in usual QFT context, a gauge field can be coupled to a current sourced by suitable matter fields. A similar situation also occur in the context of twisted holography \cite{Witten:2003nn,Costello:2020jbh}, where the gauge fields $\A$ and the currents $\tJ$ do not need to live in the same space. In what follows, we consider a bulk/defect system, where twistor space is the bulk and any pointed algebraic curve $\P^1_{\tp}$ over a spacetime point $\tp\in \cM$ will be viewed as a complex co-dimension 2 defect in $\PT$. The gauge field $\A$ will live in the bulk while $\tJ$ will be defined on the defect. Then, to construct chiral higher-spin algebras $\ca$, whose associative product are defined by the OPE structures of holomorphic higher-spin currents on $\P^1_{\tp}$, we can employ Koszul duality -- a symmetry-preserving map, which transfers associative structures of h$\hs$ to $\ca$ as in \cite{Costello:2022wso}.

We will show that the chiral higher-spin algebras in consideration are non-unitary $W_{1+\infty}[\mg]$-algebras. Namely, their generators can have negative conformal weights and take values in the same Lie algebra $\mg$ as the bulk twistor fields.\footnote{Here, the non-unitarity of $\ca$ can also understood from the fact that the dual bulk twistor theories are non-unitary theories. Note that we choose to work in the helicity basis to simplify our analysis.} (For previous work related to unitary $W_{1+\infty}$ and $W_{1+\infty}[\mg]$ algebras, see e.g. \cite{Pope:1989sr,Pope:1989ew,Pope:1991ig} and \cite{Odake:1990rr,Bakas:1991fs}.)

%%%%%%%%%%%%%%%%%%%%%%%%%%%%%%%%%%%%%%%%%%%%%%%%%%%%
\subsection{Kozsul duality and chiral higher-spin algebra}\label{sec:Koszul-duality}

As stated, we want to induce the underlying symmetry h$\hs$ in the bulk, i.e. twistor space, onto the defect $\P^1_{\tp}$ in terms of chiral symmetry algebras $\ca$ by constructing the OPEs of some holomorphic higher-spin currents $\tJ[\Delta,\tH]$.\footnote{This is regarded as defect construction, which is opposite with the usual bulk construction in holography. In particular, in the usual bulk reconstruction procedure, one often starts with a global symmetry and try to gauge it as the local gauge symmetry in the bulk. However, this does not always guarantee the existence of a bulk theory, especially when the symmetry is intricate, see e.g. \cite{Sharapov:2019vyd}.} As a result, there should be a symmetry preserving map, which transfer the associative structures of h$\hs$ to $\ca$. This map is known as Koszul homomorphism or Koszul duality, cf. \cite{Costello:2022wso}. %Let us now introduce some terminology that will be needed in the following. Note that we will refer back to the material in Subsection \ref{sec:BV} whenever appropriate.

\paragraph{Koszul duality.} Formally, Koszul duality (see e.g. Chapter 3 in \cite{loday2012algebraic})
\begin{align}\label{eq:Koszul-duality}
    S(V)\cong \bigwedge(V^{\vee})^!
\end{align}
is a duality between the derived category of an exterior algebra $\bigwedge(V^{\vee})=\bigoplus_{k=0}^{\dim(V)}\wedge^k(V^{\vee})$ and that of a symmetric algebra $S(V)=\bigoplus_{k=0}^{\infty}\Sym^k(V)$ with $V$ being some vector space of dimension $\dim(V)$. 

In our context, the exterior algebra above is a non-commutative algebra denoted as $(\Coh,\star)$ and $S(V)$ corresponds to the chiral higher-spin algebra $\ca$ that we aim to construct. Then, the Koszul duality in our setting is the identification \cite{Costello:2022wso}:
\begin{align}\label{eq:Kozsul1}
    \ca \cong \big(\Coh,\star\big)^{!}\,,\qquad \Coh:=\left\{\Phi\in\cF_{BV}\,\big|\,\deg(\Phi)=0\right\}\,,
\end{align}
where the exclamation mark ! denotes the Koszul dual operation of the pair $(\Coh,\star)$. Here, $\cF_{BV}$ is the space of BV fields. 

Note that the projection $\cF_{BV}\rightarrow \Coh$ 
is an augmentation map, which projects $\cF_{BV}$ to a subspace, which contains only elements of cohomological degree zero, i.e. the space of physical fields $(\sA,\vartheta)$. This map certainly preserves the associativity of the $\star$-product and is compatible with the degree one nilpotent cohomological vector field
\begin{align}
    Q:=(S_{BV},-)_{BV}\,,
\end{align}
where $(-,-)$ is the graded Poisson bracket on $\cF_{BV}$ induced by the symplectic form $\omega_{BV}$ of degree $-1$ (see e.g. Section 3.2 in \cite{Tran:2025uad} for a quick recap). 
%The Koszul dual algebra of $\Coh^{\star}=(\Coh,\star)$ is then given by the extension algebra
%\begin{align}
%    \cA^!=Ext^{\bullet}_{\Coh^{\star}}(\C,\C)=\,,
%\end{align}

Although one may expect that as $\Coh^{\star}$ is associative, $\ca$ will also be associative by virtue of Koszul duality, this expectation does not always hold if quantum corrections are taken into account. In some cases, $\ca$ must be extended by introducing suitable axionic currents so that associativity can be restored at quantum level, cf. \cite{Costello:2022upu}.

%%%%%%%%%%%%%%%%%%%%%%%%%%%%%%%%%%%%%
\paragraph{Chiral algebra.} To construct $\ca$ explicitly, we shall implement \eqref{eq:Koszul-duality} by introducing the couplings \cite{Witten:2003nn}
 \begin{align}
   S_{\tJ}= \int_{\P^1} e^0 \tJ[\Delta;\tH] \tA_{2h-2}\,,\qquad  e^0=\langle\lambda\, d\lambda\rangle
\end{align}
where $\tJ$ denotes a holomorphic current, which couples to a physical bulk field $\tA$ near the defect. Here, $\tA$ is a restriction to the $\P^1_{\tp}$ fiber of the bulk field $\sA\in \Omega^{0,1}(\PTc,\mg)$. 
\begin{figure}[ht!]
    \centering
    \includegraphics[scale=0.36]{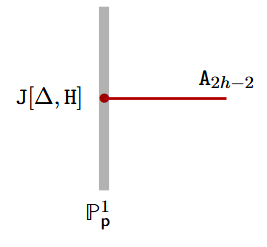}
    \caption{The coupling between a bulk field $\tA$ near a defect $\P^1_{\tp}$ and a current $\tJ[\Delta;\tH]$.}
    \label{fig:Koszul-coupling}
\end{figure}

Our task from now on is to construct the OPE between higher-spin currents $\tJ$ by imposing gauge invariance on the partition function
\begin{align}\label{eq:partition-function}
    \cZ=\int D\A D\phi D\psi \exp\Big(S_{BV}+ \int_{\P^1}e^0\tJ(\phi,\psi)\,\tA+ S[\phi,\psi]\Big)\,,
\end{align}
of the bulk/defect system above.
Here, we assume that $\tJ$ can be constructed from some chiral matter fields $\phi,\psi$ on the defect. (This topic will be discussed in Section \ref{sec:boson-fermi}.) 

Let us now expand the exponential with the source term in \eqref{eq:partition-function} in radial ordering and
consider the BRST (or gauge) transformation\footnote{We will perform most of the computation with the $\star$-product. However, it is also possible to consider the $\{\,,\}$ bracket as well as the usual Lie algebra bracket. These cases will be discussed whenever it is appropriate.}
\begin{align}
    Q_{\star}\big|_{\P^1}\tA=\bar{\p}\big|_{\P^1}\cc+[\tA,c]_{\star}\,,\qquad \bar{\p}\big|_{\P^1}=\bar{e}^0\bar{\p}_0\,,\qquad \bar{\p}_0=\langle\lambda\,\hat\lambda\rangle\lambda_{\alpha}\frac{\p}{\p \hat\lambda_{\alpha}}\,.
\end{align}
Recall that $\tA_{2h-2}$ has weight $2h-2$. Thus, near the defect, we may consider the following `plane-wave' representation of the bulk field
\begin{align}\label{eq:plane-wave-1}
    \tA:=\sA_{0} \langle \hat\lambda \,d\hat\lambda\rangle\,,\qquad  \sA_0=\sA^a_{2h-2}(\lambda,\hat\lambda)T_a\,e^{-[w\,\tilde v]}\,.%\qquad \bar{e}^0:=\frac{\langle \hat\lambda d\hat\lambda\rangle}{\langle \lambda\,\hat\lambda\rangle^2}\in \Omega^{0,1}(\P^1_{\tp},\cO(-2))\,.
\end{align}
%Furthermore $\tilde v^{\dot\alpha}$ can be viewed as 
%where 
%\begin{align}\label{eq:plane-wave-1}
%    \sA_0=\sA^a_{2h}(\lambda,\hat\lambda)T_a\,e^{-[w\,\tilde v]}\,.
%\end{align}
Here, the spinors $\tilde v$ are external data associated with $\tA$, and $T_a$ are generators of some Lie algebra $\mg$ induced by a stack of $N$ space-filling D${}_5$ branes wrapping $\PT$.\footnote{The above plane-wave basis is inspired by the half-Fourier or Penrose transform, cf. \cite{Witten:2003nn,Cachazo:2004kj}, of the null momentum $p_i^{\alpha\dot\alpha}=\lambda_i^{\alpha}\tilde\lambda_i^{\dot\alpha}$, from momentum space to twistor space. After the Penrose transform, the left-handed spinors $\lambda$ can be viewed as coordinates on the pointed algebraic curves $\P^1_{\tp}$, cf. \cite{Nair:1988bq}. Meanwhile, the right-handed spinor $\tilde \lambda_i$ of $SU(2)_-\subset SL(2,\C)$ can be identified with the spinors $\tilde v_i$ in \eqref{eq:plane-wave-1}. Note that, in a complexified setting, $\tilde v_i$ does not depend on $\lambda_i$. } 

Assuming the higher-spin currents are holomorphic, we perform an integration by part and pick up a boundary term that forces two nearby points to coincide in radial direction (see the computation along the line of \eqref{eq:tree-OPE-2}). The result of the BRST variation is
\begin{align}\label{eq:BRST-variation}
    \int_{\P^1}e^0\int_{\P^1}e^0\tA^a \bar\p\cc^b[\tJ_a,\tJ_b]=\int_{\P^1}e^0[\tA,\bar\p\cc]_{\star}^c\tJ_c=\sum_p\int_{\P^1}e^0\tg_p^{abc}\frac{[\tilde v_{\tA}\,\tilde v_{\cc}]^p}{p!}\tA_{a}\bar\p\cc_b\tJ_c\,,
\end{align}
where $\tg^{abc}_p$ represents a structure constant that depends on the number of $\msu(2)$-contractions between the $\tilde v$ spinors associated to the bulk fields $\tA$ and $\cc$, respectively. In particular,
\begin{subequations}
    \begin{align}
    \tg_{p\in 2\N_0}^{abc}&=f^{abc}\,,\quad &f^{abc}&=\Tr(\tT^a[\tT^b,\tT^c])\,,\\
    \tg_{p\in 2\N_0+1}^{abc}&=d^{abc}\,,\qquad &d^{abc}&=\Tr(\tT^a\{\tT^b,\tT^c\})\,.
\end{align}
\end{subequations}
\normalsize
We obtain the following simple relations
\begin{align}\label{eq:ca-test-1}
    [\tJ^a,\tJ^b]=\sum_p\frac{\tg^{abc}_p}{\langle\lambda_{\tA}\,\lambda_{\cc}\rangle}\frac{[\tilde v_{\tA}\,\tilde v_{\cc}]^p}{p!}\tJ_c\,.
\end{align}
Thus, the operator product expansion (OPE) of the higher-spin currents $\tJ$ indeed encodes the information of the $\star$-product as expected. %\footnote{Note that we can replace the $\star$-product with the Lie bracket $[\,,]$ or the Poisson bracket $\{\,,\}$, and all arguments above remain valid. Here, we aim to keep the discussion as general as possible.} 
Note that even though the above relations are well-defined and can be checked to be associative, there remain several issues. Namely, we do not know the conformal helicity weights nor the $\msu(2)$-charge of the higher-spin currents $\tJ$. To resolve this situation in a manifestly Lorentz covariant way,  %for the chiral algebra $\ca$ associated to holomorphic twistorial theories, 
we will unfold the coupling $S_{\tJ}$ as:
\begin{align}
    S_{\tJ}=\sum_{k}\frac{1}{k!}\int_{\P^1_{\tp}} e^0\tJ_{\dot\alpha(k)}\p^{\dot\alpha(k)}\tA_{2h-2}=\sum_k\frac{1}{k!}\int_{\P^1}e^0\tilde v^{\dot\alpha(k)}\tJ_{\dot\alpha(k)}\tA_{2h-2}\,,
\end{align}
where it is convenient to condense our notation as
\begin{align}
    \p^{\dot\alpha(k)}\equiv\p^{\dot\alpha_1}\ldots \p^{\dot\alpha_k}\,,\qquad \tilde v^{\dot\alpha(k)}\equiv v^{\dot\alpha_1}\ldots v^{\dot\alpha_k}\,.
\end{align}
Here, $\tJ_{\dot\alpha(k)}$ is a rank-$k$ symmetric higher-spin current valued in the $k$th jet $j^{k}(\tA_{2h-2})$ of $\tA_{2h-2}$. Upon unpacking $\tilde v^{\dot\alpha}=(\tilde v^{\dot 1},\tilde v^{\dot 2})$, we reproduce the coupling found in e.g. \cite{Costello:2022wso}. Namely,
\begin{align}
    S_{\tJ}=\sum_{m+n=k}\int_{\P^1}\frac{(\tilde v^{\dot 1})^m(\tilde v^{\dot 2})^n}{m!n!}\tJ_{\dot 1(m)\,\dot 2(n)}\tA_{2h-2}\equiv \sum_{m+n=k}\int_{\P^1}\frac{(\tilde v^{\dot 1})^m(\tilde v^{\dot 2})^n}{m!n!}\underline{J}[m,n]\tA_{2h-2}\,.
\end{align}
Let us now introduce two quantum numbers to properly define the holomorphic currents\footnote{Our CFT data differ from those in, e.g. \cite{Costello:2022upu,Bittleston:2023bzp}.}
\begin{align}\label{eq:J-def-1}
    \tJ[\Delta;\tH]:=\sum_{k\geq 0}\frac{\tilde v^{\dot\alpha(k)}}{k!}\tJ_{\dot\alpha(k)}\,.
\end{align}
In particular, we denote
\begin{itemize}
    \item[1-] $\Delta=h$, with $h$ being the \emph{helicity} of the higher-spin field $\tA_{2h-2}$ that $\tJ[\Delta;\tH]$ is Koszul dual to, as the conformal \emph{helicity} weight;\footnote{Here, we could use the word `spin' instead of `helicity' as in the literature cf. \cite{Pope:1991ig}. However, we find it more appropriate to use the word helicity since the conformal weight $\Delta$ can be negative, leading to a non-unitary CFT in this context.} 
    \item[2-] $\tH=k$, which is the number of external spinors $\tilde v$ to which the rank-$k$ symmetric tensors $\tJ_{\dot\alpha(k)}$ are contracted with, as the $SU(2)_-$ charge. 
\end{itemize}
Note that %he quantum numbers $(h;k)$ can be used to define some sort of chiral higher-spin modules $\cV(h;k)$. However, 
unlike the standard CFT (see e.g. \cite{DiFrancesco:1997nk}), $\Delta$ can be negative, %, reflecting the fact that our twistorial theories are non-unitary. 
similar to the case of the putative celestial CFT, cf. \cite{Strominger:2017zoo}.

In what follows, we will stay in the patch $\C^3\subset \PT$ with the in-homogenous coordinates:
\begin{align}
    \lambda_i^{\alpha}=(1,z_i)\,,\quad \hat\lambda^{\alpha}_i=(-\bar{z}_i,1) \,,\quad w^{\dot\alpha}=(w^{\dot1},w^{\dot 2})\,,\qquad \tilde v^{\dot\alpha}=(\tilde v^{\dot 1},\tilde v^{\dot 2})\,.
\end{align}
This is the patch where we can identify $\P^1_{\tp}$ with the celestial twistor sphere. 
Remarkably, in this patch, many expressions will receive great simplification. For instance,
\begin{align}
   \langle i\,j\rangle\equiv  \epsilon^{\beta\alpha}\lambda_{i\alpha}\lambda_{j\beta}=z_i-z_j\equiv z_{ij}\,,\quad [1\,2]=[\tilde v_1\,\tilde v_2]\,;\qquad e^0=dz\,,\quad \bar{\p}\big|_{\P^1}=d\bar{z}\p_{\bar{z}}\,.
\end{align}
Assuming $\P^1_{\tp}$ is extended enough so that it can wrap around $\C$. We can then view $\tJ[h_i;k_i](z_i)\in \C[z,z^{-1}]$ where
\begin{align}\label{eq:tJ-def}
    \tJ[h;k](z)%\sum_{n\in \Z} \frac{J_n[-k]}{z^{n+1+h}}
    =\sum_{n\in \Z} \frac{J_{n}[h;k]}{z^{n+1}}=\sum_{n\in\Z}\frac{J_n[k]}{z^{n+h+1}}\,,
\end{align}
as elements of $\cF(U_i)$ -- the space of local operators assigned to each open subset $U_i\subset \C$ around the point $z_i\in \C^{\times}$. We will require $\cF(U_i)$ to fulfill the factorization condition 
\begin{align}
    \cF(U_1 \sqcup\ldots \sqcup U_n)\cong \cF(U_1)\otimes \ldots \otimes \cF(U_n),
\end{align}
 where $\sqcup$ denotes the disjoint union. This allows us to view chiral algebras as factorization algebras \cite{beilinson2004chiral,frenkel2004vertex,costello2021factorization}. Then, for two points $z_1,z_2\in \C$, the fusion of 
 \begin{align}
     \tJ_1(z_1)\circ \tJ_2(z_2)\rightarrow V(\tJ_1,\tJ_2)\,,
 \end{align}
where $V(\tJ_1,\tJ_2)$ denotes the fusion vertex, should admit an asymptotic expansion in $\frac{1}{z_{12}}$ and satisfy the locality condition
 \begin{align}
     (z_1-z_2)^n\tJ_1(z_1)\circ \tJ_2(z_2)= (z_2-z_1)^n \tJ_2(z_2)\circ \tJ_1(z_1)\,,
 \end{align}
 for some number $n\in \N_0$. 
 
 In what follows, we will condense our notations as
\begin{subequations}
    \begin{align}
    \text{currents}&: 
    &\tJ[i]&\equiv \tJ[h_i;k_i](z_i)\,,\\ 
    \text{generators}&: &\cJ_{m}[i]&\equiv \cJ_{m}[h_i;k_i]\,,
\end{align}
\end{subequations}
where the generator $\cJ_{m}[i]$ associated to the current $\tJ[i]$ is given by%\footnote{Note that  half-interger $m$ is also possible.}
\begin{align}
    \cJ_{m}[i]:=\Res\,\tJ[i]z_i^m =\oint dz_i \sum_{n\in \Z} \frac{J_n[h_i;k_i]}{z_i^{n+1-m}}\,.
\end{align}
Note that we will suppress the factors $\frac{1}{2\pi i}$ when writing contour integrals. 
%Note that $\tJ$ can also be Lie algebra valued, i.e. $\tJ:=\tJ_aT^a$ for $T^a$ the being the generators in some representations of the Lie algebra $\mg$. 
Then, to consistently construct the OPE of $\tJ$s, we will impose the associativity conditions
\begin{align}\label{eq:to-be-checked}
  &\oint_{|w|=2} dw\, w^n\tJ^{a_1}[h_1](w)\oint_{|z|=1}\tJ^{a_2}\Big[s_2;2s_2-1\Big](0)\tJ^{a_3}\Big[s_3;2s_3-1\Big](z)\nn\\
    &=+\oint_{|z|=2} dz \tJ^{a_3}\Big[s_3;2s_3-1\Big](z)\oint_{|w|=1}\tJ^{a_2}\Big[s_2;2s_2-1\Big](0)\tJ^{a_1}[h_1](w)w^n\nn\\
    &\ \ \ \,+\oint_{|z|=2} dz \tJ^{a_2}\Big[s_2;2s_2-1\Big](0)\oint_{|z-w|=1}\tJ^{a_3}\Big[s_3;2s_3-1\Big](z)\tJ^{a_1}[h_1](w)w^n\,,
\end{align}
where $n\in \N_0$. Here, the power of $n$ are determined by the pole structures in the OPEs of the current currents. %As argued in \cite{Fernandez:2024qnu}, the associativity of $\ca$ can be used to determine all OPE's coefficients of the higher-spin currents. 
When combining these conditions with gauge invariance of \eqref{eq:partition-function}, we have a well-posed framework for efficiently fixing all couplings and OPE data.

%%%%%%%%%%%%%%%%%%%%%%%%%%%%%%%%%%%%%%%
%\subsection{Associativity and gauge invariance}

%%%%%%%%%%%%%%%%%%%%%%%%%%%%%%%%%%%%%%%%%%%%%%%%%%%
\subsection{Higher-spin current OPEs}\label{sec:algorithm}
%After providing somewhat a formal definition of Koszul duality and chiral algebra $\ca$, 
%Let us now construct the OPE between higher-spin currents $\tJ[\Delta, \tH]$ by imposing gauge invariance on the partition function \eqref{eq:partition-function} to fix the OPE coefficients order by order in perturbation theory. What we mean by this is the following.

% where we will replace the boundary of AdS by a defect $\P^1_{\tp}$ over a spacetime point $\tp$, as explained above. Then, upon imposing gauge invariance on the scattering amplitudes between bulk fields and the defect, we can extract non-trivial constraints, which fix almost all OPE coefficients of the higher-spin currents that are Koszul dual to higher-spin bulk fields. 
For our purposes of fixing OPEs resulting from the fusions of operators on the defect, we will consider processes in which two bulk fields $\tA$s interact with a defect by either couple directly to currents, or interact with themselves in the bulk first then couple to currents on the defect. This can be depicted as
\begin{align}\label{eq:diagrams}
    \parbox{65pt}{\includegraphics[scale=0.19]{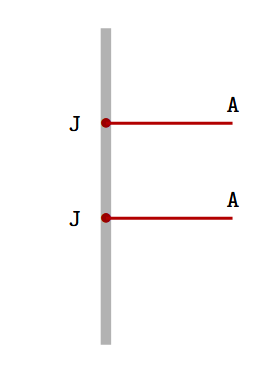}}\rightarrow &\quad \parbox{55pt}{\includegraphics[scale=0.19]{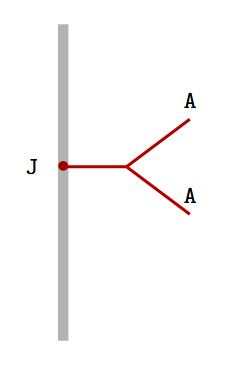}}+\quad \parbox{53pt}{\includegraphics[scale=0.19]{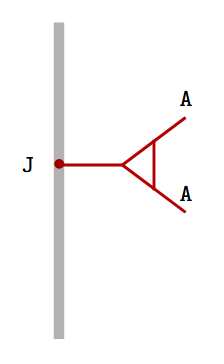}}+\quad (\text{higher bulk-loops})\nn\\
    &+\parbox{55pt}{\includegraphics[scale=0.19]{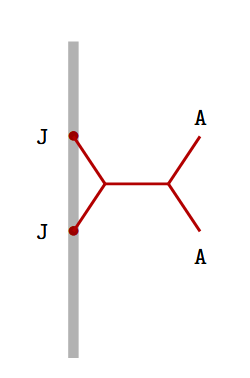}}+\quad \parbox{55pt}{\includegraphics[scale=0.19]{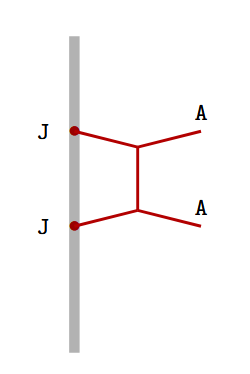}}+\quad \parbox{55pt}{\includegraphics[scale=0.19]{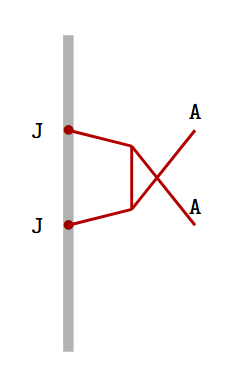}}+\quad (\text{higher bulk-loops})\nn\\
    &+(\text{higher numbers of $\tJ$s on the defect and higher bulk-loops})\,.
\end{align}
%where the first diagram on the rhs. of the arrow is a tree-process for OPE and after that are loop diagrams for quantum corrections.
%Notice that the above is a $2\rightarrow n$ scattering process, where two currents ``sourced by'' two bulk fields $\tA$ interact (the lhs. of \eqref{eq:diagrams}) and fragment into multiple operators on the defect [\textcolor{black!50}{\bf gray}] (the rhs. of \eqref{eq:diagrams}). %Note that since the bulk theory, cf. \eqref{eq:BV-action}, contains only cubic interactions, the number of diagrams one can draw at a given order can be (very) large but finite. 

 %The Koszul duality for holomorphic twistorial theories can provide strong constraints on the possible terms one can write down. 
The algorithm by which one can construct OPEs for $\tJ[i]\tJ[j]$ via Koszul duality is as follows:
\begin{itemize}
    \item[1-] Fix the number of bulk fields near the defect. Here, it is two.
    \item[2-] Draw all possible bulk/defect diagrams, where the number of currents/fields on the $1d$ holomorphic curve $\P^1_{\tp}$ can be more than two. %Note that we expect $\P^1_{\tp}$ to be extended enough so that it can wrap around $\C_{\tp}$. %with coordinates $z$. 
    \item[3-] At each order in perturbation theory, one enforces bulk/defect amplitudes to be gauge invariant on-shell. Namely, the BRST variation of the amplitudes on the lhs. of the arrow in \eqref{eq:diagrams}, should cancel out with the BRST variation of the amplitudes on the rhs. order by order in perturbation theory\footnote{similarly with the anomaly-inflow mechanism}. This generally imposes non-trivial constraints, allowing one to fix a large number of OPE coefficients. 
    \item[4-] For OPE coefficients that are unconstrained by gauge invariance of the bulk/defect amplitudes, we can fix them uniquely by imposing \eqref{eq:to-be-checked} with suitable power of $n$. 
    
\end{itemize}
%The above procedure may be regarded as the boundary-construction, which is the inverse direction of the usual bulk-construction in higher-spin literature, see e.g. \cite{Sharapov:2019vyd}.
Note that while the above procedure can determine the OPE structures between $\tJ$s, it does \emph{not} guarantee that the resulting chiral algebra will be associative to all orders in quantum corrections, cf. \cite{Costello:2022wso,Costello:2022upu}. The reason is that for $\ca$ to be associative, it typically requires the absence of anomalies order by order in perturbation theory. This often demands the introduction of additional fields such as axionic currents. This will be discussed in Section \ref{sec:axion-currents}. (Previous work on $\ca$ of various self-dual theories can be found e.g. in \cite{Costello:2022upu,Bittleston:2022jeq}.)

%%%%%%%%%%%%%%%%%%%%%%%%%%%%%%%%%%%%%%
\subsubsection{Chiral higher-spin algebras at classical level} 
Let us now consider the simplest example to see how the above algorithm works. Consider the gauge variation of 
\begin{align}\label{eq:tree-amplitudes-OPE}
   \delta\Bigg( \parbox{55pt}{\includegraphics[scale=0.19]{2-point-precollide.png}}+\parbox{55pt}{\includegraphics[scale=0.19]{cubic.png}}\Bigg)=0\,,\qquad \text{where}\quad \delta \tA=\bar{\p}\cc\,.
\end{align}
By virtue of gauge invariance, we write the above as
\small
\begin{align}\label{eq:tree-OPE-1}
    \int_{\C\times \C} dz_idz_j \Big(\tJ[h_i]\bar{\p}\cc_i\,\tJ[h_j]\tA[j]+\tJ[h_i]\tA_i\,\tJ[h_j]\bar{\p}\cc_j\Big)+\int_{\C} dz \tJ[h_x](\tA_i\star\bar{\p}\cc_j+\bar{\p}\cc_i\star\tA_j)=0\,.
\end{align}
\normalsize
Upon making a change of variables
\begin{align}\label{eq:change-var}
    z_0=\frac{z_i+z_j}{2}\,,\qquad z_{ij}=z_i-z_j\,,
\end{align}
the above can be cast into 
\small
\begin{align}\label{eq:tree-OPE-2}
    \int_{\C}dz_0\int_{|z_{ij}|=\epsilon} \!\!\!\!\!\!\! dz_{ij} \Big(\tJ[h_i]\bar{\p}\cc_i\,\tJ[h_j]\tA_j+\tJ[h_i]\tA_i\,\tJ[h_j]\bar{\p}\cc_j\Big)=\int_{\C} dz \tJ[h_x](\bar{\p}\cc_i\star\tA_j+\tA_i\star\bar{\p}\cc_j)\,,
\end{align}
\normalsize
where we note that there is an extra minus sign comes from the Jacobian regarding the change of variables \eqref{eq:change-var}. 
%%%%%%%%%%%%%%%%%%%%%%%%%%%%%%%%%%%%%%%%%%%%%%%%%%%%%%
Equating \eqref{eq:tree-OPE-2} using the test functions
\begin{align}
    \tA_{2h-2}=f_h(z)e^{-[w\,\tilde v]}d\bar{z} \,,\qquad \cc_{2h-2}=g_h(z)e^{-[w\,\tilde v] }\,,\qquad h\in \Z\,,%\theta(h)\frac{[w\,\tilde v]^{2|h|-1}}{(2|h|-1)!}d\bar{z}+\theta(-h)\frac{(-)^{(2|h|-1)_p}}{(2|h|-1)_p}\frac{1}{[w\,\tilde v]^{2|h|-1}}\,,
\end{align}
noting that $\tilde v$ has $U(1)$-charge $-1$, %$(a)_p=a(a+1)\ldots (a+p)$ being the ascendant Pochhammer symbol, 
we get:
\begin{align}\label{eq:full-tree-OPE}
    \tJ[h_i]^a\tJ[h_j]^b\sim \frac{1}{z_{ij}}\sum_p\tg^{abc}_p\frac{[\tilde v_i\,\tilde v_j]^p}{p!}\tJ^c[h_i+h_j-1-p]\,,\qquad \tJ[h_i]\equiv \tJ[h_i;0]\,,
\end{align}    
where $\tg_{k\in 2\N_0}^{abc}=f^{abc}\,,\tg_{k\in 2\N_0+1}=d^{abc}$ and
\begin{align}
    p=h_i+h_j-h_x-1\,.
\end{align}
Observe that the above helicity constraint %coming from the conservation of $SU(2)_-$ charges on the patch $\C$, which 
is slightly different with the helicity constraint obtained from the trivial monodromy condition on $S^7$, cf. \eqref{eq:k-constraint}. This is because we treat the momentum moving toward the defect as outgoing from the bulk vertices. %Thus, at classical level, all OPE coefficients of $\ca$ can be fixed unambigiously by gauge invariance of the Witten-like diagrams in the bulk/defect context. %Note that, we again temporarily suppress the $SU(2)_-$ charges for simplicity. They will be important when $\ca$ receives non-trivial quantum corrections.

By identifying $z_{ij}=\langle i\,j\rangle$ and $[i\,j]=[\tilde v_i\,\tilde v_j]$, we see that the factor
\begin{align}
    \cAmp_{split}^{tree-level}\sim\frac{[1\,2]^p}{\langle 1\,2\rangle}\,,\qquad p\geq 0\,,
\end{align}
can be identified with the holomorphic collinear limit of the soft factors for various chiral or self-dual higher-spin theories in flat space \cite{Tran:2022amg}. 

At classical level, the chiral algebra $\ca$ can be identified to the color-kinematic algebra studied in \cite{Ponomarev:2017nrr,Monteiro:2022xwq} by simply replacing $[\tilde v_i\,\tilde v_j]^p\mapsto X(\tilde v_i,\tilde v_j)^{p}$ using the notations of \cite{Monteiro:2011pc}, or $\overline{\mathbb{P}}$ in the notation in e.g. \cite{Ponomarev:2017nrr}. %The advantage of these variables is that they can be extended to arbitrary off-shell momenta.%\footnote{The extra power of $X(v_i,v_j)$ or $\overline{\mathbb{P}}$ comes from light-cone projection, see e.g. \cite{Chalmers:1996rq} for an explicit map.} 
We also note that the above OPE can reduce correctly to the OPEs associated to the affine non-unitary Kac-Moody algebra at level-0 (for $p=0$ and $\mg\neq \mathfrak{u}(1)$), and the non-unitary $w_{1+\infty}$ algebras (for $p=1$ and $\mg= \mathfrak{u}(1)$) \cite{Strominger:2021mtt}.\footnote{Here, non-unitarity stems from the fact that $h$ can be negative. If $h$s are strictly positive, the chiral algebras that we are constructing will be some wedge algebras instead.} Namely,
\begin{subequations}
    \begin{align}
        \underline{\text{level-0 Kac-Moody}}&:&\tJ^a[+1;0](z_1)\tJ^b[+1;0](z_2)&\sim\frac{f^{abc}}{z_{12}}\tJ^c[+1;0]\,,\nn\\
        & &\tJ^a[+1;0](z_1)\tJ^b[-1;0](z_2)&\sim\frac{f^{abc}}{z_{12}}\tJ^c[-1;0]\,,\\
        \underline{w_{1+\infty}}&: &\tJ[+2;0](z_1)\tJ[+2;0](z_2)&\sim\frac{[\tilde v_1\,\tilde v_2]}{z_{12}}\tJ[+2;0]\,,\nn\\
        & &\tJ[+2;0](z_1)\tJ[-2;0](z_2)&\sim\frac{[\tilde v_1\,\tilde v_2]}{z_{12}}\tJ[-2;0]\,.
    \end{align}
\end{subequations}
\normalsize
It is well-known that, the symmetry of affine non-unitary Kac-Moody algebra at level-0 governs self-dual Yang-Mills theory \cite{Chau:1981gi}, while the symmetry $w_{1+\infty}$ algebra is the underlying symmetry of self-dual gravity \cite{Strominger:2021mtt}. 
%Note that the affine Kac-Moody algebra with level-0 is the usual Lie algebra without the central charge term, while the $\wedge w_{1+\infty}$ algebra is the Poisson algebra associated to the Poisson structure $\Pi=-\epsilon^{\dot\alpha\dot\beta}\p_{\dot\alpha}\wedge\p_{\dot\beta}$. 
%Observe that all coefficients in \eqref{eq:full-tree-OPE} have been completely determined by gauge invariance. 

Let us now show that $\ca$ can be indeed identified with the color-kinematic algebra of $4d$ chiral/self-dual higher-spin theories found in \cite{Ponomarev:2017nrr,Monteiro:2022xwq}. (See also \cite{Mago:2021wje}.) Consider,
\begin{align}\label{eq:OPE-Lie}
    \oint_{|z_k|} %dz_k %z_k^{n_k} 
    \tJ[k]\oint_{|z_{ij}|<|z_k|}\!\!\! %z_i^{n_i}
    \tJ[i] %z_j^{n_j}
    \tJ[j]+\text{cyclic}(i,j,k)=0\,.
\end{align}
By feeding \eqref{eq:full-tree-OPE} to \eqref{eq:OPE-Lie} and %determine whether associativity holds at each level in perturbation theory. D
denoting $[i\,j]\equiv [\tilde v_i\,\tilde v_j]$ for convenience, we find that:
\begin{itemize}
    \item[I-] The affine non-unitary Kac-Moody algebra at level-0 is a Lie algebra due to the Jacobi's relations
    \begin{align}\label{eq:Jacobi-f}
        f^{abe}f^{ecd}+\text{cyclic}(b,c,d)=0\,,
    \end{align}
    between the structure constants $f^{abc}$.
    \item[II-] The non-unitary $w_{1+\infty}$ is also a \emph{kinematic} Lie algebra by virtue of the Schouten relations
    \begin{align}
        [1\,2][3\,4]+[2\,3][1\,4]+[3\,1][2\,4]=0\,,
    \end{align}
    where $\tilde v_4$ is an extra spinor satisfying $\sum_{i=1}^4\tilde v_i^{\dot\alpha}=0$. See \cite{Monteiro:2022xwq} for the discussion.
    \item[III-] The computation of \eqref{eq:OPE-Lie} for colored or colorless chiral higher-spin algebra $\ca$ can be done similarly. We proceed with the more general case, i.e. the colored case. Computing \eqref{eq:OPE-Lie}, we find
    \begin{align}\label{eq:pre-associative}
        \sum_{p+q=n}\Big(\tg_{p}^{abe}\tg_q^{ecd}\frac{[1\,2]^p[3\,4]^q}{p!q!}\Big)+\text{cyclic}(1,2,3)\,.
    \end{align}
    Upon splitting
\begin{align}
    \sum_{p\in \N_0}\tg_p^{abc}\frac{[i\,j]^p}{p!}=\sum_{p\in 2\N_0}f^{abc}_p\frac{[i\,j]^p}{p!}+\sum_{p\in 2\N_0+1}d^{abc}_p\frac{[i\,j]^{p}}{p!}\,,
\end{align}
we obtain from \eqref{eq:pre-associative} that
\small
\begin{align}
    &\Bigg[\sum_{\substack{p,q\in 2\N_0\\
    p+q=n}}f^{a_1b_2e}f^{ec_3d_4}\frac{[1\,2]^p[3\,4]^q}{p!q!}+\text{cyclic}(1,2,3)\Bigg]\nn\\
    &+\Bigg[\sum_{\substack{p\in 2\N_0\\
    q\in 2\N_0+1\\
    p+q=n}}f^{a_1b_2e}d^{ec_3d_4}\frac{[1\,2]^p[3\,4]^q}{p!q!}+\text{cyclic}(1,2,3)\Bigg]\nn\\
    &+\Bigg[\sum_{\substack{p\in 2\N_0+1\\
    q\in 2\N_0\\
    p+q=n}}d^{a_1b_2e}f^{ec_3d_4}\frac{[1\,2]^p[3\,4]^q}{p!q!}+\text{cyclic}(1,2,3)\Bigg]\nn\\
    &+\Bigg[\sum_{\substack{p,q\in 2\N_0+1\\
    p+q=n}}d^{a_1b_2e}d^{ec_3d_4}\frac{[1\,2]^p[3\,4]^q}{p!q!}+\text{cyclic}(1,2,3)\Bigg]\,,
\end{align}
\normalsize
where each of the square brackets above can be shown to vanish by some simple algebra, cf. \cite{Monteiro:2022xwq}. 
\end{itemize}

\begin{corollary} The chiral higher-spin algebra $\ca$ can be classically identified with a color-kinematic algebra where the generators $\cJ$ obey the Jacobi's identities
\begin{align}
    [[\cJ^a,\cJ^b],\cJ^c]+ [[\cJ^b,\cJ^c],\cJ^a]+[[\cJ^c,\cJ^a],\cJ^b]=0\,,
\end{align}
%where the commutation relations between generators are given by
for
%%\begin{align}\label{eq:chiral-algebra-relations}
%    \Blb \cJ_{m}[i],\cJ_{n}[j] \Brb =\oint_0 dz_j\, z_j^n \oint_{|z_{ij}|<|z_j|}dz_i\,z_i^m\sum_{n\in \Z}\frac{1}{z_{ij}^{n+h_i+k_i}}V(J_n[k_i]\circ  \tJ[j])(z_j)\,.
%\end{align}
%One can show that the algebra relations of generators $\cJ$ have the following form:
\begin{align}
    \Big[\cJ^a[h_i],\cJ^b[h_j]\Big]=\sum_p\tg^{abc}_p\frac{[\tilde v_i\,\tilde v_j]^p}{p!}\cJ^c[h_i+h_j-1-p]\,.
\end{align} 
\end{corollary}

%We refer to the above as the \emph{associativity constraint}. %for the generators $\cJ_{-\ell_a}=\oint z^{\ell_a}\tJ^a[i]$. 
%This condition should hold to all order in perturbation theory by virtue of the Koszul duality between $\ca$ and $\Coh^{\star}$, as explained above. Note that, 
%%%%%%%%%%%%%%%%%%%%%%%%%%%%%%%%%%%%%%%%%%%%%%%%%%%
\subsubsection{Chiral higher-spin algebras at one loop}
Let us now study the chiral higher-spin algebra at one loop by considering the gauge invariance of the following diagrams:
\begin{align}\label{eq:loop-amplitudes-OPE}
   \delta\Bigg( \parbox{57pt}{\includegraphics[scale=0.19]{2-point-precollide.png}}+\parbox{50pt}{\includegraphics[scale=0.19]{s-channel.png}}+\parbox{52pt}{\includegraphics[scale=0.19]{t-channel.png}}+\parbox{53pt}{\includegraphics[scale=0.19]{u-channel.png}}\Bigg)=0\,.
\end{align}
Before proceeding, we note that the reason only the above $s$-, $t$- and $u$-channels diagrams can potentially account for quantum correction of the OPE at one loop is that tadpole, bubble-on-external-leg, and triangle diagrams vanish either by symmetry arguments or algebraic constraints. In particular, since we are dealing with massless theories, tadpole diagrams must vanish to avoid Lorentz violation. Meanwhile, contributions from bubble diagrams on external legs can be absorbed through field redefinitions. Finally, the triangle diagram vanishes since it can be written as a total derivative. Indeed, a triangle bulk-diagram can be evaluated as
\begin{align}
    \parbox{35pt}{\includegraphics[scale=0.3]{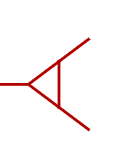}}\sim\sum_{\Spec}\int_{\C^3}\mho\Tr(\p\sA\p\sA\p\sA)=0\,.
\end{align}
Therefore, these types of diagrams can be safely discarded without affecting the one-loop correction to the OPEs of $\ca$.

%Similar to the usual context of scattering amplitudes, in the bulk/defect setting, one also fixes a gauge and compute Witten-like diagram. 

%%%%%%%%%%%%%%%%%%%%%%%%%%%%%%%%%%%%%%
\paragraph{On one-loop computation.} Let us now compute the channels in \eqref{eq:loop-amplitudes-OPE}. To have a uniform treatment, we will parametrize a generic point $X_i$ in $\C^3\subset \PT$ as
\begin{align}\label{eq:C3-patch}
    X^a_i=(z_i,w_i^{\dot\alpha})=(z_i,w_i^{\dot 1},w_i^{\dot 2})\,.
\end{align}
Since there is no confusion can arise, we can suppress the dots to ease the notation. Now, in the twistorial `Lorenz gauge' (see e.g. \cite{Bittleston:2022nfr}) 
\begin{align}
    \bar{\p}^{\dagger}\sA=0\,,\qquad \bar{\partial}^{\dagger}:=-\ast \bar\partial\,\ast\,,
\end{align}
any regularized bulk propagator of the physical fields $\sA^{0,1}$ can be described by a bi-local matrix-valued $(0,2)$-form on $\C^3_{X_1}\times \C^3_{X_2}$ as
\begin{align}
    \cP^{h_1,h_2}(X_1,X_2|\varepsilon,L)=-\delta_{h_1+h_2,0}\,\Omega_{12}^{(0,2)}\int_{\varepsilon}^{L}\frac{d\ell}{2\ell}\Big(\frac{1}{4\pi\ell}\Big)^3e^{-\frac{|X_{12}|^2}{4\ell}}\,,
\end{align}
where $\varepsilon<L$ are the UV and IR characteristic length scales, respectively. Moreover,
\begin{align}
    \Omega^{(0,2)}_{12}=\epsilon_{acb}\bar{X}_{12}^ad\bar{X}_{12}^bd\bar{X}_{12}^c\,,\quad X_{12}^a=X_1^a-X_2^a\,,\qquad a=1,2,3\,.
\end{align}
Sending $\varepsilon\rightarrow 0$ and $L\rightarrow\infty$, we obtain
\begin{align}
    \bar{\p}\cP^{h_1,h_2}(X_1,X_2|0,\infty)=-\delta_{h_1+h_2,0}\delta^{0,3}(X_1-X_2)\,.
\end{align}
where, $\delta^{0,3}$ is a holomorphic generalized Dirac delta function $(0,3)$-form. The `bulk-to-bulk' propagator can be evaluated explicitly as
\begin{align}
    \cP^{h_1,h_2}(X_1,X_2|0,\infty)=-\frac{\delta_{h_1+h_2,0}\,\Omega^{(0,2)}_{12}}{\pi^3(|X_{12}|^2)^3}=-\frac{\delta_{h_1+h_2,0}\,\Omega^{(0,2)}_{12}}{\pi^3\big(|z_{12}|^2+|w_{12}|^2\big)^3}\,,
\end{align}
where
\begin{align}
    |w_{12}|^2=|w_{12}^1|^2+|w_{12}^{2}|^2\,.
\end{align}
The, the `bulk-to-defect' propagator $\cK$ can be obtained by sending one of the legs of $\cP$ to the defect -- for instance, to the point $Y_2=(z_2,0)$. Namely,
\begin{align}
    \cK^{h_1,h_2}(X_1,Y_2|0,\infty)=-\frac{\delta_{h_1+h_2,0}\,\omega_{12}^{(0,2)}}{\pi^3\big(|X_1-Y_2|^2\big)^3}=-\frac{\delta_{h_1+h_2,0}\,\omega_{12}^{(0,2)}}{\pi^3\big(|z_{12}|^2+|w_1|^2\big)^3}\,,
\end{align}
In terms of integral representation,
\begin{align}\label{eq:star-K}
    \cK^{h_1,h_2}(X_1,Y_2)=-\delta_{h_1+h_2,0}\omega_{12}^{(0,2)}\int_{0}^{\infty}\frac{d\ell}{2\ell}\Big(\frac{1}{4\pi\ell}\Big)^3\,\exp\Big(-\frac{|z_{12}|^2+|w_1|^2}{4\ell}\Big)\,.
\end{align}
%In what follows, we will write $(\Omega,X)$ for bulk data and $(\omega,Y)$ for bulk-to-boundary data for distinction. 
An explicit computation shows that
\begin{subequations}
    \begin{align}
        \Omega^{(0,2)}_{12}&=\bar{z}_{12}[d\bar{w}_{12}\wedge d\bar{w}_{12}]-2d\bar{z}_{12}[\bar{w}_{12}d\bar{w}_{12}]\,,\\
        \omega^{(0,2)}_{12}&=\bar{z}_{12}[d\bar{w}_{1}\wedge d\bar{w}_{1}]-2d\bar{z}_{12}[\bar{w}_{1}d\bar{w}_{1}]\,.
    \end{align}
\end{subequations}
To ease the notations, we shall write
\begin{align}
    \Omega^{(0,2)}_{X_1X_2}=\bar{\Omega}_{X_1X_2}\,,\qquad \omega^{(0,2)}_{X_1Y_2}=\bar{\omega}_{X_1Y_2}\,,
\end{align}
from now on. 

Let us now point out a nice observation in \cite{Bittleston:2023bzp}, which leads to the vanishing of the $s$-channel and a good deal of other bulk/defect digrams. 
\begin{lemma}\label{lem:tip-diagrams} Any bulk/defect diagrams of the form
\begin{align}\label{eq:tip-diagrams}
    \parbox{80pt}{\includegraphics[scale=0.4]{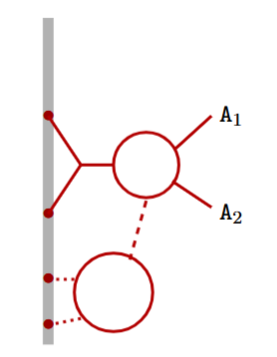}}\,,
\end{align}
where there are two bulk-to-boundary propagators form a triangle with the defect, vanish. 
\end{lemma}
\begin{proof} By direct computation, it can be shown that
\begin{align}
    \cK(X,Y_1|0,\infty)\cK(X,Y_2|0,\infty)\propto\bar{\omega}_{XY_1}\bar{\omega}_{XY_2}\sim [\bar{w}_X\,\bar{w}_X]\,[d \bar{w}_X\wedge d \bar{w}_X] d\bar{z}_{X}\wedge d\bar{z}_{X}=0\,.
\end{align}
Consequently, diagrams of type \eqref{eq:tip-diagrams} do not affect the quantum corrections to the OPE between $\tJ$, thus can be discarded systematically.
\end{proof}
A direct consequence of Lemma \ref{lem:tip-diagrams} is that all allowing processes in the bulk/defect system in consideration should form quadrilateral or polygon diagrams with the defect.

%%%%%%%%%%%%%%%%%%%%%%%%%%%%%%%%%%%%
\paragraph{The $t$-channel.} We can now consider the gauge variance of the following equation 
\begin{align}\label{eq:gauge-invariance-without-axion}
   \delta\Bigg( \parbox{55pt}{\includegraphics[scale=0.19]{2-point-precollide.png}}+\parbox{52pt}{\includegraphics[scale=0.19]{t-channel.png}}+\parbox{53pt}{\includegraphics[scale=0.19]{u-channel.png}}\Bigg)=0\,,
\end{align}
to extract the OPE between two currents at first order in quantum correction. 

It is useful to remind the reader that $\bar{\p}\cP(X_2,X_3)\sim \delta^{0,3}(X_2-X_3)$ and $\bar{\p}\sA\approx 0$. Thus, by doing an integration by part, we see that the $t$- and $u$-channels will be pinched whenever $\bar{\p}$ acts on one of their propagators. The contributions associated with pinched diagrams can then be safely discarded by virtue of Lemma \ref{lem:tip-diagrams}. Furthermore, when $\bar{\p}$ acts on $\tA$, the diagrams will also vanish on-shell.

Let us now compute the gauge variation of the $t$-channel. It is necessary to choose a convention for the orientation of vertices to avoid over counting. %, as propagators are bi-local in position space. 
In particular, we take fields to be incoming toward bulk vertices, while those coupled to currents are outgoing. Then,
\begin{align}\label{eq:t-channel}
    \parbox{60pt}{\includegraphics[scale=0.23]{t-channel.png}}&=%\sum_{\substack{\Delta_1,\Delta_4\%\
   % -k_1,-k_4}}
   \sum_{h_1,h_4\in \Spec}\int_{\C\times\C}\tJ[h_1]\tJ[h_4]\int_{\C^3\times \C^3}\cM_3(\cK,\cP,\tA)\,,
\end{align}
\normalsize
where 
\begin{align}
   \cM(\cK,\cP,\tA)= \big(\cK^{h_1,h_x}(Y_1,X_2)\star \tA_{2h_2-2}\big)\cP^{h_x,h_y}(X_2,X_3)\big(\cK^{h_y,h_4}(Y_4,X_3)\star \tA_{2h_3-2}\big)\,.
\end{align}
\normalsize
Here, $\int_{\C}$ denotes the integration over the defect, and $\int_{\C^3}$ stands for integral over the bulk points. 

Perhaps, it is useful to interlude our computation a short discussion about the sum over the spectrum in \eqref{eq:t-channel}. 
Typically, when one computes amplitudes, the external helicities should be fixed. However, in the defect-construction procedure that we are considering, the CFT data of $\tJ[h_1,\tH_1]$ and $\tJ[h_4,\tH_4]$ are not given, a priori. These data, however, should be determined by symmetries which leads to some non-trivial relations involving the helicity of the bulk fields $\tA$. For this reason, the sums over the spectrum $\Spec$ represents some sort of ``integration''. 

To proceed, let us once again make the change of variables
\begin{align*}
    z_0=\frac{z_1+z_4}{2}\,,\qquad z_{14}=z_1-z_4\,,
\end{align*}
and consider the following test functions%\footnote{In the affine patch $z$ and $d\bar z$ have $SU(2)_-$ charge $-1$.}
\begin{align}\label{eq:test-functions-1-loop}
     \cc_{2s_2-2}(X_2)=z_{2}[w_2\,\tilde v_2]^{2s_2-1}\,,\qquad \tA_{2s_3-2}(X_3)=[w_3\,\tilde v_3]^{2s_3-1}d\bar{z}_3\,.
\end{align}
Note that these are the test functions associated to the left-most diagram in \eqref{eq:gauge-invariance-without-axion}. 
(Here, we consider generating functions with only $s\geq 1$, due to the trivial monodromy constraint on the total space $S^7$, which in turn leads to non-trivial integrals. This can be seen shortly.) 
Feeding these functions to the $t$-channel, and take a gauge variation, we obtain 
\begin{align}\label{eq:t-channel-variation}
    \delta \eqref{eq:t-channel}=-&\sum_{h_1,h_4\in \Spec}\sum_{i,j}\frac{\Gamma(2s_2)}{\Gamma(2s_2-i)}\frac{\Gamma(2s_3)}{\Gamma(2s_3-j)}\tg_i^{a_1a_2e}\tg_j^{ea_3a_4}\int_{\C}dz_0 \oint_{|z_{14}|=\epsilon}\!\!\!\!\!\!dz_{14}\tJ^{a_1}[h_1]\tJ^{a_4}[h_4]\nn\\
    &\times\int_{\C^3\times\C^3}\!\!\!\! (dX_2)^3(dX_3)^3\,z_2[w_2\,\tilde v_2]^{2s_2-i-1} \,[w_3\,\tilde v_3]^{2s_3-j-1} d\bar{z}_3\frac{\bar{\omega}_{Y_1X_2}\bar{\Omega}_{X_2X_3}\bar{\omega}_{Y_4X_3}}{\pi^3|X_{23}|^6}\nn\\
    &\times\int \frac{d\ell_1}{2\ell_1}\Big(\frac{1}{4\pi\ell_1}\Big)^3\frac{d\ell_2}{2\ell_2}\Big(\frac{1}{4\pi\ell_2}\Big)^3\Big(\frac{[\bar{w}_2\,\tilde v_2]^i}{(4\ell_1)^ii!}\Big)\Big(\frac{[\bar w_3\,\tilde v_3]^j}{(4\ell_2)^jj!}\Big)e^{-\frac{|Y_1-X_2|^2}{4\ell_1}-\frac{|Y_4-X_3|^2}{4\ell_2}}\,,
\end{align}
\normalsize
where
\begin{align}\label{eq:monodromy-t-channel}
    i=-h_1+s_2+h_x-1\,,\qquad j=-h_4+s_3-h_x-1
\end{align}
is the helicity-spin constraint coming from trivial monodromy condition on $S^7$. Note that
\begin{align}
    h_x\in[h_1-s_2,s_3-h_4]
\end{align}
stands for the helicities in the exchange. Now, we do the integration over $\ell_i$ and obtain
\begin{align}
    \eqref{eq:t-channel}&=-\frac{1}{4\pi^9}\sum_{h_1,h_4\in \Spec}\sum_{i,j}\frac{\Gamma(2s_2)}{\Gamma(2s_2-i)}\frac{\Gamma(2s_3)}{\Gamma(2s_3-i)}\tg_i^{a_1a_2e}\tg_j^{ea_3a_4}\int_{\C}dz_0 \oint_{|z_{14}|=\epsilon}\!\!\!\!\!\!dz_{14}\tJ^{a_1}[h_1]\tJ^{a_4}[h_4]\nn\\
    &\times\int_{\C^3\times\C^3}\!\!\!\! (dX_2)^3(dX_3)^3\,\,z_2[w_2\,\tilde v_2]^{2s_2-i-1} \,[w_3\,\tilde v_3]^{2s_3-j-1} d\bar{z}_3\frac{\bar{\omega}_{Y_1X_2}\bar{\Omega}_{X_2X_3}\bar{\omega}_{Y_4X_3}}{\pi^3|X_{23}|^6}\nn\\
    &\times (i+1)(i+2)(j+1)(j+2)\frac{[\bar w_2\,\tilde v_2]^i[\bar{w}_3\,\tilde v_3]^j}{|Y_1-X_2|^{2(3+i)}|Y_4-X_3|^{2(3+j)}}\,.
\end{align}
It is a simple computation to show that
\begin{align}
    \bar{\omega}_{Y_1X_2}\bar{\Omega}_{X_2X_3}\bar{\omega}_{Y_4X_3}=2\bar{z}_{14}[\bar{w}_2\,\bar{w}_3](d\bar{X}_2)^3(d\bar{X}_3)^3\,,
\end{align}
where $(d\bar{X})^3=d\bar z_X[d\bar w_X\wedge d\bar w_X]$. Thus,
\begin{align}\label{eq:integral-t-1}
      \eqref{eq:t-channel}&=-\frac{1}{2\pi^9}\sum_{h_1,h_4\in \Spec}\sum_{i,j}\frac{\Gamma(2s_2)}{\Gamma(2s_2-i)}\frac{\Gamma(2s_3)}{\Gamma(2s_3-i)}(i+1)(i+2)(j+1)(j+2)\tg_i^{a_1a_2e}\tg_j^{ea_3a_4}\nn\\
      &\times\int_{\C}dz_0\oint_{|z_{14}|=\epsilon}dz_{14}\,\bar{z}_{14}\tJ^{a_1}[h_1]\tJ^{a_4}[h_4]\times I_{23}^{a_2a_3}\,,
\end{align}
where 
\begin{align}
    I_{23}^{a_2a_3}:=\int_{\C^3\times\C^3}DX_2DX_3 %\frac{(\bar{z}_{1x}[\bar\mu_u\,\bar\mu_x]+\bar{z}_{4y}[\bar{\mu}_x\,\bar{\mu}_y])}{}
      \frac{[\bar w_2\,\tilde v_2]^i[\bar w_2\,\bar{w}_3][\bar{w}_3\,\tilde v_3]^j\, z_2[w_2\,\tilde v_2]^{2s_2-i-1} \,[w_3\,\tilde v_3]^{2s_3-j-1} d\bar{z}_3}{|Y_1-X_2|^{2(3+i)}|X_2-X_3|^6|Y_4-X_3|^{2(3+j)}}\,,
\end{align}
and $DX=(dX)^3(d\bar{X})^3$. 

Notice that the integral over the bulk points resembles a doubly nested bubble integral. Furthermore, since the measures $DX_2,DX_3$, and the propagators are real, we must require
\begin{align}\label{eq:nominator-1}
    [\bar w_2\,\tilde v_2]^i[\bar w_2\,\bar{w}_3][\bar{w}_3\,\tilde v_3]^jz_2[w_2\,\tilde v_2]^{2s_2-i-1} \,[w_3\,\tilde v_3]^{2s_3-j-1} d\bar{z}_3
\end{align}
to be also real, otherwise $I_{23}^{a_2a_3}$ will vanish. As we are doing integration over complex variables, we can again assume that they are charged under $U(1)$. Then, the integral in the bulk, cf. \eqref{eq:integral-t-1}, survives iff the nominator has a trivial monodromy. We deduce that 
\begin{align}\label{eq:monodromy-1-loop-t}
    i=s_2-1\,,\qquad j=s_3-1\,.
\end{align}
Since $i,j\geq 0$. The above explains why we chose $s_{2,3}\geq 1$ in the first place. Namely, it is the only way for the variation of the t-channel, cf. \eqref{eq:t-channel}, to receive non-trivial quantum correction. Using the above and \eqref{eq:monodromy-t-channel}, we can fix
\begin{align}
    h_1=-h_4\,.
\end{align}
This is a simple yet robust constraint, which allows us to insert a Kronecker delta $\delta_{h_1+h_4,0}$ and take the sum over helicities.\footnote{This constraint may be also guessed from the beginning since we are computing loop integral. However, it is still useful to show where it it comes from.} Moreover, for simplicity, we can set 
\begin{align}
    \tilde v_2^{\dot\alpha}=(+1,0)\,,\qquad \tilde v_3^{\dot\alpha}=(0,-1)\,,
\end{align}
so that
\begin{align}\label{eq:doubly-nested-int}
    I_{23}:=\int_{\C^3\times\C^3}DX_2DX_3 
      \frac{z_2|w_2^{\dot 2}|^{2s_2}|w_3^{\dot 1}|^{2s_3}}{|Y_1-X_2|^{2(s_2+2)}|X_2-X_3|^6|Y_4-X_3|^{2(s_3+2)}}d\bar{z}_0\,,
\end{align}
%\begin{align}\label{eq:doubly-nested-int}
%    I_{23}:=\int_{\C^3\times\C^3}DX_2DX_3 
%      \frac{z_2|w_2^{\dot 2}|^{2(2s_2-1)}|w_3^{\dot 1}|^{2(2s_3-1)}}{|Y_1-X_2|^{2(2s_2+1)}|X_2-X_3|^6|Y_4-X_3|^{2(2s_3+1)}}d\bar{z}_0\,,
%\end{align}
upon restricting $d\bar{z}_y=d\bar{z}_0$. %Note that our choices of test functions are different from those of \cite{Costello:2022upu} (for SDYM), and \cite{Bittleston:2023bzp} (for SDGR). Namely, we will keep all external bulk fields to be the same in \eqref{eq:gauge-invariance-without-axion}, but let the monodromy determine the integral for us. 
We can now reduce \eqref{eq:integral-t-1} to
\begin{align}\label{eq:integral-t-3}
      \eqref{eq:integral-t-1}&=-\frac{2}{\pi^9}\sum_{h\in \Spec}\frac{(s_2+1)\Gamma(2s_2)}{\Gamma(s_2)}\frac{(s_3+1)\Gamma(2s_2)}{\Gamma(s_3)}\tg_{s_2-1}^{a_1a_2e}\tg_{s_3-1}^{ea_3a_4}\nn\\
      &\qquad \times\int_{\C}d^2z_0\oint_{|z_{14}|=\epsilon}dz_{14}\bar{z}_{14}\tJ^{a_1}[h]\tJ^{a_4}[-h]\times I_{23}^{a_2a_3}\,.
\end{align}
%\begin{align}\label{eq:integral-t-3}
%      \eqref{eq:integral-t-1}&=-\frac{2}{\pi^9}\sum_{s_2,s_3\geq 1}\frac{(s_2+1)\Gamma(2s_2)}{\Gamma(s_2)}\frac{(s_3+1)\Gamma(23_2)}{\Gamma(3_2)}\tg_{2h_2-2}^{a_1a_2e}\tg_{2h_3-2}^{ea_3a_4}\nn\\
%      &\times\int_{\C}d^2z_0\oint_{|z_{14}|=\epsilon}dz_{14}\bar{z}_{14}\tJ^{a_1}[h_1]\tJ^{a_4}[h_4]\times I_{23}^{a_2a_3}\,.
%\end{align}
After a short computation in the Appendix \ref{app:integral}, we end up with 
\begin{align}
    I_{23}=\frac{\Ccurl_{s_2,s_3}}{2|z_{14}|^2}\Big[(s_2+2)z_0-s_2\frac{z_{14}}{2}\Big]\,,
\end{align}
where 
\begin{align}
    \Ccurl_{s_2,s_3}=\frac{(2\pi)^6}{4s_3(s_2+1)(s_2+2)}\,.
\end{align}
This simplifies the variation of the $t$-channel to
\begin{align}
     \eqref{eq:t-channel-variation}&=\sum_{h\in \Spec}\tC_{s_2,s_3}^t\times\tg_{s_2-1}^{a_1a_2e}\tg_{s_3-1}^{ea_3a_4}\nn\\
     &\qquad \qquad \times\int_{\C}d^2z_0\oint_{|z_{14}|=\epsilon}\frac{dz_{14}}{z_{14}}\Big[(s_2+2)z_0-s_2\frac{z_{14}}{2}\Big]\tJ^{a_1}[h]\tJ^{a_4}[-h]\,,
\end{align}
where
\begin{align}
    \tC^t_{s_2,s_3}=-\frac{32}{\pi^3}\frac{\Gamma(2s_2)\Gamma(2s_3)}{\Gamma(s_2)\Gamma(s_3)}\frac{(s_3+1)}{s_3(s_2+2)}\,.
\end{align}
To proceed, we will split 
\begin{align}
    \eqref{eq:t-channel}=\cI_t^{(\tJ)}+\cI_t^{(\tJ\tJ)}\,,
\end{align}
where
\begin{subequations}
    \begin{align}
    \cI^{(\tJ)}_t&=-\frac{1}{2}\sum_{h\in\Spec}\tC^t_{s_2,s_3}\tg_{s_2-1}^{a_1a_2e}\tg_{s_3-1}^{ea_3a_4}\int_{\C}d^2z_0\oint_{|z_{14}|=\epsilon} dz_{14}\tJ^{a_1}[h]\tJ^{a_4}[-h]\,,\\
    \cI^{(\tJ\tJ)}_t&=+\sum_{h\in \Spec}\tC^t_{s_2,s_3}(s_2+2)\tg_{s_2-1}^{a_1a_2e}\tg_{s_3-1}^{ea_3a_4}\int_{\C}d^2z_0 z_0\tJ^{a_1}[h]\tJ^{a_4}[-h]\,.
\end{align}
\end{subequations}
Here, $(\tJ)$ stands for single-$\tJ$ contribution, and $(\tJ\tJ)$ stands for double-$\tJ$ contribution. Then, using the classical OPE's \eqref{eq:full-tree-OPE}, we can reduce
\begin{align}
    \cI^{(\tJ)}_t&=-\frac{1}{2}\sum_{h\in \Spec}\tC^t_{s_2,s_3}\tg_{s_2-1}^{a_1a_2e}\tg_{s_3-1}^{ea_3a_4}\int_{\C}d^2z_0\sum_p\tg_p^{a_1a_4f}\frac{[\tilde v_2\,\tilde v_3]^p}{p!}\tJ^f[-1-p]\,.
\end{align}
\normalsize
Intriguingly, setting $p=0$ we reproduce the result of \cite{Costello:2022upu}, while setting $p=1$ we obtain the result of \cite{Bittleston:2023bzp}. What we observe here is that the single-$\tJ$ operator at one loop is a derivative-dependent term, while the double-$\tJ$ operator $\sum_h\tJ[-h]\tJ[h]$ solely depends on the spectrum of the bulk theory.

%%%%%%%%%%%%%%%%%%%%%%%%%%%%%%%%%%%%%%%%%%%%
\paragraph{The $u$-channel.} In computing the $u$-channel, we can simply perform the permutation $(a_2\leftrightarrow a_3)$ and $(s_2\leftrightarrow s_3)$. This results in 
\begin{align}
    (\text{u-channel})=\cI_u^{(\tJ)}+\cI_u^{(\tJ\tJ)}\,,
\end{align}
where
\begin{subequations}
    \begin{align}
    \cI^{(\tJ)}_u&=-\frac{1}{2}\sum_{h\in\Spec}\tC^u_{s_2,s_3}\tg_{s_3-1}^{a_1a_3e}\tg_{s_2-1}^{ea_2a_4}\int_{\C}d^2z_0\oint_{|z_{14}|=\epsilon} dz_{14}\tJ^{a_1}[h]\tJ^{a_4}[-h]\,,\\
    \cI^{(\tJ\tJ)}_u&=+\sum_{h\in\Spec}\tC^u_{s_2,s_3}(s_3+2)\tg_{s_3-1}^{a_1a_3e}\tg_{s_2-1}^{ea_2a_4}\int_{\C}d^2z_0 z_0\tJ^{a_1}[h]\tJ^{a_4}[-h]\,.
\end{align}
\end{subequations}
Note that in this case,
\begin{align}
    \tC^u_{s_2,s_3}=- \frac{32}{\pi^3}\frac{\Gamma(2s_2)\Gamma(2s_3)}{\Gamma(s_2)\Gamma(s_3)}\frac{(s_2+1)}{s_2(s_3+2)}\,.
\end{align}
As usual, for color-ordered partial amplitudes, we will not need to include the $u$-channel. Here, we only compute the $u$-channel for completeness.
%%%%%%%%%%%%%%%%%%%%%%%%%%%%%%%%%%%%%%%%%%%%
\paragraph{Correction to the OPE at one loop.} We can now finalize the correction to the OPE between $\tJ[1]$ and $\tJ[4]$ at one loop. Adding the $t$-channel and $u$-channel together, we get
\begin{align}
    \cI^{(\tJ)}&=\cI^{(\tJ)}_t+\cI^{(\tJ)}_u\nn\\
    &=-\frac{1}{2}\sum_{h\in\Spec}\cU_{s_2,s_3}^{a_1a_2a_3a_4}\int_{\C}d^2z_0\sum_p\tg_p^{a_1a_4f}\frac{[v_2\,v_3]^p}{p!}\tJ^f[-1-p]\,.
\end{align}
where
\begin{align}
    \cU_{s_2,s_3}^{a_1a_2a_3a_4}=\tC^t_{s_2,s_3}\tg_{s_2-1}^{a_1a_2e}\tg_{s_3-1}^{ea_3a_4}+\tC^u_{s_2,s_3}\tg_{s_3-1}^{a_1a_3e}\tg_{s_2-1}^{ea_2a_4}\,.
\end{align}
We also have
\begin{align}
    \cI^{(\tJ\tJ)}&=\cI^{(\tJ\tJ)}_t+\cI^{(\tJ\tJ)}_u\nn\\
    &=\sum_{h\in\Spec}\cT_{s_2,s_3}^{a_1a_2a_3a_4}\int_{\C}d^2z_0 z_0\tJ^{a_1}[h]\tJ^{a_4}[-h]
\end{align}
where
\begin{align}
    \cT_{s_2,s_3}^{a_1a_2a_3a_4}=-\frac{1}{2}\times\frac{32}{\pi^3}\frac{\Gamma(2s_2)\Gamma(2s_3)}{\Gamma(s_2)\Gamma(s_3)}\Big[\frac{(s_3+1)}{s_3}\tg_{s_2-1}^{a_1a_2e}\tg_{s_3-1}^{ea_3a_4}+\frac{(s_2+1)}{s_2}\tg_{s_3-1}^{a_1a_3e}\tg_{s_2-1}^{ea_2a_4}\Big]\,.
\end{align}
Note that a factor of $\frac{1}{2}$ is inserted to account for the fact that we have counted everything twice in the $(\tJ\tJ)$-sector.
%%%%%%%%%%%%%%%%%%%%%%%%%%%%%%%%%%%%%%%%
\paragraph{Matching.} Together, $\cI^{(\tJ)}+\cI^{(\tJ\tJ)}$ should cancel the gauge (or BRST) variation of the bilocal term 
\begin{align}\label{eq:bilocal-cancellation-1}
    -\int_{\C}dz_0\oint_{|z_{14}|=\epsilon}dz_{14} \tJ[s_2,\tH_2]\tJ[s_3,\tH_3]\cc_{2s_2-2} \tA_{2s_3-2}\,,
\end{align}
on the defect. 

Next, to eliminate all $w$'s factors in the test functions \eqref{eq:test-functions-1-loop}, we fix $\tH_i=2s_i-1$ so that
\begin{subequations}
    \begin{align}
    \tJ\Big[s_2;2s_2-1\Big]&=\frac{1}{(2s_2-1)!}\tJ[s_2]_{\dot\alpha(2s_2-1)}\p^{\dot\alpha(2s_2-1)}\,,\\
    \tJ\Big[s_3;2s_3-1\Big]&=\frac{1}{(2s_3-1)!}\tJ[s_3]_{\dot\alpha(2s_3-1)}\p^{\dot\alpha(2s_3-1)}\,,
\end{align}
\end{subequations}
Using the fact that $z_1=z_0+\frac{z_{14}}{2}$, we write
\begin{align}
    \eqref{eq:bilocal-cancellation-1}=-\int d^2z_0\oint_{|z_{14}|=\epsilon}dz_{14}\Big(z_0+\frac{z_{14}}{2}\Big)\tJ[s_2]\tJ[s_3]\,.
\end{align}
This leads to the following identification at one loop: 
\begin{align}\label{eq:1-loop-OPE}
    \tJ^{a_2}\Big[s_2&;2s_2-1\Big](z)\tJ^{a_3}\Big[s_3;2s_3-1\Big](0)\nn\\
    \sim&-\Big(\frac{1}{z^2}-\frac{1}{2z}\p_z\Big) \sum_{h\in\Spec}\cU^{a_1a_2a_3a_4}_{s_2,s_3}\sum_p\tg^{a_1a_4f}_p\frac{[v_2\,v_3]^p}{p!}\tJ^f[-1-p](z)\nn\\
    &+\frac{1}{z}\sum_{h\in\Spec}\cT^{a_1a_2a_3a_4}_{s_2,s_3}\nor \tJ^{a_1}[h]\tJ^{a_4}[-h](z)\nor\,.
\end{align}
\normalsize
%\begin{align}\label{eq:1-loop-OPE}
%    \tJ^{a_2}\Big[h_1+s_2&;-s_2+\frac{1}{2}\Big](z)\tJ^{a_3}\Big[h_4+s_3;-s_3+\frac{1}{2}\Big](0)\sim\nn\\
%    &+\frac{8}{\pi^3}\Big(\frac{1}{z^2}+\frac{1}{2z}\p_z\Big) \!\!\!\!\!\sum_{s_2+s_3=h_1+h_4+2} \!\!\!\!\!\cU^{a_1a_2a_3a_4}_{s_2,s_3}\sum_p\tg^{a_1a_4f}_p\frac{[v_1\,v_4]^p}{p!}\tJ^f[h_1+h_4-1-p]\nn\\
%    &-\frac{8}{\pi^3}\frac{1}{z}\sum_{s_2+s_3=h_1+h_4+2}\cT^{a_1a_2a_3a_4}_{s_2,s_3}\nor \tJ^{a_1}[h_1]\tJ^{a_4}[h_4]\nor\,.
%\end{align}
\normalsize
Here, we have introduced a normal-ordered product denoted by $\nor$ $\nor$ for the double-$\tJ$ operator $\tJ\tJ$ as to regularize the singular behavior. This normal-ordered product is given by
\begin{align}\label{eq:normal-order-def}
    \nor A \,B(z)\nor\ \  =\oint_{|w-z|=1}\frac{dw}{w-z}A(z)B(w)\,.
\end{align}
Note that in the simplest case where $s_2=s_3=1$, $h=1$, and $p=0$ we can use the Jacobi's relations \eqref{eq:Jacobi-f} to reproduces the result of \cite{Costello:2022upu} up to some overall factors. %Notice that we can fix all coefficients uniquely by virtue of gauge invariance of the partion function \eqref{eq:partition-function}. 

\paragraph{Quantum associativity.} The OPE between two higher-spin currents $\tJ$ up to first order in quantum correction reads
\begin{align}\label{eq:1-loop-OPE-summary}
    \tJ^{a_2}\Big[s_2&;2s_2-1\Big](z)\tJ^{a_3}\Big[s_3;2s_3-1\Big](0)\nn\\
    \sim&+\frac{1}{z}\sum_p\tg_p^{a_2a_3c}\frac{[v_2\,v_3]^{p}}{p!}\tJ^c[s_2+s_3-1-p;2s_2+2s_3-2]\nn\\
    &-\tau_{\tJ}\Big(\frac{1}{z^2}-\frac{1}{2z}\p_z\Big) \sum_{h\in\Spec}\cU^{a_1a_2a_3a_4}_{s_2,s_3}\sum_p\tg^{a_1a_4f}_p\frac{[v_2\,v_3]^p}{p!}\tJ^f[-1-p](z)\nn\\
    &+\tau_{\tJ}\frac{1}{z}\sum_{h\in\Spec}\cT^{a_1a_2a_3a_4}_{s_2,s_3}\nor \tJ^{a_1}[h]\tJ^{a_4}[-h](z)\nor\,.
\end{align}
\normalsize
where $\tau_{\tJ}$ is a yet-to-determined coefficient. It can be fixed uniquely by the associativity of the chiral algebras $\ca$. 

Observe that when $\Spec=\Z,2\Z+1,2\Z$, cf. \eqref{eq:quantum-protected-spectrum}, the single-$\tJ$ operator can be regularized to zero, leaving us with the non-vanishing double-$\tJ$ operator. Thus, the chiral algebra OPEs are, in general, not quantum protected in the presence of the defect. This is expected since the defect may break some part of the bulk higher-spin symmetry, cf. \eqref{eq:star-product}. %Nevertheless, we find that

\begin{theorem}\label{thm:quantum-associative} The chiral algebras associated with anomaly-free holomorphic twistorial higher-spin theories, i.e. theories with 
\begin{align}\label{eq:spec-diamond}
    \Spec^{\diamond}=\Z,2\Z+1,2\Z\,,
\end{align}
are associative to first order in quantum correction without additional input.
\end{theorem}
\begin{proof} We prove by a direct computation following along the lines in Appendix \ref{app:check-associativity}. Plugging \eqref{eq:1-loop-OPE-summary} in  \eqref{eq:to-be-checked} and take $n=1$, we find that the lhs. is
\begin{align}
    &-\oint dw\,w\tJ^{a_1}[h_1](w)\sum_{p}\tg_p^{a_2a_3c}\frac{[2\,3]^p}{p!}\tJ^c[s_2+s_3-1-p;2s_2+2s_3-2](z)\nn\\
    &-\frac{\tau_{\tJ}}{2}\oint dw\,w\tJ^{a_1}[h_1](w)\p_{z} \sum_{h\in\Spec}\cU^{a_ma_2a_3a_n}_{s_2,s_3}\sum_p\tg_p^{a_ma_nf}\frac{[2\,3]^{p}}{p!}\tJ^f[-1-p](z)\nn\\
    &-\tau_{\tJ}\oint dw\,w\tJ^{a_1}[h_1](w)\sum_{h\in\Spec}\cT^{a_ma_2a_3a_n}_{s_2,s_3}\nor \tJ^{a_m}[h]\tJ^{a_n}[-h](z)\nor\nn\\
    =&-\frac{\tau_{\tJ}}{2}\sum_{h\in\Spec}\cU_{s_2,s_3}^{a_ma_2a_3a_n}\sum_{p,q}\tg^{a_ma_nc}_p\tg^{a_1c\bullet}_q\frac{[2\,3]^p[1\,4]^q}{p!q!}\tJ^{\bullet}[h_1-2-p-q]\nn\\
    &+\tau_{\tJ}\sum_{h\in\Spec}\cT_{s_2,s_3}^{a_ma_2a_3a_n}\sum_{p,q}\tg_p^{a_na_1c}\tg_q^{ca_m\bullet}\frac{[2\,3]^p[1\,4]^q}{q!p!}\tJ^{\bullet}[h_1-2-p-q]\,,
    \end{align}
\normalsize
where we have subsequently used the classical OPEs \eqref{eq:tree-OPE-1} and the definition of the double-$\tJ$ operator, cf. \eqref{eq:normal-order-def}. 

On the other hand, the rhs. of \eqref{eq:to-be-checked} (with $n=1$) reads
\begin{align}
&\tau_{\tJ}\sum_{h\in\Spec}\sum_{p,q}\cU^{a_ma_2ca_n}_{s_2,s_3+h_1-1-p}\tg^{a_3a_1c}_p\tg_q^{a_ma_n\bullet}\frac{[3\,1]^p[2\,4]^{q}}{p!q!}\tJ^\bullet[h_1-2-p-q]\,.
\end{align}
Observe that there are overall sum over helicities on both side of the associativity condition \eqref{eq:to-be-checked} when $n=1$. Thus, for $\Spec=\Spec^{\diamond}$, these sum can be regularized to zero. As a result, the chiral algebras $\ca$ associated with theories with $\Spec=\Spec^{\diamond}$ are associative. Note that in this case $\tau_{\tJ}$ can be set to any real number. 
\end{proof}
It is intriguing to point out that only theories with $\Spec=\Spec^{\diamond}$ can admit higher-derivative interactions. All other cases will be forced by symmetry to have only Yang-Mills type interactions as shown in the next section.

%%%%%%%%%%%%%%%%%%%%%%%%%%%%%%%%%%%%%%%%%%%%%%%%%%%
\subsubsection{Enlarging chiral CFT with axionic currents}\label{sec:axion-currents}

%As noted in the beginning of Subsection~\ref{sec:algorithm}, the above OPE does not guarantee that $\ca$ is associative -- and indeed, it is not. 

Let us now study the cases where $\ca$ are associated with theories whose $\Spec\neq \Z,2\Z+1,2\Z$. In these cases, \eqref{eq:to-be-checked} generally do not hold; leading to the failure of associativity of $\ca$. Then, to restore associativity at one loop, we can introduce axionic currents, which are Koszul dual to the bulk axion field $\vartheta$. 
Recall that we have introduced axionic field $\vartheta$ subjected to the constraint $\p\vartheta=0$ to render holomorphic twistorial theories anomaly-free (on-shell) via Green-Schwarz anomaly cancellation mechanism in Subsection \ref{sec:anomaly-cancellation}.

\paragraph{Axionic currents.} As usual, there will be certain axionic interactions with the defect. (See e.g. \cite{Costello:2022upu,Bittleston:2022jeq} for previous work.) Note that even though we have only one type of axion bulk field $\vartheta$, or rather its source $\varrho$ where %(by virtue of holomorphic Poincar\'e lemma) %-- recall that $\p \vartheta=0$)
\begin{align}
    \vartheta=\p\varrho\,,\qquad \varrho\in\Omega^{1,1}(\PT,\cO(0))\,,
\end{align}
it, nevertheless, induces two new axionic higher-spin currents, say $\tU$ and $\tV$, on the defect. %This amounts to the enlargement of the chiral higher-spin algebras $\ca$ associated with anomalous twistorial higher-spin theories. 
This will be explained shortly below.
%Let us take a moment to explain why there can be two towers of higher-spin axionic currents associated with a single bulk field. Recall that the axion field $\vartheta\in \Omega^{2,1}(\PT,\cO(0))$ in \eqref{eq:axion-BVaction} subjects to the constraint $\p \vartheta=0$. Thus, locally, we can write
%by virtue of holomorphic Poincar\'e lemma. 
For now, we write \eqref{eq:axion-BVaction} as
\begin{align}\label{eq:axion-BVaction-improve}
    S[\varrho,\sA]=\int_{\PT}\varrho\,\bar{\p}\p\,\varrho-c_{\mg}\int_{S^7}\varrho\,\tr(\p\sA\star \p\sA)\,, \qquad \varrho\in \Omega^{1,1}(\PT,\cO(0))\,,
\end{align}
where we have done an integration by part to reach \eqref{eq:axion-BVaction-improve}. It is useful to remind ourselves that `$\tr$' stands for the trace of fundamental representations. The above action is invariant under
\begin{align}\label{eq:rho-variation}
    \delta\varrho^{1,1}=\p\xi^{0,1}+\bar{\p}\xi^{1,0}\,,
\end{align}
when $\sA\in \Omega^{0,1}(\PT)$ is on-shell, i.e. $\bar{\p}\sA\approx 0$.\footnote{The form degrees should be obvious from our notation.} 

To this end, let us consider the Koszul coupling
\begin{align}
    \tK_{\varrho}=\int_{\P^1}J_{\varrho}\,\varrho^{1,1}\,,
\end{align}
where $J_{\varrho}$ is some axionic current to be determined. 

Under the gauge transformation \eqref{eq:rho-variation}, we obtain the following non-trivial constraint 
\begin{align}
    \p\,J_{\varrho}=0\,,\qquad \p=dX^a\frac{\p}{X^{a}}\,,\quad a=1,2,3\,,
\end{align}
for the holomorphic current $J_{\varrho}$.\footnote{We remind the reader that we are working in the patch $\C^3\subset\PT$ with coordinates $X^a$, cf. \eqref{eq:C3-patch}.} Under the split $X^a=(z,w^{\dot\alpha})$, we can write
\begin{subequations}\label{eq:axion-derivative-structures}
    \begin{align}
        \p\,J_\varrho&=(\eth_z+\eth_w)J_{\rho}=0\,,\\
        \p\sA\wedge \p\sA&=\eth_z\sA\wedge \eth_w\sA+\eth_w\sA\wedge \eth_z\sA+\eth_w\sA\wedge \eth_w\sA\,,\\
        \p\sA\wedge\p\varrho&=\eth_z\sA\wedge \eth_w\varrho+\eth_w\sA\wedge \eth_z\varrho+\eth_w\sA\wedge \eth_w\varrho\,,
    \end{align}
\end{subequations}
where $\eth_z:=dz\p_z\,,\ \eth_w:=dw^{\dot\alpha}\p_{\dot\alpha}\,.$ 
%\begin{align}\label{eq:eth-def}
%    \eth_z:=dz\p_z\,,\qquad \eth_w:=dw^{\dot\alpha}\p_{\dot\alpha}\,.
%\end{align}
Observe that due to the difference in $SU(2)_-$ charges arising from having different derivative structures, it is convenient to introduce two currents $(\tU, \tV)$ to couple to derivatives of $\tA$, as mentioned.

\begin{figure}[ht!]
    \centering
    \includegraphics[scale=0.33]{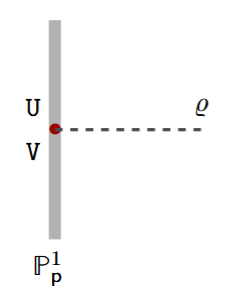}
    \caption{The Koszul couplings between a bulk axion $\varrho$ and two currents $\tU,\tV$.}
    %\label{fig:enter-label}
\end{figure}

\paragraph{The OPEs.} In the presence of the axion, the tree-level relation \eqref{eq:tree-amplitudes-OPE} gets modified to
\begin{align}\label{eq:tree-amplitudes-OPE-axion}
   \delta\Bigg( \parbox{45pt}{\includegraphics[scale=0.13]{2-point-precollide.png}}+\quad \parbox{40pt}{\includegraphics[scale=0.13]{cubic.png}}+\quad\parbox{40pt}{\includegraphics[scale=0.19]{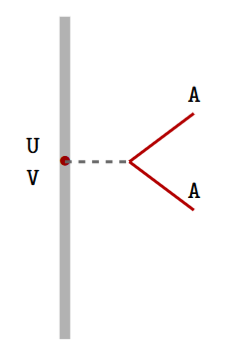}}\Bigg)=0\,,\qquad \delta \tA=\bar{\p}\cc\,.
\end{align}
Using \eqref{eq:axion-derivative-structures} and the cubic vertex in \eqref{eq:axion-BVaction-improve}, we obtain
\small
\begin{align}
    \parbox{42pt}{\includegraphics[scale=0.19]{UV-cubic.png}}=&-c_{\mg}\int \Big(\tU\big(\eth_z\tA_{2h_1-2}\star \eth_w\tA_{2h_2-2}+\eth_w\tA_{2h_1-2}\star\eth_z\tA_{2h_2-2})+\tV(\eth_w\tA_{2h_1-2}\star \eth_w\tA_{2h_2-2})\Big)\,,
\end{align}
\normalsize
where $c_{\mg}$ is given in \eqref{eq:axion-coupling-constant}. Then, feeding $\cc=z e^{-[w\,\tilde v]}$ and $\tA=e^{-[w\,\tilde v]}d\bar{z}$ 
into \eqref{eq:tree-amplitudes-OPE-axion}, we obtain
\begin{align}\label{eq:JJ-with-axion}
        \tJ^{a}\Big[&s_1;2s_1-1\Big](z)\tJ^{b}\Big[s_2;2s_2-1\Big](0)\sim\nn\\
        &-c_{\mg}\Big(\frac{1}{z^2}+\frac{1}{z}\p_z\Big)\sum_p\kappa^{ab}%\tg_p^{abc}
        \frac{[\tilde v_1\,\tilde v_2]^p}{p!}\tU[s_1+s_2-2-p;2s_1+2s_2-2]\nn\\
        &-\frac{c_{\mg}}{z}\sum_p \kappa^{ab}%\tg_p^{abc}
        \frac{[\tilde v_1\,\tilde v_2]^{p}}{p!}\tV[s_1+s_2-2-p;2s_1+2s_2-4]\,.
\end{align}
Here, $\kappa^{ab}$ is the Killing form, see \eqref{eq:Killing-form}. Note that the above should be understood as an OPE at first order in quantum correction. Indeed, the role of axion field in the Green-Schwarz anomaly cancellation is to produce tree-level exchange diagram to cancel one-loop diagram with gauge field on the external leg. See e.g. \cite{Tran:2025uad}. 
%To obtain the above OPE, we can feed in \eqref{eq:tree-amplitudes-OPE-axion} as the test functions
%\begin{align}
%    \cc=z e^{-[w\,\tilde v]}\,,\qquad \tA=e^{-[w\,\tilde v]}d\bar{z}\,.
%\end{align}
%Furthermore, we have formally assign the axionic $\tU,\tV$ some indices which are convenient for checking associativity in Appendix \ref{app:check-associativity}. In particular, one can simply replace $\tg^{abc}$ by the Killing form $\kappa^{ab}$ whenever one of the colored indices of $\tg^{abc}$ gets contracted with the axionic currents $\tU,\tV$. 

In summary, the complete OPE between two higher-spin currents $\tJ$ (with axionic currents included)
up to first order in quantum correction is
\begin{align}\label{eq:master-1-loop-OPE-main}
    \tJ^{a_2}\Big[s_2&;2s_2-1\Big](z)\tJ^{a_3}\Big[s_3;2s_3-1\Big](0)\sim\nn\\
    &+\frac{1}{z}\sum_p\tg_p^{a_2a_3c}\frac{[\tilde v_2\,\tilde v_3]^{p}}{p!}\tJ^c[s_2+s_3-1-p;2s_2+2s_3-2]\nn\\
    &-\tau_\tJ\Big(\frac{1}{z^2}-\frac{1}{2z}\p_z\Big) \sum_{h\in\Spec}\cU^{a_1a_2a_3a_4}_{s_2,s_3}\sum_p\tg^{a_1a_4f}_p\frac{[\tilde v_2\,\tilde v_3]^p}{p!}\tJ^f[-1-p](z)\nn\\
    &+\tau_\tJ\frac{1}{z}\sum_{h\in\Spec}\cT^{a_1a_2a_3a_4}_{s_2,s_3}\nor \tJ^{a_1}[h]\tJ^{a_4}[-h](z)\nor\nn\\
    &-c_{\mg}\Big(\frac{1}{z^2}+\frac{1}{z}\p_z\Big)\sum_p\kappa^{a_2a_3} %\tg_p^{a_2a_3c}
    \frac{[\tilde v_2\,\tilde v_3]^{p}}{p!}\tU[s_2+s_3-2-p;2s_2+2s_3-2]\nn\\
        &-\frac{c_{\mg}}{z}\sum_p\kappa^{a_2a_3}%\tg_p^{a_2a_3c}
        \frac{[\tilde v_2\,\tilde v_3]^{p}}{p!}\tV[s_2+s_3-2-p;2s_2+2s_3-4]\,.
\end{align}
\normalsize
Our next stop is the computation of the OPE between the physical higher-spin currents $\tJ$ and the axionic higher-spin currents $\tU,\tV$. Consider the process,
\begin{align}\label{eq:tree-OPE-J-axion}
   \delta\Bigg( \parbox{45pt}{\includegraphics[scale=0.22]{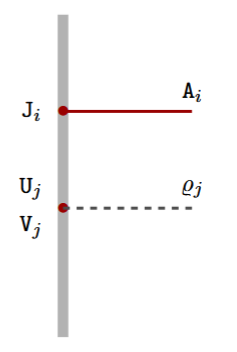}}+\quad\parbox{42pt}{\includegraphics[scale=0.22]{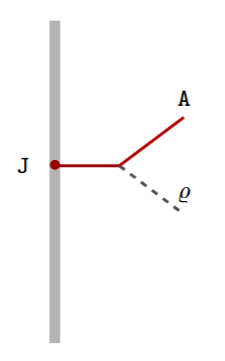}}\Bigg)=0\,,\qquad \delta \tA=\bar{\p}\cc\,.
\end{align}
Similar to the previous case, we can do an integration by part and use \eqref{eq:axion-derivative-structures} to arrive at 
\begin{subequations}\label{eq:J-axion}
    \begin{align}
        \tJ^{a}\big[h_i\big](z)\tV[0](0)&\sim -c_{\mg}\Big(\frac{1}{z^2}+\frac{1}{z}\p_z\Big)\sum_p%\tg_p^{abc}
        \frac{[\tilde v_i\,\tilde v_j]^p}{p!}\tJ^a[h_i-2-p]\,,\\
        \tJ^{a}\big[h_i;\tH_i\big](z)\tU[0](0)&\sim -\frac{c_{\mg}}{z}\sum_{p}%\tg^{abc}_p
        \frac{[\tilde v_i\,\tilde v_j]^{p}}{p!}\tJ^a[h_i-2-p;\tH_i-2]\,.
    \end{align}
\end{subequations}
%Then, the central result of this work is the following theorem on quantum associativity of chiral higher-spin algebra $\ca$:
\begin{theorem}\label{thm:quantum-HS} The chiral higher-spin algebras associated with anomalous holomorphic twistorial higher-spin theory with $\Spec\neq \Spec^{\diamond}$ are associative iff the interactions are of Yang-Mills type. Furthermore, all positive external helicities of the current appearing in first-order OPE should be $+1$ and $\tau_{\tJ}=\frac{i}{2^{10}}$\,.
\end{theorem}
\begin{proof} The proof is relegated to Appendix \ref{app:check-associativity}.
\end{proof}
This completes our study of the chiral higher-spin algebra $\ca$ up to first order in quantum correction. %

%It is useful recalling that $1/z$ has conformal weight $+1$, and $d\bar{z}$ has conformal weight $-1$. Meanwhile, the plane-wave factor has conformal weight zero. 

%%%%%%%%%%%%%%%%%%%%%%%%%%%%%%%%%%%%%%%%%%%%%%%%%%%%
\section{Higher-spin correlation functions and form factors}\label{sec:4}

Having defined the OPEs of the chiral algebra $\ca$ up to first order in quantum correction, let us now study some simple correlation functions 
\begin{align}
    \cAmp_n:=\langle
    \tJ_1\ldots \tJ_n\rangle\,,
\end{align}
of the currents that generate chiral higher-spin algebras following the approach of \cite{Costello:2022wso,Costello:2022upu}. Here, $\tJ_m\equiv \tJ[m]\equiv \tJ_m[h_m,\tH_m](z_m)$ are holomorphic currents on the defect. Note that $\cAmp_n$ can be identified with $4d$ form factors in some spacetime theories, which we do not need to have explicit spacetime actions, a priori. Our procedure up to this point can be summarized as
\begin{figure}[ht!]
    \centering    \includegraphics[scale=0.55]{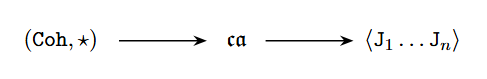}
    \caption{Starting with the non-commutative algebra $(\Coh,\star)$, we have constructed $\ca$ via Koszul duality as in \cite{Costello:2022wso}. After the OPE between holomorphic currents are determined, we can now bootstrap the correlation functions $\langle \tJ\ldots\tJ\rangle$ up to some number of loops. These correlators can be interpreted as form factors in some $4d$ spacetime theories. }
\end{figure}\\
For simplicity, we will work with matrix-valued currents, which allows us to compute form factors in terms of trace invariants of $N\times N$ matrices without relying heavily on the structure constants $\tg^{abc}$. 

As far as it concerns, the chiral bootstrapping technique in \cite{Costello:2022wso,Costello:2022upu} is well-suited for computing non-trivial scattering amplitudes in theories that are small deformation away from chiral or self-dual theories, with Yang-Mills theory being a prime example \cite{Chalmers:1998jb,Mason:2005zm}. Remarkably, one does not need to know the spacetime theories in order to compute form factors if the OPEs between $\tJ$ are given from the outset. This is an advantage of the chiral bootstrap program. Furthermore, if the form factors turn out to be rational, the bulk theories are expected to exhibit strong integrability, which implies that some form factors are, in fact, honest scattering amplitudes. We emphasize, however, that this is only well justified for theories with Yang-Mills-like interactions.

% We will comment further on this point later in this section.

%%%%%%%%%%%%%%%%%%%%%%%%%%%%%%%%%%%%%%%%%%%%%%%%%%%%
\subsection{On \texorpdfstring{$SU(2)_-$}{SU}-invariant OPEs}

%%%%%%%%%%%%%%%%%%%%%%%%%%%%%%%%%%%%%%%%%%%%%%%%
Note that as we can write 
\begin{align}
    \tJ\Big[s_i;m_i\Big]=\frac{\tilde v_i^{\dot\alpha(m_i)}}{m_i!}\tJ_{\dot\alpha(m_i)}[s_i]\,,\qquad \tilde v^{\dot\alpha(m)}\equiv \tilde{v}^{\dot\alpha_i}\ldots \tilde{v}^{\dot\alpha_{m}}\,,
\end{align}
we can also impose $SU(2)_-$ invariant condition on the OPE between higher-spin currents. In practice, this means that we can use the $\msu(2)$-invariant matrices $\epsilon^{\dot\alpha\dot\beta}$ to contract the external spinors $\tilde v_i$ associated with the current $\tJ_i$ as
\begin{align*}
    \tJ^{a_2}\Big[s_2&;m_2\Big](z)\tJ^{a_3}\Big[s_3;m_3\Big](0)%\nn\\
    =\frac{[\tilde v_2\,\tilde v_3]^{m_3}}{m_3!}\frac{(m_2-m_3)!}{m_2!}\tJ^{a_2}\Big[s_2;m_2-m_3\Big](z)\tJ^{a_3}\Big[s_3;0\Big](0)\,.
\end{align*}
Here, we assume that $m_2\geq m_3$. Then, on the ground of $SU(2)_-$ invariance, we can replace \eqref{eq:master-1-loop-OPE-main} with ($\tau_\tJ=i/2^{10}$)
\small
\begin{align}\label{eq:JJ-SU(2)}
    \tJ^{a_2}\Big[s_2&;2s_2-1\Big](z)\tJ^{a_3}\Big[s_3;2s_3-1\Big](0)\sim\nn\\
    &+\frac{1}{z}\sum_p\tg_p^{a_2a_3c}\frac{[v_2\,v_3]^{p+2s_3-1}}{p!(2s_3-1)!}\tJ^c[s_2+s_3-1-p;2s_2-2s_3]\nn\\
    &-\tau_\tJ\Big(\frac{1}{z^2}-\frac{1}{2z}\p_z\Big) \sum_{h\in\Spec}\cU^{a_1a_2a_3a_4}_{s_2,s_3}\sum_p\tg^{a_1a_4f}_p\frac{[2\,3]^{p+2s_3-1}}{p!(2s_3-1)!}\tJ^f[-1-p](z)\nn\\
    &+\frac{\tau_{\tJ}}{z}\sum_{h\in\Spec}\cT^{a_1a_2a_3a_4}_{s_2,s_3}\frac{[2\,3]^{2s_3-1}}{(2s_3-1)!}\nor \tJ^{a_1}[h]\tJ^{a_4}[-h](z)\nor\,\nn\\
    &-c_{\mg}\Big(\frac{1}{z^2}+\frac{1}{z}\p_z\Big)\sum_p\kappa^{a_2a_3}\frac{[2\,3]^{p+2s_3-1}}{p!(2s_3-1)!}\tU[s_2+s_3-2-p;2s_2-2s_3]\nn\\
    &-\frac{c_{\mg}}{z}\sum_p \kappa^{a_2a_3}\frac{[2\,3]^{p+2s_3-1}}{p!(2s_3-1)!}\tV[s_2+s_3-4-p;2s_2-2s_3]\,.
\end{align}
\normalsize
%Similarly, we also have
%\begin{align}\label{eq:JJ-with-axion-1}
%        \tJ^{a}\Big[&s_1;2s_1-1\Big](z)\tJ^{b}\Big[s_2;2s_2-1\Big](0)\sim\nn\\
 %       &-c_{\mg}\Big(\frac{1}{z^2}+\frac{1}{z}\p_z\Big)\sum_p\tg_p^{abc}\frac{[v_1\,v_2]^{p+2s_2-1}}{p!(2s_2-1)!}\tU^c[s_1+s_2-1-p;2s_1-2s_2]\nn\\
 %       &+\frac{c_{\mg}}{z}\sum_p\tg_p^{abc}\frac{[v_1\,v_2]^{p+2s_2-1}}{p!(2s_2-1)!}\tV_{c}[s_1+s_2-1-p;2s_1-2s_2]\,.
%\end{align}
%From the above, it is easy to notice that if an external higher-spin current has a non-trivial $SU(2)_-$ charge, it typicaly implies that the form factors will contain loop contributions. 
We note that the double poles $\frac{1}{z^2}$ in the above OPEs can be identified with the double poles usually appear in the one- or two-loop integrands of Yang-Mills or QCD theories \cite{Bern:2005ji}. 

%%%%%%%%%%%%%%%%%%%%%%%%%%%%%%%%%%%%%%%%%%
\subsection{Tree-level amplitudes}
Let us now compute some simple tree-level amplitudes. %, which should be some functions of $(\{z_i\},\{v_i\},\{h_i\})$. 
Since we can identify $\tilde v^{\dot\alpha}$ with the spacetime spatial momentum in the light-cone gauge, cf.
\cite{Bengtsson:1983pd,Metsaev:2005ar}, any scattering amplitude expressions with negative powers in $[\tilde v_i\, \tilde v_j]$, where $i\neq j$, will be diagnosed as non-local. As a result, we can make quite explicit statements about whether a form factor, or, equivalently, a chiral CFT correlation function is healthy by simply looking at the power of the square brackets $[\tilde v_i\,\tilde v_j]$.

%%%%%%%%%%%%%%%%%%%%%%%%%%%%%%%%%%%%%%%%%%
\paragraph{2-point functions.} As in any usual CFTs, the chiral higher-spin symmetry can also fix the structure of the two-point functions uniquely up to a normalization constant, which we simply set to 1 for simplicity. Due to symmetry
\begin{align}\label{eq:2pt-normalized}
    \langle \tJ^a_1[h_1;s]\tJ^b_2[h_2;s]\rangle=\delta_{h_1,h_2}\frac{\kappa^{ab}}{z_{12}^{2h_1}}\frac{[1\,2]^s}{s!}\,,
\end{align}
where
\begin{align}\label{eq:Killing-form}
    \kappa^{ab}:=\Tr(T^aT^b)
\end{align}
denotes the usual Killing bilinear form. Then, all other higher-point amplitudes can then be computed by doing Wick contractions, which is a purely algebraic process. 
%%%%%%%%%%%%%%%%%%%%%%%%%%%%%%%%%%%%%%%%%%
\paragraph{3-point functions.} Now, consider the 3-point functions:
\begin{align}\label{eq:3-point}
    \langle \tJ_1[h_1]\tJ_2[h_2]\tJ_3[h_3]\rangle =&+\frac{1}{z_{23}} \sum_p\frac{[2\,3]^p}{p!}\langle \tJ_1[h_1]\tJ_{2}[h_2+h_3-1-p]\rangle\nn\\
    &+\frac{1}{z_{31}}\sum_q\frac{[3\,1]^q}{q!}\langle \tJ_1[h_3+h_1-1-q]\tJ_2[h_2]\rangle\,.
\end{align}
Using \eqref{eq:2pt-normalized}, we obtain the following general amplitudes, which are constrained by $\ca$'s symmetries:
\begin{align}\label{eq:3point-step-1}
    \cAmp(h_1,h_2,h_3)=\frac{1}{z_{23}}\frac{[2\,3]^{h_2+h_3-h_1-1}}{\Gamma(h_2+h_3-h_1)}z_{12}^{-2h_1}+\frac{1}{z_{31}}\frac{[3\,1]^{h_3+h_1-h_2-1}}{\Gamma(h_3+h_1-h_2)}z_{12}^{-2h_2}\,.
\end{align}
Requiring the power of all square brackets to be non-negative by virtue of locality, i.e.
\begin{align}
    h_2+h_3-h_1\geq 1\,,\qquad  h_3+h_1-h_2\geq 1\,,
\end{align}
we obtain the following constraints
\begin{align}
    h_1=h_2=-s\quad(s\geq 0)\,,\qquad h_3=1\,,
\end{align}
for the above 3-point amplitudes to be non-vanishing. Remarkably, this turns out to be precisely the 3-point amplitudes of HS-YM theory in \cite{Adamo:2022lah}, where
\begin{align}\label{eq:MHV-HS-YM}
    \cAmp(-s,-s,1)=z_{12}^{2s}\Big(\frac{1}{z_{23}}+\frac{1}{z_{31}}\Big)=-\frac{z_{12}^{2s+2}}{z_{12}z_{23}z_{31}}\,.
\end{align}
Upon identifying $z_{ij}=\langle i\,j\rangle$, the above yields
\begin{align}
    \cAmp(-s,-s,1)=-\frac{\langle 1\,2\rangle^{2s+2}}{\langle 1\,2\rangle \langle 2\,3\rangle \langle 3\,1\rangle}\,.
\end{align}
Observe that, for $s=1$, we recover the 3-point MHV gluon amplitudes.\footnote{Previous work on computing gluon scattering amplitudes using twistor string theory can be found in e.g. \cite{Nair:1988bq,Witten:2003nn,Cachazo:2004kj}.} 

It is hard not to notice that the negative-helicity fields behave quite differently with positive helicity fields in chiral/self-dual higher-spin theories. In particular, they play the roles of linear fluctuations around the chiral/self-dual background set by the positive-helicity fields \cite{Tran:2025yzd}. Namely, they do not play important roles in deforming the background, and, thus are less constrained at asymptotic infinity. %For this reason, there can be non-trivial amplitudes with negative-helicity higher-spin fields. We observe that this pattern persists at higher-point amplitudes.

Note that we do not recover the usual MHV 3-point amplitude of gravity
\begin{align}
    \cA_3(1_{-2},2_{-2},3_{+2})=\frac{\langle 1\,2\rangle^6}{\langle 1\,2\rangle^2\langle 2\,3\rangle^2\langle 3\,1\rangle^2}\,.
\end{align}
This stems from the fact that the classical OPE between higher-spin currents contain only simple poles, and the fields we used to construct $\ca$ via Koszul duality are chiral field representations \cite{Krasnov:2021nsq,Adamo:2022lah}.\footnote{Note, however, that it is possible to derive gravity amplitudes using twistor framework, see e.g. \cite{Hodges:2012ym}.} 

%%%%%%%%%%%%%%%%%%%%%%%%%%%%%%%%%%%%%%%%%%
\paragraph{Comment on higher-derivative interactions.} Note that the above computation of form factors is mainly sensible for Yang-Mills-like interacting theories, since the classical OPE \eqref{eq:full-tree-OPE} can be viewed as the holomorphic collinear limit of the higher-spin soft factors for numerous chiral higher-spin theories \cite{Tran:2022amg}. As a result, in reconstructing the tree-level amplitudes from these soft factors, one may need to adjust the Wick contraction rule by including some factors of $\frac{\langle \alpha \,i\rangle}{\langle j\,\alpha\rangle}$ where $\alpha$ are some reference spinors, see \cite{Boucher-Veronneau:2011rwd}. 

Although this argument may also apply to loop level, the results of Theorem \ref{thm:quantum-associative} and \ref{thm:quantum-HS} do, in fact, remove this possibility. Therefore, we only need to worry about higher-derivative interactions at classical level, which was in agreement with the finding of \cite{Ball:2021tmb}. In what follows, we will simply ignore the case of higher-derivative interactions and leave the investigation for a future work.

%%%%%%%%%%%%%%%%%%%%%%%%%%%%%%%%%%%%%%%%%%
\paragraph{4-point functions.} Let us now bootstrap higher-point correlation functions to check whether there can be actually non-trivial higher-spin amplitudes with complex kinematics, as having been seen previously in the case of HS-YM \cite{Adamo:2022lah}. Following the strategy outlined in \cite{Costello:2022wso}, we start to do Wick contractions with the closest neighbors, say $\tJ_{i-1}$ and $\tJ_{i+1}$, of $\tJ_i$ in the string of operators $\langle \tJ_1\ldots \tJ_n\rangle$. %This is possible due to the fact that we are computing color-ordered amplitudes. 
For instance, %the next simplest example, which is the 4-point function
\begin{align}\label{eq:4-pt-step-1}
    \langle \tJ_1[h_1]\tJ_2[h_2]\tJ_3[h_3]\tJ_4[h_4]\rangle\,,
\end{align}
is equal to
\begin{align}
    \eqref{eq:4-pt-step-1}=&+\sum_p\frac{[4\,1]^p[3\,1]^{h_1-h_2+h_3+h_4-2-p}}{p!(h_1-h_2+h_3+h_4-2-p)!}\frac{1}{z_{41}z_{31}z_{12}^{2h_2}}\nn\\
    &+\sum_p\frac{[4\,1]^p[2\,3]^{h_2+h_3-h_1-h_4+p}}{p!(h_2+h_3-h_1-h_4+p)!}\frac{1}{z_{41}z_{23}z_{12}^{2(h_1+h_4-p-1)}}\nn\\
    &+\sum_p\frac{[3\,4]^p[3\,1]^{h_1-h_2+h_3+h_4-2-p}}{p!(h_1-h_2+h_3+h_4-2-p)!}\frac{1}{z_{34}z_{31}z_{12}^{2h_2}}\nn\\
    &+\sum_p\frac{[3\,4]^p[2\,3]^{-h_1+h_2+h_3+h_4-2-p}}{p!(-h_1+h_2+h_3+h_4-2-p)!}\frac{1}{z_{34}z_{23}z_{12}^{2h_1}}\,.
\end{align}
Notice that the above sums have somewhat similar pattern with the ones appear in the context of tree-level amplitudes of chiral higher-spin gravity \cite{Skvortsov:2018jea,Skvortsov:2020wtf}. Evaluating the sum, we obtain 
\begin{align}
    \eqref{eq:4-pt-step-1}=&+\frac{\big([4\,1]+[3\,1]\big)^{h_1-h_2+h_3+h_4-2}}{(h_1-h_2+h_3+h_4-2)!}\frac{1}{z_{41}z_{31}z_{12}^{2h_2}}\nn\\
    &+\frac{[4\,1]^{h_1+h_4-h_2-h_3}}{(h_1+h_4-h_2-h_3)!}\frac{1}{z_{41}z_{23}z_{12}^{2(h_2+h_3-1)}}\nn\\
    &+\frac{\big([3\,4]+[3\,1]\big)^{h_1-h_2+h_3+h_4-2}}{(h_1-h_2+h_3+h_4-2)!}\frac{1}{z_{34}z_{31}z_{12}^{2h_2}}\nn\\
    &+\frac{\big([3\,4]+[2\,3]\big)^{-h_1+h_2+h_3+h_4-2}}{(-h_1+h_2+h_3+h_4-2)!}\frac{1}{z_{34}z_{23}z_{12}^{2h_1}}\,.
\end{align}
For the amplitudes to be non-trivial, the total power of square brackets should be non-negative. We obtain the constraints
%\begin{subequations}
%    \begin{align}
%        h_1-h_2+h_3+h_4&\geq 2\,,\\
%        h_1+h_4&\geq h_2+h_3\,,\\
%        -h_1+h_2+h_3+h_4&\geq 2\,.
%    \end{align}
%\end{subequations}
%From these inequalities, we deduce that
\begin{align}
     h_3,h_4\geq 1\,,\qquad h_1=h_2\,.
\end{align}
Thus, \eqref{eq:4-pt-step-1} reduces to
\begin{align}
    \eqref{eq:4-pt-step-1}=&+\frac{\big([4\,1]+[3\,1]\big)^{h_3+h_4-2}}{(h_3+h_4-2)!}\frac{1}{z_{41}z_{31}z_{12}^{2h_1}}\nn\\
    &+\frac{[4\,1]^{h_4-h_3}}{(h_4-h_3)!}\frac{1}{z_{41}z_{23}z_{12}^{2(h_1+h_3-1)}}\nn\\
    &+\frac{\big([3\,4]+[3\,1]\big)^{h_3+h_4-2}}{(h_3+h_4-2)!}\frac{1}{z_{34}z_{31}z_{12}^{2h_1}}\nn\\
    &+\frac{\big([3\,4]+[2\,3]\big)^{h_3+h_4-2}}{(h_3+h_4-2)!}\frac{1}{z_{34}z_{23}z_{12}^{2h_1}}\,.
\end{align}
Thus, as long as $h_4\geq h_3\geq 1$, the above may be a well-defined amplitude. However, as stated, we do not fully understand the case of higher-derivative interactions. Thus, we will only focus on the the helicity bound where $h_3=h_4=1$. We find
\begin{align}
    \eqref{eq:4-pt-step-1}=\frac{1}{z_{12}^{2h_1}}\Big(\frac{1}{z_{41}z_{31}}+\frac{1}{z_{41}z_{23}}+\frac{1}{z_{34}z_{31}}+\frac{1}{z_{34}z_{23}}\Big)=\frac{1}{z_{12}^{2h_1-2}}\frac{1}{z_{12}z_{23}z_{34}z_{41}}\,.
\end{align}
In terms of angled brackets, 
\begin{align}
    \langle \tJ_1[h_1]\tJ_2[h_1]\tJ_3[1]\tJ_4[1]\rangle=\frac{\langle 1\,2\rangle^{-2h_1+2}}{\langle 1\,2\rangle\langle 2\,3\rangle\langle 3\,4\rangle\langle 4\,1\rangle}\,.
\end{align}
When $h_1=-s$ for $s\geq 0$, we recover the 4-point MHV amplitude in HS-YM theory \cite{Adamo:2022lah}. Of course, the $n$-point MHV amplitude of HS-YM can be obtained in an inductive way, but it is not the point.

The point is that the chiral higher-spin symmetry algebras $\ca$ allows for more non-trivial higher-spin amplitudes than one might expect. This suggests the possible existence of some mysterious higher-spin theories that are not yet constructed. Nevertheless, it is important to note that these amplitudes should be understood as arising from non-unitary higher-spin theories, which are only well-defined in Euclidean, split-signature spacetime, or complexified $4d$ spacetime. These theories may be understood as certain higher-spin and higher-derivative extension of the usual unitary Yang-Mills theory. (See also a recent study in the light-cone gauge \cite{Serrani:2025owx}, which shows somewhat similar conclusions along this direction.\footnote{Since the light-cone deals directly with physical degrees of freedom, there is no ambiguity of gauge redundancy. Therefore, if an amplitude exists in the light-cone gauge, its covariant description must also exist. However, the covariant expression may be more complicated to find.}) Note that when suitable matter fields are included in this framework, one may also obtain certain higher-spin extensions of QCD. 
%%%%%%%%%%%%%%%%%%%%%%%%%%%%%%%%%%%%%%%%%%

%%%%%%%%%%%%%%%%%%%%%%%%%%%%%%%%%%%%%%%%%%
\subsection{Loop amplitudes and rational sector}
%The fact that the correlation functions of the holomorphic currents on the celestial twistor sphere can encode information of the form factors in some $4d$ spacetime theories that are dual to the bulk twistor theories is known long time ago, cf. \cite{Nair:1988bq} and \cite{Witten:2003nn}. Nonetheless, 
For chiral/self-dual theories with strong integrable properties, %there are some interesting statements one can make. In particular, 
one may tentatively suggest that if an $\ell$‑loop form factor with a given helicity configuration happens to be finite and rational, while the corresponding form factor (with the same helicity) at one lower order in quantum correction vanishes, then it is reasonable to be viewed as a \emph{genuine amplitude} at $\ell$ loop. These kind of form factors belong to what we call \emph{rational sector}.
 
%The bootstrap program using chiral OPE in \cite{Costello:2022wso,Costello:2022upu} computes form factors (loop integrands) of $4d$ spacetime theories -- the dual theories of some specific anomaly-free twistorial theories. 

\begin{table}[ht!]
    \centering
    \begin{tblr}{
    colspec = {|c|c|c|c|c|},
    cell{2}{2} = {green!60!white!30},
    cell{2}{3} = {green!60!white!30},
    cell{2}{4} = {green!60!white!30},
    cell{3}{2} = {green!60!white!30},
    cell{3}{3} = {green!60!white!30},
    cell{4}{2} = {green!60!white!30},
    cell{5}{2} = {green!60!white!30},
  }\hline
     Amplitude    &  Tree  & 1 loop & 2 loop & higher loop \\\hline\hline
       $(+,+,+\ldots,+)$  & vanishing  & vanishing & rational &  divergent \\\hline
       $(-,+,+\ldots,+)$ & vanishing & rational & divergent & divergent \\\hline
       $(-,-,+,\ldots,+)$ & rational & divergent & divergent & divergent \\\hline 
       more $-$ &  rational &  divergent &  divergent & divergent \\\hline
    \end{tblr}
    \caption{The form factors which one can interpret as genuine amplitudes, using the OPEs up to first order in quantum correction, are highlighted in [\textcolor{green!60!black!90}{\bf green}]. They form what we will call the [\textcolor{green!60!black!90}{\bf rational sector}]. Note that the rationality of the loop amplitudes arises from the fact that the loop or tree-level amplitudes at one order lower in quantum corrections vanish. All other form factors require doing explicit loop integrals to become actual amplitudes, and they may to exhibit both UV and IR divergences. Note that these amplitudes are resulted from certain deformation away from the chiral/self-dual sectors.}
    \label{tab:form factors}
\end{table} 

For various holomorphic twistorial theories, one can show that their one-loop all-plus helicity amplitudes are trivial once appropriate couplings are introduced, cf. \cite{Costello:2021bah,Bittleston:2022nfr,Tran:2025uad}. As a result, the chiral OPE data we have derived can be used to compute two-loop all-plus amplitudes, which are \emph{rational}. In fact, it was shown by direct computation in \cite{Dixon:2024mzh} that the two-loop all-plus helicity amplitudes in QCD with special matter content do not exhibit divergence after doing a suitable IR subtraction.%In particular, the IR divergence found in \cite{Dixon:2024mzh} may be viewed as a renormalization scheme artifact if one works with dimensional regularization. (See \cite{Catani:1998bh} for a universal factorization formula for dimensional regularization, which may work to all order in perturbation theory.) It was shown later in  \cite{Dixon:2024mzh} that if one uses a suitable mass regularization, the IR divergence in QCD with special matter content indeed disappears.} 
\footnote{This story, of course, is slightly different with the usual Yang-Mills or QCD two-loop amplitudes. In particular, it is well-known that all two-loop amplitudes of YM or QCD exhibit both UV and IR divergences \cite{Badger:2013gxa,Dunbar:2017nfy,Badger:2018enw} because the one-loop amplitudes of YM/QCD are either rational (all-plus helicity sector \cite{Mahlon:1993fe,Mahlon:1993si,Bern:1993qk} or amplitudes with at least one negative helicitiy \cite{Bern:2005ji}), or divergent. }

%In contrast, for theories with trivial one-loop amplitudes -- such as chiral higher-spin gravity and its subsectors \cite{Ponomarev:2017nrr}, certain the self-dual theories with axion fields \cite{Costello:2021bah,Bittleston:2022nfr}, or theories with specifically tuned matter content\footnote{$SU(2)$ QCD with $N_f=8$; or $SU(3)$ QCD with $N_f=9$.} \cite{Costello:2023vyy} -- it is possible to bootstrap the two-loop form factors using chiral OPE data. Remarkably, as the all-plus two-loop form factors belong to the integrable sector, see Table \ref{tab:form factors}, they can be viewed as genuine two-loop scattering amplitudes 
%%%%%%%%%%%%%%%%%%%%%%%%%%%%%%%%%%%%%%%%%%

\subsubsection{One-loop amplitudes} 
Let us now study the one-loop $n$-point amplitudes in the rational sector, cf. Table \ref{tab:form factors}, via correlation functions of $\tJ[h;k]$. This involves the insertion of $2$ higher-spin currents $\tJ[h;k]$ with non-trivial $SU(2)_-$ charges and $n-2$ other classical higher-spin currents. Due to symmetry, it is easy to notice that all $SU(2)_-$ charges must be the same. We also note that the form factors computed in this section can be identified with the one-loop amplitudes with specific helicity configurations, such as $(+,+,\ldots,+)$ and $(-,+,\ldots,+)$. Since the UV and IR divergences of the one-loop amplitudes are proportional to the corresponding tree-level amplitudes with the same helicity configurations, which vanish in these cases, the one-loop amplitudes are UV and IR finite, as well as being rational.

%In computing $4d$ form factors or correlation functions of $\ca$, we observe that, at classical level, the correlation functions of $\ca$ or tree-level higher-spin amplitudes are not always vanishing. Nevertheless, this phenomenon does not persist at loop levels. 

We find that correlation functions of currents that generate $\ca$ associated with anomaly-free twistorial higher-spin theories are always vanishing. However,  for theories whose $\Spec\neq \Z,2\Z+1$, the form factors are non-trivial, and have similar expression with the one of mostly-plus Yang-Mills amplitudes in \cite{Costello:2022upu,Costello:2023vyy}. %Note also that even though the anomaly cancellation via Okubo's relations, cf. \eqref{eq:Okubo}, requires specific values of $N$ in the case where $\ca$ contains axionic currents, we will pretend that number $N$ is ``large'' and restrict our attention to the planar sector. 

%%%%%%%%%%%%%%%%%%%%%%%%%%%%%%%%%%%%%%%%%%%%%%%%%
\paragraph{One-loop 3-point amplitudes.} Before studying 4-point loop amplitudes. Let us make some simple statements, which proves to be useful in the follows.
\begin{lemma}\label{lem:quantum-protect} The 3-point scattering amplitudes 
\begin{align}\label{eq:3pt-1-loop-JJJ}
    \big\langle \tJ_1[h_1]\,\tJ_2[s;2s-1]\,\tJ_3[s;2s-1]\big\rangle\,,\qquad s\geq 1\,,
\end{align}
are quantum protected at one loop.
\end{lemma}
\begin{proof} We prove by direct computation. By virtue of Theorem \ref{thm:quantum-HS}, we first consider the case where $\Spec^{\diamond}=\Z,2\Z+1$. In this case, the OPE \eqref{eq:JJ-SU(2)} with $s_2=s_3=s$ reduces to
\begin{align}\label{eq:JJ-same-s}
    \tJ^{a_2}\Big[s;2s-&1\Big](z)\tJ^{a_3}\Big[s;2s-1\Big](0)\sim\nn\\
    &+\frac{1}{z}\sum_pg^{a_2a_3c}_p\frac{[v_2\,v_3]^{p+2s-1}}{p!(2s-1)!}\tJ^c[2s-1;0]\nn\\
    %&-\tau_\tJ\Big(\frac{1}{z^2}-\frac{1}{2z}\p_z\Big) \sum_{h\in\Spec^{\diamond}}\cU^{a_1a_2a_3a_4}_{s}f^{a_1a_4f}\frac{[2\,3]^{2s-1}}{(2s-1)!}\tJ^f[-1](z)\nn\\
    &+\frac{\tau_\tJ}{z}\sum_{h\in\Spec^{\diamond}}\cT^{a_1a_2a_3a_4}_{s}\frac{[2\,3]^{2s-1}}{(2s-1)!}\nor \tJ^{a_1}[h]\tJ^{a_4}[-h](z)\nor\,\,.
\end{align}
We find
\begin{align}
    \eqref{eq:3pt-1-loop-JJJ}= \frac{1}{z_{31}}[3\,1]\big\langle \tJ_1[h_1+s-1;2s-1]\tJ_2[s;2s-1]\big\rangle\,,
\end{align}
where
\begin{align}
    \sum_{h\in\Spec}\langle \tJ_1[h_1]\nor\tJ[h]\tJ[-h](z)\nor\rangle =0\,,
\end{align}
since the classical OPEs lead to a contour integral with double pole. Using \eqref{eq:2pt-normalized}, we obtain
\begin{align}
    \eqref{eq:3pt-1-loop-JJJ}=\delta_{h_1,+1}\frac{[1\,2]^{2s-1}z_{23}}{(2s-1)!}\frac{1}{z_{12}^{2s}z_{23}z_{31}}\,,
\end{align}
This vanishes on the support of momentum conservation. Thus, the 3-point amplitude \eqref{eq:3pt-1-loop-JJJ} is protected from quantum correction if $\Spec=\Spec^{\diamond}$. 

In the case where $\Spec=|h|\geq 1$, (so that $\sum_{h\in\Spec}1=-1$) we can impose the external helicities of the currents which enter quantum process to be $+1$ from the outset and the interactions are of Yang-Mills type. Then, proceed similarly with along the line above, we again see that the 3-point amplitude \eqref{eq:3pt-1-loop-JJJ} is also protected from quantum correction in this case. Indeed, recall that
\begin{align}\label{eq:JJ-same-s-axion}
    \tJ^{a_2}\Big[s;2s-&1\Big](z)\tJ^{a_3}\Big[s;2s-1\Big](0)\sim\nn\\
    &+\frac{1}{z}f^{a_2a_3c}\frac{[v_2\,v_3]^{2s-1}}{(2s-1)!}\tJ^c[2s-1;0]\nn\\
    &+\tau_\tJ\Big(\frac{1}{z^2}-\frac{1}{2z}\p_z\Big) \cU^{a_1a_2a_3a_4}_{s}f^{a_1a_4f}\frac{[2\,3]^{2s-1}}{(2s-1)!}\tJ^f[-1](z)\nn\\
    &-\frac{\tau_\tJ}{z}\cT^{a_1a_2a_3a_4}_{s}\frac{[2\,3]^{2s-1}}{(2s-1)!}\nor \tJ^{a_1}[h]\tJ^{a_4}[-h](z)\nor\,\nn\\
    &-c_{\mg}\Big(\frac{1}{z^2}+\frac{1}{z}\p_z\Big)\kappa^{a_2a_3}\frac{[2\,3]^{2s-1}}{(2s-1)!}\tU[2s-1;0]\nn\\
    &-\frac{c_{\mg}}{z}\kappa^{a_2a_3}\frac{[2\,3]^{2s-1}}{(2s-1)!}\tV[2s-1;0]\,,
\end{align}
we find
\begin{align}\label{eq:3pt-1-loop-JJJ-step-1}
    \eqref{eq:3pt-1-loop-JJJ}=&\delta_{s,1}\tau_{\tJ}\Big(\frac{1}{z_{23}^2}-\frac{1}{2z_{23}}\p_{2}\Big)\frac{1}{z_{23}}\frac{[2\,3]^{2s-1}}{(2s-1)!}\cU_{s,s}\big\langle \tJ_1[h_1]\tJ_2[-1;0]\big\rangle\nn\\
    &+ \delta_{s,1}\frac{1}{z_{31}}[3\,1]\big\langle \tJ_1[h_1+s-1;2s-1]\tJ_2[s;2s-1]\big\rangle\,,
\end{align}
since the 2-point functions \begin{align}\label{eq:no-JU-JV}
    \langle\tJ[h_x]\tU[h_y]\rangle=0\,,\qquad \langle \tJ[h_x]\tV[h_y]\rangle=0\,,
\end{align}
for any value of $h_x$ and $h_y$, and 
\begin{align}
    \cU_{s,s}=-\frac{32}{\pi^3}\frac{\Gamma(2s)^2}{\Gamma(s)^2}\frac{(s+1)}{s(s+2)}\,.
\end{align}
As a result, 
\begin{align}\label{eq:3-pt-1-loop-result}
    \eqref{eq:3pt-1-loop-JJJ}=&+\delta_{h_1,-1}\delta_{s,1}\tau_{\tJ}\Big(\frac{1}{z_{23}^2}-\frac{1}{2z_{23}}\p_{2}\Big)\frac{[2\,3]^{2s-1}}{(2s-1)!}\cU_{s,s} z_{12}^{2}\nn\\
    &+\delta_{h_1,+1}\delta_{s,1}\frac{[1\,2]^{2s-1}z_{23}}{(2s-1)!}\frac{1}{z_{12}^{2s}z_{23}z_{31}}\,.
\end{align}
Observe that both terms vanish on the support of momentum conservation. Thus, the 3-point amplitude \eqref{eq:3pt-1-loop-JJJ} is protected from quantum correction if $\Spec=\Spec^{\diamond}$.
\end{proof}

\begin{lemma}\label{lem:mixed-term} The one-loop three-point functions
    \begin{subequations}
    \begin{align}
        &\left\langle\tU[0](z_1)\tJ_2\big[s;2s-1\big](z_2)\tJ_3\big[s;2s-1\big](z_3)\right\rangle\,,\\
        &\left\langle\tV[0](z_1)\tJ_2\big[s;2s-1\big](z_2)\tJ_3\big[s;2s-1\big](z_3)\right\rangle\,,\\
        &\left\langle\nor \tJ[h]\tJ[-h](z_1)\nor\tJ_2\big[s,2s-1\big](z_2)\tJ_3\big[s,2s-1\big](z_3)\right\rangle\,,
    \end{align}
\end{subequations}
vanish. 
\end{lemma} 
\begin{proof} 
The first two correlation functions can be checked to vanish since we do not have 2-point functions $\langle \tJ\,\tU\rangle$ or $\langle\tJ\,\tV\rangle$, cf. \eqref{eq:no-JU-JV}. The last term vanishes by definition of the normal ordering, cf. \eqref{eq:normal-order-def} since it contains double pole in the contour integral.
\end{proof}

%%%%%%%%%%%%%%%%%%%%%%%%%%%%%%%%%%%%%%%%%%%%%%%%%%%%%%%%
\paragraph{One-loop 4-point amplitudes.} Let us first warm up with the following amplitudes 
\begin{align}\label{eq:4-point-(-+++)}
    \big\langle \tJ_1[-1]\tJ_2[s;2s-1]\tJ_3[s;2s-1]\tJ_4[1]\big\rangle\,.
\end{align}
This amplitude can be computed easily by rewriting
\begin{align}
    \eqref{eq:3pt-1-loop-JJJ}\sim -&\sum_{h\in\Spec}\frac{[2\,3]^{2s-1}\langle 1\,2\rangle^{2}}{\langle 2\,3\rangle^2}=\sum_{h\in\Spec}\frac{[2\,3]^{2s-1}\langle 1\,2\rangle^2\langle 1\,3\rangle^2}{\langle 2\,3\rangle}\frac{1}{\langle 1\,2\rangle \langle 2\,3\rangle \langle 3\,1\rangle}\,.
\end{align}
For theories with $\Spec=\Z,2\Z+1$, the above is zero. However, with $\Spec=|h|\geq 1$, the above reduces to the result of \cite{Costello:2022upu}, with an overall minus sign coming from Riemann zeta regularization\footnote{There are still higher spins in the loop but the external states are forced to have spin one by symmetry.}
\begin{align}
    \eqref{eq:3pt-1-loop-JJJ}=-\frac{[2\,3]\langle 1\,2\rangle^2\langle 1\,3\rangle^2}{\langle 2\,3\rangle}\frac{1}{\langle 1\,2\rangle \langle 2\,3\rangle \langle 3\,1\rangle}\,,\qquad \Spec=|h|\geq 1\,.
\end{align}
We find by induction that\footnote{The induction essentially reduces to the standard BCFW recursion relations \cite{Britto:2005fq}, as the interactions are of Yang-Mills type. Previous use of BCFW recursive relations can be found e.g. in \cite{Adamo:2022lah,Tran:2022amg}.}
\begin{subequations}
    \begin{align}
    \eqref{eq:4-point-(-+++)}&=0\,,\qquad &\Spec&=\Spec^{\diamond}\,,\\
    \eqref{eq:4-point-(-+++)}&=-\frac{[2\,3]\langle 1\,2\rangle^2\langle 1\,3\rangle^2}{\langle 2\,3\rangle}\frac{1}{\langle 1\,2\rangle \langle 2\,3\rangle \langle 3\,4\rangle\langle 4\,1\rangle}\,,\qquad &\Spec&=|h|\geq 1\,.
\end{align}
\end{subequations}
%In computing the above, we have used the similar BCFW recursion relations \cite{Britto:2005fq}, which guarantee to work for theories with usual gauge interactions. 

\paragraph{One-loop $n$-point amplitudes.} A direct generalization to $n$-point amplitudes of type \small
\begin{align}\label{eq:n-point-1-loop}
   \cAmp_{1\text{-loop}}(1_{-1},2_{+s},3_{+s},4_{+1},\ldots,n_{+1})= \big\langle \tJ_1[-1]\tJ_2[s;2s-1]\tJ_3[s;2s-1]\tJ_4[1]\ldots \tJ_n[1]\big\rangle\,,
\end{align}
\normalsize
with $p=0$ and $\Spec=|h|\geq 1$ is also possible. The result is
\begin{subequations}
    \begin{align}
    \eqref{eq:n-point-1-loop}&=0\,,\qquad & \text{for}\quad  \Spec&=\Spec^{\diamond}\,,\\
    \eqref{eq:n-point-1-loop}&=-\frac{[2\,3]\langle 1\,2\rangle^2\langle 1\,3\rangle^2}{\langle 2\,3\rangle}\frac{1}{\langle 1\,2\rangle \langle 2\,3\rangle \ldots \langle n\,1\rangle}\,,\qquad  &\text{for}\quad  \Spec&=|h|\geq 1\,.
\end{align}
\end{subequations}
Again, the sign of higher-spin one-loop amplitudes are opposite with the usual Yang-Mills one. We also note that the one-loop $n$-point all-plus amplitudes 
\begin{align}\label{eq:YM-1-loop}
    \cAmp_{1\text{-loop}}(1_{+s},2_{+s},3_{+1},\ldots,n_{+1})=\left\langle \tJ_1\big[s,2s-1\big]\tJ_2\big[s,2s-1\big]\tJ_3[1]\ldots\tJ_n[1]\right\rangle\,,
\end{align}
are zero. They can be easily computed along the line of \eqref{eq:3pt-1-loop-JJJ-step-1}. 

Note that these results do not contradict with prior studies of the one-loop all-plus amplitudes in Yang-Mills or QCD (cf. \cite{Mahlon:1993fe,Mahlon:1993si,Bern:1993qk}). Rather, it is a feature of the type of axionic higher-spin theories that we consider here. In particular, $\eqref{eq:YM-1-loop}=0$ can be understood from the fact that the all-plus 4-point one-loop amplitude is cancelled by the tree-level amplitude with axion in the exchange,  %Namely, as the tree-level diagram with axion in the exchange can cancel the 4-point one-loop amplitude associated with the anomaly
\begin{align}\label{eq:anomaly-cancellation-axion}
    \parbox{150pt}{\includegraphics[scale=0.45]{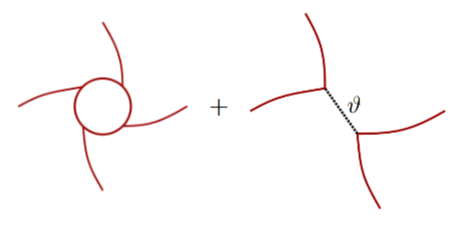}}\ =0\,,
\end{align}
cf. Section \ref{sec:anomaly-cancellation} (see also discussions in \cite{Costello:2022wso,Costello:2023vyy}).

%This is one of the nice things about studying correlation functions of the chiral CFT on $\P^1_{\tp}$. Namely, some form factors (or rather loop amplitudes in the rational sector) can be computed easily by working out some simple algebraic relations. However, it is worth reminding ourselves that to acquire this technology, we first had to work relatively hard to compute bulk/defect loop amplitudes directly in the context of Koszul duality. 

%%%%%%%%%%%%%%%%%%%%%%%%%%%%%%%%%%%%%%%%%%%%%%%%%%%%%%%%%%%%%%%%%%%%%%%%%%%%%%%%%%%%
\subsubsection{Two-loop amplitudes}

%%%%%%%%%%%%%%%%%%%%%%%%%%%%%%%%%%%%%%%%%%%%%%%%%%%%%%%
\paragraph{All-plus two-loop amplitudes.} Let us now consider the following all-plus two-loop amplitudes 
\begin{align}\label{eq:all-plus-2-loop}
    \big\langle \tJ_1[s_1;\tH_1]\tJ_2[s_2;\tH_2]\tJ_3[s_3;\tH_3]\tJ_4[s_4;\tH_4]\big\rangle \,.
\end{align}
In the case where these currents belong to the chiral algebras associated with twistorial theories with the spectrum $\Spec=\Z,2\Z+1$, the above amplitude vanish. Thus, we shall focus on the case $\Spec=|h|\geq 1$. 

By virtue of Theorem \ref{thm:quantum-HS}, all interactions should be of Yang-Mills and all external helicities are $+1$. A short computation (see Appendix \ref{app:loop-OPE}) leads to
\begin{align}
    \eqref{eq:all-plus-2-loop}&=\left\langle \tJ_1\big[1,1\big]\tJ_2\big[1,1\big]\tJ_3\big[1,1\big]\tJ_4\big[1,1\big]\right\rangle\nn\\
    &=%\sum_{h,h'\in\Spec}
    +\tau^2_{\tJ}\frac{2(64\sh^{\vee})^2}{9\pi^6}\frac{[1\,2][3\,4]}{\langle 1\,2\rangle^2\langle 3\,4\rangle^2}\Big(\langle 1\,3\rangle\langle 2\,4\rangle+\langle 1\,4\rangle\langle 2\,3\rangle\Big)\nn\\
    &\quad +%\sum_{h,h'\in\Spec}
    \tau^2_{\tJ}\frac{2(64\sh^{\vee})^2}{9\pi^6}\frac{[2\,3][4\,1]}{\langle 2\,3\rangle ^2\langle 4\,1\rangle ^2}\Big(\langle 2\,1\rangle\langle 3\,4\rangle+\langle 2\,4\rangle\langle 4\,1\rangle\Big)\nn\\
    &\quad +\tau^2_{\tJ}\frac{2(32\sh^{\vee})^2}{\pi^6}\frac{[1\,2][3\,4]}{\langle 1\,2\rangle \langle 2\,3\rangle}\nn\\
    &\quad +c_{\mg}^2\frac{[1\,2][3\,4]}{\langle 1\,2\rangle\langle 3\,4\rangle}%\sum_{h\in\Spec}
    \Big(\frac{\langle 1\,3\rangle^2+2\langle 1\,3\rangle(\langle 1\,2\rangle-\langle 3\,4\rangle)-\langle 1\,2\rangle\langle 3\,4\rangle}{\langle 1\,2\rangle\langle 3\,4\rangle }\Big)\nn\\
    &\quad +c_{\mg}^2\frac{[1\,2][3\,4]}{\langle 1\,2\rangle\langle 3\,4\rangle}%\sum_{h\in\Spec}
    \Big(\frac{\langle 2\,4\rangle^2+2\langle 2\,4\rangle(\langle 2\,3\rangle-\langle 4\,1\rangle)-\langle 2\,3\rangle\langle 4\,1\rangle}{\langle 2\,3\rangle\langle 4\,1\rangle}\Big)\nn\\
    &\quad+2c^2_{\mg}%\sum_{h\in\Spec}
    \frac{[1\,2][3\,4]}{\langle 1\,2\rangle\langle 3\,4\rangle}\,,
\end{align}
where it is useful noting that in obtaining the above, we have used
\begin{align}
    \frac{[1\,2][3\,4]}{\langle 1\,2\rangle\langle 3\,4\rangle}=\frac{[1\,3][4\,2]}{\langle 1\,3\rangle\langle 4\,2\rangle}=\frac{[1\,4][2\,3]}{\langle 1\,4\rangle\langle 2\,3\rangle}\,,
\end{align}
and 
\begin{align}
    \tau_\tJ=\frac{i}{2^{10}}\,,\qquad \sum_{h}1=-1\,,\quad \text{for}\quad \Spec=|h|\geq 1\,.
\end{align}
%At this stage, we can conclude that higher-spin loop amplitudes are similar to the one of integrable Yang-Mills theory, cf. \cite{Costello:2022upu,Costello:2023vyy}.

\paragraph{Remarks.} One of the motivations for the study of the chiral higher-spin algebra $\ca$ in this work is to compute the two-loop amplitudes of chiral higher-spin gravity, cf. \cite{Metsaev:1991mt,Metsaev:1991nb,Ponomarev:2016lrm}, which are expected to be vanishing or at least rational due to the stringent constraints of higher-spin symmetry.

At one loop, the amplitudes of colored chiral higher-spin gravity \cite{Skvortsov:2020wtf,Skvortsov:2020gpn} exhibit similar features with the standard all-plus one-loop Yang-Mills amplitudes \cite{Mahlon:1993fe,Mahlon:1993si,Bern:1993qk} up to a sum over helicities, which is regularized to zero. At two loop, the computation via dual regional momentum, cf. \cite{Chakrabarti:2005ny}, in the light-cone gauge can be quite challenging given that the number of diagrams at two loops is large.

\begin{figure}[ht!]
    \centering
    \includegraphics[scale=0.3]{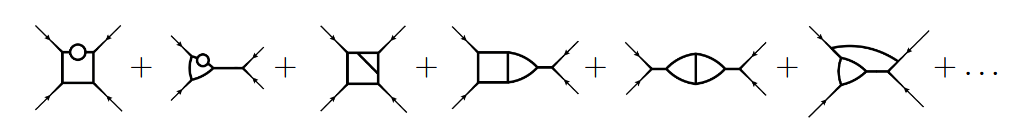}
    \caption{Some of the many two-loop diagrams for chiral higher-spin theories.}
    \label{fig:1}
\end{figure}

For this reason, we have extended the chiral bootstrap program \cite{Costello:2022upu} to higher-spin case. Theorem \ref{thm:quantum-associative} and the results from bootstrapping the form factors of all-plus amplitudes in this work indicate that the 2-loop amplitudes for chiral higher-spin gravity are zero. This is due to the fact that chiral higher-spin gravity is a higher-derivative theories and has the spectrum $\Spec=\Z$. Thus, it cannot have non-trivial imprints on the celestial twistor sphere.

%%%%%%%%%%%%%%%%%%%%%%%%%%%%%%%%%%%%%%%%%%%%%%%%%%%%%%%%%%%%%%%%%
\subsection{On the inclusion of special matter fields} 
Note that beside the axionic fields, one can also introduce suitable Weyl fermions into the anomaly cancellation mechanism. In this case, the Okubo's relation \eqref{eq:Okubo} is modified to
\begin{align}
    \Tr(T^4)-\tr_R(T^4)=C_{\mg,R}\tr(T^2)\tr(T^2)\,,
\end{align}
where $R$ is the representation that the Weyl fermions take values in. In the case where $\Spec\neq \Z,2\Z+1$, the amplitudes can be non-trivial with (higher-spin) fermions in the loop; leading to QCD-like theories with axions. In particular, one can construct a QCD theory with the number of quark flavors $N_f=3$ and suitable axion field. This theory is shown to be rational at two loop, cf. \cite{Dixon:2024tsb}. Although the coupling $C_{\mg,R}$ will be modified  accordingly to the representations in which fermions take values in, we again expect the higher-spin loop amplitudes to be similar to that of ``rational QCD'' \cite{Dixon:2024tsb} with an opposite sign. Note that it is also possible to switch off the axion couplings completely. In this case, the anomaly cancellation will be handled by fermions in the loop instead of the axions in tree-level exchange diagrams \cite{Costello:2023vyy,Dixon:2024tsb}. This is certainly an interesting direction to pursue in constructing higher-spin extensions of QCD-like theories. We leave this for future work.

%%%%%%%%%%%%%%%%%%%%%%%%%%%%%%%%%%%%%%%%%%%%%%%%
\section{Chiral theories on the defects}\label{sec:boson-fermi}
Recall that we have not excluded the possibility that the holomorphic higher-spin currents $\tJ[h,\tH]$ can be built from some matter fields or ghost systems. In this section, we will propose some chiral CFTs on the defect that can give rise to those currents.\footnote{See e.g. \cite{Ponomarev:2022ryp,Ponomarev:2022qkx,Costello:2022jpg,Costello:2023hmi} for previous proposals for constructing certain dual pairs in the context of flat holography.} 

%From the presented materials in Section \ref{sec:3}, the reader might notice a peculiar feature of the bulk/defect systems in consideration. Namely, we do not have a would-be-graviton in the bulk, which should couple to the energy-stress tensors of the proposed chiral CFTs. This is due to the fact that, the gauge potential $\tA_{2}\sim \tA_{\alpha(2)}\lambda^{\alpha(2)}\bar{e}^0$ is a source for the graviton $h_{\alpha\dot\alpha}$, rather than the graviton itself. (More specifically, the graviton should be viewed as a derived object in self-dual gravity or chiral higher-spin gravity, see e.g. \cite{Tran:2025yzd}). This, in turn, may explain the observation in Section \ref{sec:4} that there are no MHV amplitudes for gravity from correlations functions of chiral higher-spin currents. %Put differently, gravity may not admit a simple deformation from the chiral/self-dual theories constructed from twistor space.

Recall that our higher-spin currents $\tJ$ have the following form $\tJ[h;2k]=\tilde v^{\dot\alpha (2k)}\tJ_{\dot\alpha(2k)}$. 
As $\tJ_{\dot\alpha(2k)}$ should have conformal helicity weight $h$ and bilinear in the matter fields, it is natural to consider $\tJ_{\dot\alpha(2k)}=\sn_{\dot\alpha(2k)}\cG$, so that
\begin{align}
\tJ[h;2k]=[\tilde v\,\sn]^{2k}\cG\,,\qquad  \cG=z^{m}\Big(\Phi (\overleftrightarrow{\p_z})^n\Phi\Big)\,.
\end{align}
where $\sn$ are some reference spinors, and $\cG$ bilinear in fields. %Since $\p_z$ increases and $z$ decreases the conformal dimensions, we may consider the following holomorphic function
%\begin{align}\label{eq:build-G}
%    \cG=z^{m}\Big(\Phi (\overleftrightarrow{\p_z})^n\Phi\Big)\,,
%\end{align}
%where $\Phi$ denote some matter fields. %As a result, $\cG$ encodes the higher-spin currents $\tJ[h]$ where
%\begin{align}
%    h=n+2\Delta_{\Phi}-m\,.
%\end{align}
%Our task is to construct CFTs with some specific matter fields $\Phi$. As usual, we split the analysis into colored and colorless case. We start with the easier one, i.e. abelian chiral matter fields.
%%%%%%%%%%%%%%%%%%%%%%%%%%%%%%%%%%%%%%%%%%%%%%%%
%\subsection{Chiral bosonic and fermionic theories on the defect}
Two of the simplest abelian chiral theories which can source the above holomorphic higher-spin currents are the chiral boson and chiral fermion theories. %In two dimensions, these two theories can be related by bosonization.

\paragraph{Chiral boson theory.} Consider the following action:
\begin{align}
    S_{\phi}=\frac{1}{2}\int_{\C^{\times}} \p \phi \bar{\p}\phi-\mathsf{b}\bar{\p}\phi\,, \qquad \mathsf{b}\in \Omega^{1,0}(\C^{\times})\,,
\end{align}
where $\mathsf{b}$ is a Langrangian multiplier that allows us to impose the holomorphicity condition 
\begin{align}
    \bar{\p}\phi=0\,, \qquad \bar{\p}\equiv d\bar{z}\p_{\bar{z}}\,.
\end{align}
Note that the above action is not Lorentz invariant, as in various actions for chiral bosons available in the literature, see e.g. \cite{Sonnenschein:1988ug,Siegel:1983es}.\footnote{See also \cite{Arvanitakis:2022bnr} and \cite{Chen:2025xlo} for a modern take on this problem.} As classical OPE of holomorphic higher-spin currents should give \eqref{eq:full-tree-OPE}, we consider
\begin{align}
    \cG_{\phi}=z^{-h+1}\phi\p_z\phi\,,
\end{align}
provided $\Delta_{\phi}=0$, and
\begin{align}
    \langle \phi(z_1)\phi(z_2)\rangle \sim \log(z_{12})\,,\qquad z_{12}=z_1-z_2\,.
\end{align}
Then, higher-spin currents take the form:
\begin{align}
    \tJ[h,2k](z,\tilde v)=\frac{[\tilde v\,\sn]^{2k}}{(2k)!}z^{-h+1}\phi\p_z\phi\,.
\end{align}
Notice the special role of the current $\tJ[1;0]$, which acts as the seed upon which all other currents are generated from. This is quite different with the usual higher-spin currents in the literature. (See e.g. \cite{Ponomarev:2022ryp,Ponomarev:2022qkx} for a recent dual pair proposal for chiral higher-spin gravity in flat space, formulated in the same spirit with Flato-Fronsdal theorem \cite{Flato:1978qz}.) %Namely, we do not need to take higher derivatives to generate higher-spin currents. Rather, all chiral higher-spin modes in chiral/self-dual higher-spin theories should come from the massless Kaluza-Klein modes of the $\P^1$-fibers in the fibration $\PTc\rightarrow \cM$. 
%See \cite{Costello:2023hmi} for some similar consideration. 

%%%%%%%%%%%%%%%%%%%%%%%%%%%%%%%%%%%%%%%%%%%%%%%%
\paragraph{Chiral fermion.}
The chiral fermionic theory has the following action
\begin{align}
    S_{\psi}=\int d^2z \,\psi\p_{\bar z}\psi\,.
\end{align}
Similar to the chiral boson case, we want holomorphic currents in this case to also produce OPEs with simple poles at classical level. Since $\psi$ has conformal weight $\frac{1}{2}$, it is natural to consider 
\begin{align}
    \tJ[h,2k](z,\tilde v)=\frac{[\tilde v\,\sn]^{2k}}{(2k)!}z^{-h+1}\psi\psi\,,\quad \text{where}\quad \langle \psi(z_1)\psi(z_2)\rangle\sim \frac{1}{z_{12}}\,.
\end{align}
we may promote the matter fields to matrix-valued fields.
%%%%%%%%%%%%%%%%%%%%%%%%%%%%%%%%%%%%%%5
\paragraph{Matrix-valued currents.} Note that we may also promote the matter fields to matrix-valued fields. In particular,
\begin{align}
    \phi\mapsto\phi_{Ir}\,,\quad \psi\mapsto \psi_{Ir}
\end{align}
where now the matter fields take value in the bi-fundamental representations of, say, $Sp(2N)$ and $O(m)$. Here, $(I,J)$ denote $Sp(2N)$ indices and $(r,s)$ are $O(m)$ indices. Then, $Sp(2N)$-valued currents in the adjoints may be written as
\begin{subequations}\label{eq:explicit-currents}
    \begin{align}
  \text{bosonic matters}&:  &\tJ[h,2k](z,\tilde v)&=\frac{[\tilde v \,\sn]^{2k}}{(2k)!}z^{-h+1}\phi_{Ir}\overleftrightarrow{\p}_{z}\phi_{J}{}^r\,,\\
  \text{fermionic matters}&: &\tJ[h,2k](z,\tilde v)&=\frac{[\tilde v \,\sn]^{2k}}{(2k)!}z^{-h+1}\psi_{Ir}\,\psi_J{}^r\,.
\end{align}
\end{subequations}
We will stop our investigation on the chiral CFTs on the defect here, and refer the reader to \cite{Costello:2023hmi} for another proposal.
%%%%%%%%%%%%%%%%%%%%%%%%%%%%%%%%%%%%%%%%%%%%%%%%
%\paragraph{Bosonization.} It is well-known that bosonic and fermionic systems in $2d$ can be related to each other through bosonization. In particular, one can consider the following non-local maps
%\begin{align}       W=\log(\psi^{\dagger})\,\psi^{\dagger}\p_z\psi\,,\qquad \text{where}\quad  \psi=e^{i\phi}\,.
%\end{align}
%It can now be checked that the higher-spin currents in the above abelian systems can be indeed related to each other.\footnote{Since there is a transmutation in statistics, it is easy to anticipate that the bosonization maps should be somewhat non-local.}
%%%%%%%%%%%%%%%%%%%%%%%%%%%%%%%%%%%%%%%%%%%%%%%%

%%%%%%%%%%%%%%%%%%%%%%%%%%%%%%%%%%%%%%%%%%%%%%%%%%%%%%%%%%%%%%%%

%%%%%%%%%%%%%%%%%%%%%%%%%%%%%%%%%%%%%%%%%%%%%%%%%%%%%%%%%%%%%%%%
\section{Discussion}\label{sec:discuss}

In this work, we study the chiral higher-spin symmetry algebras $\ca$ of various twistorial higher-spin theories via Koszul duality \cite{Costello:2022wso} to first quantum order. Our choice of conventions allow us to identify $\ca$ with the color-kinematic algebra of $4d$ chiral/self-dual higher-spin theories in \cite{Ponomarev:2017nrr,Monteiro:2022xwq} at classical level. We show that, to first order in quantum correction, the $\ca$ associated with anomaly-free twistorial higher-spin theories are associative, while the chiral algebras associated with anomalous twistorial higher-spin theories, cf. \cite{Tran:2025uad}, require additional (axionic) currents to be associative. %-- the byproducts of the anomaly cancellation in twistor space -- via Koszul duality. 

Upon analyzing chiral algebras at one-loop, we observed that, for quantum consistency, only theories with $\Spec=\Z,2\Z+1,2\Z$ can \emph{effectively} have higher-derivative interactions, while other cases are forced to have Yang-Mills type interactions\footnote{These are one-derivative interactions in light-cone language, cf. \cite{Metsaev:2005ar,Ponomarev:2017nrr}.}. 

%Furthermore, for the associativity of the chiral algebras associated to anomalous theories to hold, all external states, which carry $SU(2)_-$ charges, should have helicity $+1$ and interactions should be of Yang-Mills type.

In computing higher-spin form factors%, which can be also viewed as genuine scattering amplitudes in the rational sector (see Table \ref{tab:form factors})
, we find that there can be non-trivial tree-level higher-spin amplitudes. However, the case of higher-derivative theories are not fully justified. Then, at loop levels we find that the amplitudes for theories with higher-derivative interactions are essentially zero due to quantum-integrability, cf. Theorem \ref{thm:quantum-associative}. Restricting our attention to the loop amplitudes in twistorial theories with $\Spec=|h|\geq 1$ and Yang-Mills like interations, we find that their amplitudes are similar to the ones in \cite{Costello:2022upu,Costello:2023vyy} with some overall signs different due to the choice of Riemann regularization for the sum over the spectrum. 

We also propose some chiral CFTs which can source the higher-spin currents which generate the chiral algebras studied in this work. %Upon introducing suitable ghost systems, we find that the proposed chiral CFTs can have structures that closely resemble those of a worldsheet theory. 
However, it is important stressing that there is no standard energy-momentum tensor as well as the couplings between the spin-2 currents of the chiral theories with the bulk gravitons. This stems from the fact that all higher-spin theories in consideration are constructed from the chiral field representations (see e.g. \cite{Krasnov:2021nsq,Adamo:2022lah,Lang:2025rxt}).

 What we have learned through this work as well as \cite{Tran:2025uad} is that there are a large set of higher-spin theories that are quantum-protected. In particular, theories with the right spectrum and higher-derivative interactions should have strong integrability properties such that their amplitudes are either zero or sufficiently simple. 
Our work indicates that chiral higher-spin gravity \cite{Metsaev:1991mt,Metsaev:1991nb,Ponomarev:2016lrm}, should essentially have trivial amplitudes to all order in perturbation theory due to quantum-integrability. %(It is this feature of higher-derivative interactions in chiral or self-dual higher-spin theories, which allow them to calmly escape various no-go statements in flat space, cf. \cite{Weinberg:1964ew,Coleman:1967ad}.) 

%(A nice example along this line in the light-cone gauge can be found e.g. in \cite{Skvortsov:2018jea,Skvortsov:2020wtf,Skvortsov:2020gpn,Serrani:2025owx}.) %We also learn that for chiral or quasi-chiral theories with higher-derivative interactions, their loop amplitudes should essentially be trivial. T
Lastly, our results suggest that there may be some theories have yet to be constructed in twistor space or spacetime. These statements are in agreement with the recent light-cone analysis in \cite{Serrani:2025owx}. There, it was observed that there is a large class of self-dual higher-derivative higher-spin theories. Some of these theories can have non-trivial scattering amplitudes, and some may even have a finite spectrum, or spectrum with fractional spins. Since the light-cone method avoids issues with gauge redundancy, it will be interesting to covariantize the results of \cite{Serrani:2025owx} to check if they are, in fact, invariant statements.

%Intriguingly, in a recent light-cone analysis, cf. \cite{Serrani:2025owx}, one finds that there exist a large class of chiral higher-spin theories. These theories can have spectrum with or without infinitely many spins at classical level. Some of them can even have fractional spins in $4d$. It would be interesting to covariantize these results to check whether they are merely light-cone statements or gauge-invariant statements.

%Nevertheless, this turns out to be an advantage since it implies that there may be non-trivial higher-spin scattering amplitudes which can be written in terms of correlation functions on the celestial twistor sphere. As a result, we may be able to probe for more interesting bulk theories starting from the dual chiral CFT. This will be discussed in Section \ref{sec:4}.

%The next obvious goal in the higher-spin program (see \cite{Bekaert:2022poo} so a large set of higher-spin problems) is to construct a worldsheet theory -- a longstanding objective that has remained elusive for decades.

 %(See, nevertheless, a recent proposal on the worldline model for higher-spin gravity [].) We hope that this problem will be solved in a near future. 

%%%%%%%%%%%%%%%%%%%%%%%%%%%%%%%%%%%%%
\section*{Acknowledgment}

Many useful discussions with Alexander Alexandrov, Norton Lee, Matthew Roberts and Zhenya Skvortsov are grateful acknowledged. The author also thanks Tim Adamo for the collaboration in \cite{Adamo:2022lah}, during which he learned many valuable lessons in twistor theory. This work is supported by the Young Scientist Training (YST) program at the Asia Pacific Center of Theoretical Physics (APCTP) through the Science and Technology Promotion Fund and Lottery Fund of the Korean Government, and also the Korean Local Governments
– Gyeongsangbuk-do Province and Pohang City.

%%%%%%%%%%%%%%%%%%%%%%%%%%%%%%%%%%%%%%%%%%%%%%%%%%
\appendix

%%%%%%%%%%%%%%%%%%%%%%%%%%%%%%%%%%%%%%%%%%%%%%%%%%
\section{Bubble-integrals}\label{app:integral}
This appendix evaluates the doubly nested bubble-integral \eqref{eq:doubly-nested-int}. The computation below is similar to the one along the lines of e.g. \cite{Bittleston:2023bzp,Fernandez:2023abp}. Recall that with the factor $d\bar{z}_0$ removed, 
\begin{align}
    I_{23}:=\int_{\C^3\times\C^3}DX_2DX_3
      \frac{z_2|w_2^{\dot 2}|^{2s_2}|w_3^{\dot 1}|^{2s_3}}{|Y_1-X_2|^{2(s_2+2)}|X_2-X_3|^6|Y_4-X_3|^{2(s_3+2)}}\,.
\end{align}
%\begin{align}
%    I_{23}:=\int_{\C^3\times\C^3}DX_2DX_3
%      \frac{z_2|w_2^{\dot 2}|^{2(2h_2-1)}|w_3^{\dot 1}|^{2(2h_3-1)}}{|Y_1-X_2|^{2(2h_2+1)}|X_2-X_3|^6|Y_4-X_3|^{2(2h_3+1)}}\,.
%\end{align}
Using standard the Feynman parametrization
\small
\begin{align}\label{eq:Feynman-parametrization}
    \frac{1}{X_1^{a_1}\ldots X_n^{a_n}}=\frac{\Gamma(a_1+\ldots+a_n)}{\Gamma(a_1)\ldots\Gamma(a_n)}\int_{[0,1]^n} \Big(\prod_{i=1}^ndt_i\,t_i^{a_i-1}\Big)\frac{\delta(1-\sum_{i=1}^nt_i)}{\big(t_1X_1+\ldots+t_nX_n\big)^{a_1+\ldots+a_n}}\,.
\end{align}
\normalsize
Since we have constrained $h_2,h_3\geq 1$, cf. , the integral above can be written as
\small
\begin{align}\label{eq:int-step-1}
    I_{23}:&=\int DX_2 \frac{z_2|w_2^{\dot 2}|^{2s_2}}{|Y_1-X_2|^{2(s_2+2)}}\int DX_3 \frac{|w_3^{\dot 1}|^{2s_3}}{|X_2-X_3|^6|Y_4-X_3|^{2(s_3+2)}}\nn\\
    &=\frac{\Gamma(s_3+5)}{2!\Gamma(s_3+1)}\int DX_2 \frac{z_2|w_2^{\dot 2}|^{2s_2}}{|Y_1-X_2|^{2(s_2+2)}}\int_{[0,1]}dt \,t^2(1-t)^{s_3+1}\int D\widetilde{X} \frac{|\tilde w^{\dot 1}|^{2s_3}}{\Big(|\tilde w|^2+t(1-t)|X_2-Y_4|^2\Big)^{s_3+5}}
\end{align}
\normalsize
%\small
%\begin{align}\label{eq:int-step-1}
%    I_{23}:&=\frac{\Gamma(2h_3+4)}{2!\Gamma(2h_3+1)}\int DX_2 \frac{z_2|w_2^{\dot 2}|^{2(2h_2-1)}}{|Y_1-X_2|^{2(2h_2+1)}}\int DX_3 \frac{|w_3^{\dot 1}|^{2(2h_3-1)}}{|X_2-X_3|^6|Y_4-X_3|^{2(2h_3+1)}}\nn\\
%    &=\frac{\Gamma(2h_3+4)}{2!\Gamma(2h_3+1)}\int DX_2 \frac{z_2|w_2^{\dot 2}|^{2(2h_2-1)}}{|Y_1-X_2|^{2(2h_2+1)}}\int_{[0,1]}dt \,t^2(1-t)^{2h_3}\int D\widetilde{X} \frac{|\tilde w^{\dot 1}|^{2(2h_3-1)}}{\Big(|\tilde w|^2+t(1-t)|X_2-Y_4|^2\Big)^{2(h_3+2)}}
%\end{align}
%\normalsize
where $\tilde X^a=X_3^a-tX_2^a -(1-t)Y_4^a$ ($a=1,2,3$). It is useful to recall that
\begin{align}
    X_2^a=(z_2,w_2^{\dot \alpha})\,,\quad X_3^a=(z_3,w_3^{\dot\alpha })\,;\qquad Y_1^a=(z_1,0)\,,\quad Y_4^a=(z_4,0)\,.
\end{align}
At this stage, we can perform the integral over real variables
\begin{align}
    q_0=|z_3|^2\,,\qquad q_1=|w_3^{\dot 1}|^2\,,\qquad q_2=|w_3^{\dot 2}|^2\,,
\end{align}
which yields
\small
\begin{align}\label{eq:int-step-2}
    I_{23}&=(2\pi)^3\frac{\Gamma(s_3+5)}{2!\Gamma(s_3+1)}\int DX_2 \frac{z_2|w_2^{\dot 2}|^{2s_2}}{|Y_1-X_2|^{2(s_2+2)}}\nn\\
    &\qquad \qquad \times\int_{[0,1]}dt \,t^2(1-t)^{s_3+1}\int_{[0,\infty)^3} \frac{ dq_0dq_1dq_2\,q_1^{s_3}}{\Big(q_0+q_1+q_2+t(1-t)|X_2-Y_4|^2\Big)^{s_3+5}}\nn\\
    &=\frac{(2\pi)^3}{2s_3}\int DX_2\frac{z_2|w_2^{\dot 2}|^{2s_2}}{|Y_1-X_2|^{2(s_2+2)}|X_2-Y_4|^4}\,.
\end{align}
\normalsize
%\small
%\begin{align}\label{eq:int-step-2}
%    I_{23}&=(2\pi)^3\frac{\Gamma(2h_3+4)}{2!\Gamma(2h_3+1)}\int DX_2 \frac{z_2|w_2^{\dot 2}|^{2(2h_2-1)}}{|Y_1-X_2|^{2(2h_2+1)}}\nn\\
%    &\qquad \qquad \times\int_{[0,1]}dt \,t^2(1-t)^{2h_3}\int_{[0,\infty)^3} \frac{ dq_0dq_1dq_2\,q_1^{2h_3-1}}{\Big(q_0+q_1+q_2+t(1-t)|X_2-Y_4|^2\Big)^{2(h_3+2)}}\nn\\
%    &=\frac{(2\pi)^3}{4h_3(2h_3-1)}\int DX_2\frac{z_2|w_2^{\dot 2}|^{2(2h_2-1)}}{|Y_1-X_2|^{2(2h_2+1)}|X_2-Y_4|^4}\,.
%\end{align}
\normalsize
Note that the $(2\pi)^3$ factor comes from the angular integrals. The $X_2$-integral can be done analogously by writing the integrated variables in terms of
\begin{align}
    \cX^a=X_2^a-t Y_1^a-(1-t)Y_4^a\,,\qquad a=1,2,3\,,
\end{align}
so that $z_2=\cX^0+t z_1+(1-t)z_4$. Since the integral should be real, we may safely discard the contribution associated to $\cX^0$. As a result,
\begin{align}
    I_{23}:&=\frac{(2\pi)^6}{4s_3}\frac{\Gamma(s_2+4)}{\Gamma(s_2+1)}\int_{[0,1]}dt \,t(1-t)^{s_2+1}(tz_1+(1-t)z_4)\nn\\
    &\qquad \qquad \times\int_{[0,\infty)^3} dr_0dr_1dr_2\frac{r_2^{s_2}}{\Big(r_0+r_1+r_2+t(1-t)|z_{14}|^2\Big)^{s_2+4}}\,,
\end{align}
Thus, our final result reads
\begin{align}
    I_{23}=\Ccurl_{s_2,s_3}\frac{z_1+(1+s_2) z_4}{|z_{14}|^2}=\frac{\Ccurl_{s_2,s_3}}{2|z_{14}|^2}\Big[(s_2+2)z_0-s_2\frac{z_{14}}{2}\Big]\,,
\end{align}
where
\begin{align}
    \Ccurl_{s_2,s_3}=\frac{(2\pi)^6}{4s_3(s_2+1)(s_2+2)}\,.
\end{align}

%%%%%%%%%%%%%%%%%%%%%%%%%%%%%%%%%%%%%%%%%%%%%%%%%%%%
\section{On associativity at one loop}\label{app:check-associativity}
This Appendix checks the associativity of the chiral higher-spin algebra associated with anomalous higher-spin theories in twistor space at first order in quantum correction. Here, we work with the case where the OPE of the higher-spin currents $\tJ$ also includes the contributions coming from the axionic currents since it is more general. 

Now, for associativity to hold at first order in quantum correction, the following equation must hold \cite{Costello:2022upu}
%\begin{align}\label{eq:to-be-check-app}
%  &\oint_{|z_1|} dz_1 z_1\tJ^{a_1}[h_1]\oint_{|z_{23}|<|z_1|}\tJ^{a_2}\Big[s_2;2s_2-1\Big]\tJ^{a_3}\Big[s_3;2s_3-1\Big]\nn\\
%    &=+\oint_{|z_2|} dz_2 \tJ^{a_2}\Big[s_2;2s_2-1\Big]\oint_{|z_{31}|<|z_2|}\tJ^{a_3}\Big[s_3;2s_3-1\Big]\tJ^{a_1}[h_1]z_1\nn\\
%    &\ \ \ \,+\oint_{|z_3|} dz_3 \tJ^{a_3}\Big[s_3;2s_3-1\Big]\oint_{|z_{12}|<|z_3|}\tJ^{a_1}[h_1]z_1\tJ^{a_2}\Big[s_2;2s_2-1\Big]\,.
%\end{align}
\begin{align}\label{eq:to-be-check-app}
  &\oint_{|w|=2} dw\, w\tJ^{a_1}[h_1](w)\oint_{|z|=1}\tJ^{a_2}\Big[s_2;2s_2-1\Big](0)\tJ^{a_3}\Big[s_3;2s_3-1\Big](z)\nn\\
    &=+\oint_{|z|=2} dz \tJ^{a_2}\Big[s_2;2s_2-1\Big](0)\oint_{|z-w|=1}\tJ^{a_3}\Big[s_3;2s_3-1\Big](z)\tJ^{a_1}[h_1](w)w\nn\\
    &\ \ \ \,+\oint_{|z|=2} dz \tJ^{a_3}\Big[s_3;2s_3-1\Big](z)\oint_{|w|=1}\tJ^{a_2}\Big[s_2;2s_2-1\Big](0)\tJ^{a_1}[h_1](w)w\,.
\end{align}
Note that the power $n$ in \eqref{eq:to-be-checked} is chosen based on the pole and derivative structures in the OPEs
\begin{align}\label{eq:master-1-loop-OPE}
    \tJ^{a_2}\Big[s_2&;2s_2-1\Big](z)\tJ^{a_3}\Big[s_3;2s_3-1\Big](0)\sim\nn\\
    &+\frac{1}{z}\sum_p\tg_p^{a_2a_3c}\frac{[2\,3]^{p}}{p!}\tJ^c[s_2+s_3-1-p;2s_2+2s_3-2]\nn\\
    &+\tau_{\tJ}\Big(\frac{1}{z^2}-\frac{1}{2z}\p_z\Big) \sum_{h\in\Spec}\cU^{a_1a_2a_3a_4}_{s_2,s_3}\sum_p\tg^{a_1a_4f}_p\frac{[2\,3]^p}{p!}\tJ^f[-1-p](z)\nn\\
    &-\frac{\tau_{\tJ}}{z}\sum_{h\in\Spec}\cT^{a_1a_2a_3a_4}_{s_2,s_3}\nor \tJ^{a_1}[h]\tJ^{a_4}[-h](z)\nor\nn\\
    &-c_{\mg}\Big(\frac{1}{z^2}+\frac{1}{z}\p_z\Big)\sum_p\kappa^{a_2a_3}\frac{[2\,3]^{p}}{p!}\tU[s_2+s_3-2-p;2s_2+2s_3-2]\nn\\
        &-\frac{c_{\mg}}{z}\sum_p\kappa^{a_2a_3}\frac{[2\,3]^{p}}{p!}\tV[s_2+s_3-2-p;2s_2+2s_3-4] \,,
\end{align}
\normalsize
and
\begin{subequations}\label{eq:relevant-axion-OPE}
    \begin{align}
        \tJ^{a}\big[h_i\big](z)\tV[0](0)&\sim -c_{\mg}\Big(\frac{1}{z^2}+\frac{1}{z}\p_z\Big)\sum_p\frac{[i\,j]^p}{p!}\tJ^{a}[h_i-2-p]\,,\\
        \tJ^{a}\big[h_i,\tH_i\big](z)\tU[0](0)&\sim -\frac{c_{\mg}}{z}\sum_{p}\frac{[i\,j]^{p}}{p!}\tJ^{a}[h_i-2-p;\tH_i-2]\,.
    \end{align}
\end{subequations}
Here, we simplify the notations by writing $[i\,j]\equiv [\tilde v_i\,\tilde v_j]$.

%For convenience, let us recall the result in Section \ref{sec:3}, which are

%It is useful to note that the color indices associated with the axionic currents are introduced merely for convenience. In particular, whenever an index is contracted with an axionic current such as $\tilde{U}$ or $\tilde{V}$, we will replace $\tilde{g}^{abc}$ with $\kappa^{ab}$.
%%%%%%%%%%%%%%%%%%%%%%%%%%%%%%%%%%%%%%%%%%%%%%%%%%%%%%%%
\subsection{Checking associativity}
%Let us now execute each line in \eqref{eq:to-be-check-app} explicitly. 
In executing the first layer of the contour integrals, we will look for contributions with appropriate poles: 1st order pole in the first line of \eqref{eq:to-be-check-app}, and 2nd order poles in the second and third lines of \eqref{eq:to-be-check-app}. 
%%%%%%%%%%%%%%%%%%%%%%%%%%%%
 \paragraph{The first line of \eqref{eq:to-be-check-app}.} Let us first proceed with 
\begin{align}\label{eq:1st-line}
   &\oint_{|w|=2} dw \,w\tJ^{a_1}[h_1](w)\oint_{|z|=1}\tJ^{a_2}\Big[s_2;2s_2-1\Big](0)\tJ^{a_3}\Big[s_3;2s_3-1\Big](z)\nn\\
   =&-\oint dw\,w\tJ^{a_1}[h_1](w)\sum_{p}\tg_p^{a_2a_3c}\frac{[2\,3]^p}{p!}\tJ^c[s_2+s_3-1-p;2s_2+2s_3-2](z)\nn\\
    &-\frac{\tau_{\tJ}}{2}\oint dw\,w\tJ^{a_1}[h_1](w)\p_{z} \sum_{h\in\Spec}\cU^{a_ma_2a_3a_n}_{s_2,s_3}\sum_p\tg_p^{a_ma_nf}\frac{[2\,3]^{p}}{p!}\tJ^f[-1-p](z)\nn\\
    &-\tau_{\tJ}\oint dw\,w\tJ^{a_1}[h_1](w)\sum_{h\in\Spec}\cT^{a_ma_2a_3a_n}_{s_2,s_3}\nor \tJ^{a_m}[h]\tJ^{a_n}[-h](z)\nor\,\nn\\
    &+c_{\mg}\oint dw\,w\tJ^{a_1}[h_1](w)\p_{z}\sum_p\kappa^{a_2a_3}\frac{[2\,3]^{p}}{p!}\tU[s_2+s_3-2-p;2s_2+2s_3-2](z)\nn\\
        &+c_{\mg}\oint dw\,w\tJ^{a_1}[h_1](w)\sum_p\kappa^{a_2a_3}\frac{[2\,3]^{p}}{p!}\tV[s_2+s_3-2-p;2s_2+2s_3-4](z)
\end{align}
where we have extracted terms with first order poles in \eqref{eq:master-1-loop-OPE}. 

Next, we look for terms that have second order poles since the measure is $dw\,w$. Notice that the last line of \eqref{eq:1st-line} does not contribute. As a result, \eqref{eq:1st-line} reduces to
\begin{align}\label{eq:1st-line-2}
    &-\oint dw\,w\tJ^{a_1}[h_1](w)\sum_{p}\tg_p^{a_2a_3c}\frac{[2\,3]^p}{p!}\tJ^c[s_2+s_3-1-p;2s_2+2s_3-2](z)\nn\\
    &-\frac{\tau_{\tJ}}{2}\oint dw\,w\tJ^{a_1}[h_1](w)\p_{z} \sum_{h\in\Spec}\cU^{a_ma_2a_3a_n}_{s_2,s_3}\sum_p\tg_p^{a_ma_nf}\frac{[2\,3]^{p}}{p!}\tJ^f[-1-p](z)\nn\\
    &-\tau_{\tJ}\oint dz_1z_1\tJ^{a_1}[h_1](w)\sum_{h\in\Spec}\cT^{a_ma_2a_3a_n}_{s_2,s_3}\nor \tJ^{a_m}[h]\tJ^{a_n}[-h](z)\nor\,\nn\\
    &+c_{\mg}\oint dw\,w\tJ^{a_1}[h_1](w)\p_{z}\sum_p\kappa^{a_2a_3}\frac{[2\,3]^{p}}{p!}\tU[s_2+s_3-2-p;2s_2+2s_3-2](z)
\end{align}
Using again \eqref{eq:master-1-loop-OPE} and the definition of the double-$\tJ$ operator \eqref{eq:normal-order-def}
\begin{align}
    \nor\tJ^{a_m}[h]\tJ^{a_n}[-h](z)\nor = \oint_{|t-z|=1}\frac{dt}{t-z_2}\tJ^{a_m}(z)\tJ^{a_n}(t) 
\end{align}
we obtain 
\begin{align}
    \eqref{eq:1st-line}&=c_{\mg}\sum_p\tg^{a_2a_3c}_p\kappa^{a_1c}\frac{[2\,3]^{p}}{p!}\frac{[1\,4]^q}{q!}\tU[h_1+s_2+s_3-3-p-q;2s_2+2s_3-2]\nn\\
    &-\frac{\tau_{\tJ}}{2}\sum_{h\in\Spec}\cU_{s_2,s_3}^{a_ma_2a_3a_n}\sum_{p,q}\tg^{a_ma_nc}_p\tg^{a_1c\bullet}_q\frac{[2\,3]^p[1\,4]^q}{p!q!}\tJ^{\bullet}[h_1-2-p-q]\nn\\
    &+\tau_{\tJ}\sum_{h\in\Spec}\cT_{s_2,s_3}^{a_ma_2a_3a_n}\sum_{p,q}\tg_p^{a_na_1c}\tg_q^{ca_m\bullet}\frac{[2\,3]^p[1\,4]^q}{q!p!}\tJ^{\bullet}[h_1-2-p-q]\nn\\
    &-c_{\mg}^2\sum_{p,q}\kappa^{a_2a_3}\frac{[2\,3]^{p}[1\,4]^q}{p!q!}\tJ^{a_1}[h_1+s_2+s_3-4-p-q;2s_2+2s_3-4]\,,
    \end{align}
\normalsize
where we have subsequently used the classical OPEs \eqref{eq:tree-OPE-1} and \eqref{eq:relevant-axion-OPE} to evaluate the last layer of integral. 
%%%%%%%%%%%%%%%%%%%%%%%%%%%%
 \paragraph{The second line of \eqref{eq:to-be-checked}.} In the second line of \eqref{eq:to-be-checked}, we have
 \begin{align}\label{eq:2nd-line}
     &\oint_{|z|=2} dz \tJ^{a_2}\Big[s_2;2s_2-1\Big](0)\oint_{|z-w|=1}\tJ^{a_3}\Big[s_3;2s_3-1\Big](z)\tJ^{a_1}[h_1](w)w\nn\\
     &=\oint_{|z|=2} dz \tJ^{a_2}\Big[s_2;2s_2-1\Big](0)\oint_{|w|=1}dw\tJ^{a_3}\Big[s_3;2s_3-1\Big](z)\tJ^{a_1}[h_1](w+z)\, (w+z)\,.
 \end{align}
 \normalsize
Note that we have made a change of variables to reach the second line above. Now, we look for structures that have first and second poles in $w$ from \eqref{eq:master-1-loop-OPE}. We obtain 
\begin{align}
    &-\oint dz\,z\tJ^{a_2}\Big[s_2;2s_2-1\Big](0)\sum_p\tg_p^{a_3a_1c}\frac{[3\,1]^p}{p!}\tJ^c\Big[s_3+h_1-1-p;2s_3-1\Big](z)\nn\\
    &+c_{\mg}\oint dz\,z\tJ^{a_2}\Big[s_2;2s_2-1\Big](0)\sum_p\kappa^{a_3a_1}\frac{[3\,1]^{p}}{p!}\tV\Big[s_3+h_1-2-p;2s_3-3\Big](z)\nn\\
    &+c_{\mg}\oint dz \,z\tJ^{a_2}\Big[s_2;2s_2-1\Big](0)\p_{z}\sum_p\kappa^{a_3a_1}\frac{[3\,1]^{p}}{p!}\tU\Big[s_3+h_1-2-p;2s_3-1\Big](z)\nn\\
        &+c_{\mg}\sum_p\kappa^{a_3a_1}\frac{[3\,1]^p}{p!}\oint dz \tJ^{a_2}\Big[s_2;2s_2-1\Big](0)\tU\Big[s_3+h_1-2-p;2s_3-1\Big](z)
\end{align}
Using \eqref{eq:relevant-axion-OPE}, we see that the second and third lines above does not contribute. Thus, we are left with 
\begin{align}
    &-\oint dz\,z\tJ^{a_2}\Big[s_2;2s_2-1\Big](0)\sum_p\tg_p^{a_3a_1c}\frac{[3\,1]^p}{p!}\tJ^c\Big[s_3+h_1-1-p;2s_3-1\Big](z)\nn\\
     &+c_{\mg}\sum_p\kappa^{a_3a_1}\frac{[3\,1]^p}{p!}\oint dz \tJ^{a_2}\Big[s_2;2s_2-1\Big](0)\tU\Big[s_3+h_1-2-p;2s_3-1\Big](z)\,,
\end{align}
\normalsize
which is evaluated to
\begin{align}
    &+\tau_{\tJ}\sum_{h\in\Spec}\sum_{p,q}\cU^{a_ma_2ca_n}_{s_2,s_3+h_1-1-p}\tg^{a_3a_1c}_p\tg_q^{a_ma_n\bullet}\frac{[3\,1]^p[2\,4]^{q}}{p!q!}\tJ^\bullet[h_1-2-p-q]\nn\\
    &+c_{\mg}\sum_{p,q}\tg^{a_3a_1c}_p\kappa^{a_2c}\frac{[3\,1]^p[2\,4]^{q}}{p!q!}\tU[h_1+s_2+s_3-3-p-q;2s_2+2s_3-4]\nn\\
        &+c_{\mg}^2\sum_{p,q}\kappa^{a_3a_1}\frac{[3\,1]^p[2\,4]^{q}}{p!q!}\tJ^{a_2}[h_1+s_2+s_3-4-p-q;2s_2+2s_3-4]\,.
\end{align}
Here, the first two lines come from second-ordered poles in the OPE \eqref{eq:master-1-loop-OPE}. 
%%%%%%%%%%%%%%%%%%%%%%%%%%%%
 \paragraph{The third line of \eqref{eq:to-be-checked}.} In computing the third line of \eqref{eq:to-be-checked}, which is
 \begin{align}\label{eq:3nd-line-1}
     \oint_{|z|=2} dz \tJ^{a_3}\Big[s_3;2s_3-1\Big](z)\oint_{|w|=1}\tJ^{a_2}\Big[s_2;2s_2-1\Big](0)\tJ^{a_1}[h_1](w)w\,,
 \end{align}
we look for structures with double poles, and obtain
\small
\begin{align}
    -c_{\mg}\sum_p\kappa^{a_1a_2}\frac{[1\,2]^{p}}{p!}\oint dz \tJ^{a_3}\Big[s_3;2s_3-1\Big](z)\tU\Big[h_1+s_2-2-p;2s_2-1\Big](0)\,.
\end{align}
\normalsize
The final result reads
\begin{align}
    \eqref{eq:3nd-line-1}&=c_{\mg}^2\kappa^{a_1a_2} %\sum_{p,q}\frac{[1\,2]^{p}[3\,4]^{q}}{p!q!}
    \tJ^{a_3}[h_1+s_2+s_3-4-p-q;2s_2+2s_3-4]
\end{align}
%%%%%%%%%%%%%%%%%%%%%%%%%%%%%%%%%%%%%%%%%%%%%%%%%%%%%%%%
 \subsection{Constraints from associativity} 
 Now, putting all of the above together, we will organize them into the following three sectors. 

\paragraph{The $\tU[h_1+s_2+s_3-3-p-q;2s_2+2s_3-2]$ sector.} This sector is unique since it contains only $\tU$ currents. To have them cancel, we require
\begin{align}
    &c_{\mg}\sum_{p,q}\tg^{a_2a_3c}_p\kappa^{a_1c}\frac{[2\,3]^{p}}{p!}\frac{[1\,4]^q}{q!}\tU[h_1+s_2+s_3-3-p-q;2s_2+2s_3-2]\nn\\
    &=c_{\mg}\sum_{p,q}\tg^{a_3a_1c}_p\kappa^{a_2c}\frac{[3\,1]^p[2\,4]^{q}}{p!q!}\tU[h_1+s_2+s_3-3-p-q;2s_2+2s_3-2]\,.
\end{align}
Observe that the above is an equality iff 
\begin{subequations}
    \begin{align}
        p=q&=0\,,\\
        h_1+s_2+s_3&=3\,.\label{eq:spin-helicity-constraint}
    \end{align}
\end{subequations}
As a result, associativity has constrained the number of derivatives which a vertex can have to be zero. Namely, the chiral algebra is well-defined at quantum level only for theories with the usual gauge interactions of Yang-Mills theory. Thus, the spectrum $\Spec$, which trivializes the contributions in various $\tJ$ sectors should reduce to\footnote{If $\Spec=\Z,2\Z+1,2\Z$, the chiral algebra $\ca$ is automatically associative to first order in quantum correction, cf. Theorem \ref{thm:quantum-associative}.}
\begin{align}
    \Spec=|h|\geq 1\,.
\end{align}
Here, we observe that the spin-helicity constraint \eqref{eq:spin-helicity-constraint} is rather restrictive. Moreover, there is not an easy way to have associativity when anomalous theories have gravitational interactions. This, in a sense, agrees with the result of \cite{Ball:2021tmb}. Namely, the chiral algebras of gravitational or higher-derivative self-dual theories do not receive quantum correction. (See, nevertheless, \cite{Bittleston:2022jeq,Fernandez:2024qnu}.) %where one can restore the associativity of $\ca$ in the case of gravitational interaction by including the one-loop correction to the OPEs involving the axion fields.
Furthermore, as the interactions are of Yang-Mills type, we can now safely remove the contribution of the $u$-channels and consider only color-ordered amplitudes.

%%%%%%%%%%%%%%%%%%%%%%%%%%%%
\paragraph{The $\tJ[h_1-2-p-q]$ sector.} From the previous sector, we will set $p=q=0$. Now, let us focus on the contributions coming from the $\tJ[h_1-2]$ currents, which are
\begin{align}\label{eq:J[h1-2]}
    &-\frac{\tau_{\tJ}}{2}\sum_{h\in\Spec}\cU_{s_2,s_3}^{a_ma_2a_3a_n}\sum_{p,q}\tg^{a_ma_nc}_p\tg^{a_1c\bullet}_q\frac{[2\,3]^p[1\,4]^q}{p!q!}\tJ^{\bullet}[h_1-2-p-q]\nn\\
    &+\tau_{\tJ}\sum_{h\in\Spec}\cT_{s_2,s_3}^{a_ma_2a_3a_n}\sum_{p,q}\tg_p^{a_na_1c}\tg_q^{ca_m\bullet}\frac{[2\,3]^p[1\,4]^q}{q!p!}\tJ^{\bullet}[h_1-2-p-q]\nn\\
    =&\tau_{\tJ}\sum_{h\in\Spec}\cU^{a_ma_2ca_n}_{s_2,s_3+h_1-1}f^{a_3a_1c}f^{a_ma_n\bullet}%\frac{[3\,1]^p[2\,4]^{q}}{p!q!}
    \tJ^\bullet[h_1-2]
\end{align}
Observe that the above is trivial if we have $\Spec=\Z\,,2\Z+1\,,$ cf. \eqref{eq:quantum-protected-spectrum}. However, as mentioned above, if we consider $\sum_h1=-1$, we will need to set
\begin{align}\label{eq:s-h-restriction}
    h_1=s_2=s_3=1\,.
\end{align}
so that 
\begin{align}\label{eq:reduce-1}
    \tau_{\tJ}\Big(-\frac{1}{2}%\sum_{h\in\Spec}
    \cU_{1,1}^{a_ma_2a_3a_n} f^{a_ma_nc}f^{a_1c\bullet}+\cT_{1,1}^{a_ma_2a_3a_n}f^{a_na_1c}f^{ca_m\bullet}-\cU^{a_ma_2ca_n}_{1,1}f^{a_3a_1c}f^{a_ma_n\bullet}\Big)\tJ^{\bullet}[-1]
\end{align}
For the value of helicity and spins in \eqref{eq:s-h-restriction}, we obtain
\begin{subequations}\label{eq:useful-1}
    \begin{align}
    \cU_{1,1}^{a_ma_2a_3a_n}&=-\frac{2}{3}\times \frac{32}{\pi^3}f^{a_ma_2e}f^{ea_3a_n}\,\\%+f^{a_1a_3e}f^{ea_4a_2}\Big)\,,\\
    \cT_{1,1}^{a_ma_2a_3a_n}&=- \frac{32}{\pi^3}\Big(f^{a_ma_2e}f^{ea_3a_n}+f^{a_na_2e}f^{ea_3a_m}\Big)\,.
\end{align}
\end{subequations}
where we recall that the contribution from the $u$-channels have been removed since the interactions are of Yang-Mills type. (Namely, we only need to consider color-ordered partial amplitudes.) At this stage, it is already clear that there cannot be external higher-spin fields when $\Spec=|h|\geq 1$. Using \eqref{eq:useful-1}, we can reduce \eqref{eq:reduce-1} to
\small
\begin{align}
    -\tau_\tJ\frac{32\sh^{\vee}}{\pi^3}\times &\Big(\frac{1}{3}f^{a_2a_3c}f^{a_1c\bullet}+\big[f^{a_ma_2e}f^{ea_3a_n}+f^{a_na_2e}f^{ea_3a_m}\big]f^{a_na_1c}f^{ca_m\bullet}+\frac{2}{3}f^{a_3a_1c}f^{a_2c\bullet}\Big)\tJ^{\bullet}[-1]\,,
\end{align}
\normalsize
where we have used \cite{Costello:2022wso} (see also \cite{Fernandez:2023abp})
\begin{align}
        \cU^{a_m bca_n}_{1,1}f^{a_ma_n\bullet}&=-\frac{2}{3}\times \frac{32}{\pi^3}f^{a_mbe}f^{eca_n}%+f^{a_mce}f^{eba_n}\Big)
        f^{a_ma_n\bullet}\nn\\
        &=-\frac{1}{3}\times \frac{32}{\pi^3}\Big(f^{a_m be}f^{eca_n}-f^{a_n be}f^{eca_m}%+f^{a_mce}f^{ebc_n}-f^{a_nce}f^{eba_m}
        \Big)f^{a_ma_n\bullet}\nn\\
        &=+\frac{1}{3}\times \frac{32}{\pi^3}f^{a_ma_ne}f^{ebc}%+f^{a_ma_ne}f^{ecb}\Big)
        f^{a_ma_n\bullet}\nn\\
        &=+\frac{2\sh^{\vee}}{3}\times\frac{32}{\pi^3}f^{bc\bullet}\,.
\end{align}
We remind the reader that $\sh^{\vee}$ is the Coxeter number associated to the quadratic Casimir in the adjoint, where
\begin{align}
    f^{abc}f^{abd}=2\sh^{\vee}\kappa^{cd}\,.
\end{align}

%%%%%%%%%%%%%%%%%%%%%%%%%%%%
\paragraph{The $\tJ[h_1+s_2+s_3-4-p-q;2s_2+2s_3-4]$ sector.} Following from the previous sector, we will set $p=q=0$ and $h_1=s_2=s_3=1$. Then,
\begin{align}
    &c_{\mg}^2\big(\kappa^{a_1a_2}\kappa^{a_3\bullet}%\frac{[1\,2]^p[3\,4]^q}{p!q!}
    +cylic(1,2,3)\big)\tJ^{\bullet}[h_1+s_2+s_3-4;2s_2+2s_3-4]\nn\\
    &\mapsto c_{\mg}^2\big(\kappa^{a_1a_2}\kappa^{a_3\bullet}+cylic(1,2,3)\big)\tJ^{\bullet}[-1]
\end{align}
%%%%%%%%%%%%%%%%%%%%%%%%%%%%%%%
\paragraph{Evaluating $\tau_\tJ$.} Combining $\tJ[h_1-2-p-q]$- and $\tJ[h_1+s_2+s_3-4-p-q;2s_2+2s_3-4]$-sector together and repeat the computation in \cite{Costello:2022upu}, we can evaluate the value of $\tau_{\tJ}$ by equating
\begin{align}
    &-\tau_\tJ\frac{32\sh^{\vee}}{\pi^3}\Big(\frac{1}{3}f^{a_2a_3c}f^{a_1c\bullet}+\frac{1}{\sh^{\vee}}\big[f^{a_ma_2e}f^{ea_3a_n}+f^{a_na_2e}f^{ea_3a_m}\big]f^{a_na_1c}f^{ca_m\bullet}+\frac{2}{3}f^{a_3a_1c}f^{a_2c\bullet}\Big)\nn\\
    &=c_{\mg}^2\big(\kappa^{a_1a_2}\kappa^{a_3\bullet}+\kappa^{a_2a_3}\kappa^{a_1\bullet}+\kappa^{a_3a_1}\kappa^{a_2\bullet}\big)\,.
\end{align}
We can now write
\small
\begin{subequations}
    \begin{align}
        &f^{a_2a_3c}f^{a_1c\bullet}=\tr\big([T^{a_1},[T^{a_2},T^{a_3}]]T^{\bullet}\big)=\frac{1}{2\sh^{\vee}}\Tr\big([T^{a_1},[T^{a_2},T^{a_3}]]T^{\bullet}\big)%-\frac{1}{2\sh^{\vee}}\Tr([T^{a_2},T^{a_3}][T^{a_1},T^{\bullet}])
        \,\\
        &f^{a_3a_1c}f^{a_2c\bullet}=\tr\big([T^{a_2},[T^{a_3},T^{a_1}]]T^{\bullet}\big)=\frac{1}{2\sh^{\vee}}\Tr\big([T^{a_2},[T^{a_3},T^{a_1}]]T^{\bullet}\big)%-\frac{1}{2\sh^{\vee}}\Tr([T^{a_3},T^{a_1}][T^{a_2},T^{\bullet}])
        \,,
    \end{align}
\end{subequations}
\normalsize
where the lift from the trace in fundamental representation to the adjoint one comes with the prices of $2\sh^{\vee}$. Next, we also have
\small
\begin{align}
    &\big[f^{a_ma_2e}f^{e a_3a_n}+f^{a_na_2e}f^{ea_3a_m}\big]f^{a_na_1c}f^{ca_m\bullet}=-\Tr(T^{a_1}T^{a_2}T^{a_3}T^{\bullet})-\Tr(T^{a_1}T^{a_3}T^{a_2}T^{\bullet})\,.
\end{align}
\normalsize
As a result, we end up with 
\begin{align}
    &-\tau_\tJ \frac{32}{\pi^3}\Big(\frac{1}{6}\Tr\big([T^{a_1}[T^{a_2},T^{a_3}]]T^{\bullet}\big)-\Tr(T^{a_1}T^{a_2}T^{a_3}T^{\bullet})-\Tr(T^{a_1}T^{a_3}T^{a_2}T^{\bullet})+\frac{1}{3}\Tr\big([T^{a_2}[T^{a_3},T^{a_1}]]T^{\bullet}\big)\nn\\
    &=c_{\mg}^2\big(\kappa^{a_1a_2}\kappa^{a_3\bullet}+\kappa^{a_2a_3}\kappa^{a_1\bullet}+\kappa^{a_3a_1}\kappa^{a_2\bullet}\big)\nn\\
    &=\frac{-iC_{\mg}}{(2\pi)^33!}\big(\kappa^{a_1a_2}\kappa^{a_3\bullet}+\kappa^{a_2a_3}\kappa^{a_1\bullet}+\kappa^{a_3a_1}\kappa^{a_2\bullet}\big)
    %&=-\frac{i}{(2\pi)^33!}\Big(\Tr(T^{a_1}T^{a_2}T^{a_3}T^{\bullet})+\Tr(T^{a_2}T^{a_3}T^{a_1}T^{\bullet})+\Tr(T^{a_3}T^{a_1}T^{a_2}T^{\bullet})\Big)\,.
\end{align}
Note that $c^2_{\mg}$ was fine-tuned such that it carries an extra minus sign compared to the Yang-Mills case. Using the identity \cite{Costello:2022upu}
\begin{align}
    &\frac{3}{2}\Big(\Tr(T^{a_1}T^{a_2}T^{a_3}T^{\bullet})+\Tr(T^{a_1}T^{a_3}T^{a_2}T^{\bullet})\Big)-C_g\Big(\kappa^{a_1a_2}\kappa^{a_3\bullet}+\kappa^{a_2a_3}\kappa^{a_1\bullet}+\kappa^{a_3a_1}\kappa^{a_2\bullet}\Big)\nn\\
    &=\frac{1}{2}\Tr(\Tr\big([T^{a_2}[T^{a_3},T^{a_1}]]T^{\bullet}\big)+\frac{1}{4}\Tr\big([T^{a_1}[T^{a_2},T^{a_3}]]T^{\bullet}\big)\,.
\end{align}
We obtain
\begin{align}
    \tau_{\tJ}=\frac{i}{2^{10}}\,.
\end{align}
Thus, we have shown that the chiral algebra can be rendered associative to first order in quantum correction after the inclusion on the axionic currents. It is also useful to remind the reader that we do not find quantum correction to the OPEs involving gravitational or higher-derivative interactions when $\Spec\neq \Z,2\Z+1,2\Z$. In a sense, this means that higher-derivative chiral theories are strongly quantum protected.

%%%%%%%%%%%%%%%%%%%%%%%%%%%%

%%%%%%%%%%%%%%%%%%%%%%%%%%%%%%%%%%%%%%%%%%%%%%%%%%%%%%%%%%%%%%%%%%
\section{4-point two-loop computation via OPEs}\label{app:loop-OPE}
This appendix computes the two-loop 4-point loop amplitudes  
\begin{align}
     \left\langle \tJ_1\big[1;1\big] \tJ_2\big[1;1\big]\tJ_3\big[1;1\big]\tJ_4\big[1;1\big]\right\rangle\,,
\end{align}
in the main text. Since the above amplitude is symmetric, we can proceed by simply computing the $s$- and $t$-color-ordered channels and sum them up. Note that by virtue of Lemmas \ref{lem:quantum-protect} and \ref{lem:mixed-term}, there will be no mixed terms. With this, let us now spell out the result of each sector, noting that the two-point functions between axionic currents will be normalized as
\begin{align}\label{eq:UU-VV-normalization}
    \lim_{z_a\rightarrow z_b}\langle \tU[a](z_a)\tU[b](z_b)\rangle=\delta_{a,b}\,,\qquad \lim_{z_a\rightarrow z_b}\langle \tV[a](z_a)\tV[b](z_b)\rangle=\delta_{a,b}\,.
\end{align}
%%%%%%%%%%%%%%%%%%%%%%%%%%%%
\paragraph{Contribution coming from $\tJ$.} We find that there no contribution from higher-derivative interactions in this sector since the sum over derivatives is of the form
\begin{align}
    \sum_p\frac{1}{p!}\frac{1}{(-p)!}\,,
\end{align}
after we making used of the helicity constraint coming from the Kronecker delta. As such, 
\begin{align}
    &\left\langle \tJ_1\big[1,1\big] \tJ_2\big[1,1\big]\tJ_3\big[1,1\big]\tJ_4\big[1,1\big]\right\rangle\Big|_{\tJ\text{-sector}}\nn\\
    &=\tau_\tJ^2\sum_{h,h'\in\Spec}\frac{(64\sh^{\vee})^2}{9\pi^6}\frac{[1\,2][3\,4]}{z_{12}^2z_{34}^2}\Big(2z_{13}^2-2z_{13}(z_{12}-z_{34})-z_{12}z_{34}\Big)\nn\\
    &\ \ +\tau_\tJ^2\sum_{h,h'\in\Spec}\frac{(64\sh^{\vee})^2}{9\pi^6}\frac{[2\,3][4\,1]}{z_{23}^2z_{41}^2}\Big(2z_{24}^2-2z_{24}(z_{23}-z_{41})-z_{23}z_{41}\Big)\,.
\end{align}
In terms of angled brackets, we have
\begin{align}
    &\left\langle \tJ_1\big[1,1\big] \tJ_2\big[1,1\big]\tJ_3\big[1,1\big]\tJ_4\big[1,1\big]\right\rangle\Big|_{\tJ\text{-sector}}\nn\\
    &=\tau_\tJ^2\sum_{h,h'\in\Spec}\frac{(64\sh^{\vee})^2}{9\pi^6}\frac{[1\,2][3\,4]}{\langle 1\,2\rangle^2\langle 3\,4\rangle^2}\Big(2\langle 1\,3\rangle^2-2\langle 1\,3\rangle\big(\langle 1\,2\rangle-\langle 3\,4\rangle\big)-\langle 1\,2\rangle \langle 3\,4\rangle\Big)\nn\\
    &\ \ +\tau_\tJ^2\sum_{h,h'\in\Spec}\frac{(64\sh^{\vee})^2}{9\pi^6}\frac{[2\,3][4\,1]}{\langle 2\,3\rangle ^2\langle 4\,1\rangle ^2}\Big(2\langle 2\,4\rangle^2-2\langle 2\,4\rangle \big(\langle 2\,3\rangle-\langle 4\,1\rangle\big)-\langle 2\,3\rangle \langle 4\,1\rangle\Big)\,.
\end{align}
Using Schouten identities, we can reduce the above to
\begin{align}
    &\left\langle \tJ_1\big[1,1\big] \tJ_2\big[1,1\big]\tJ_3\big[1,1\big]\tJ_4\big[1,1\big]\right\rangle\Big|_{\tJ\text{-sector}}\nn\\
    &=\tau_\tJ^2\sum_{h,h'\in\Spec}\frac{2(64\sh^{\vee})^2}{9\pi^6}\frac{[1\,2][3\,4]}{\langle 1\,2\rangle^2\langle 3\,4\rangle^2}\Big(\langle 1\,3\rangle\langle 2\,4\rangle+\langle 1\,4\rangle\langle 2\,3\rangle\Big)\nn\\
    &\ \ +\tau_\tJ^2\sum_{h,h'\in\Spec}\frac{2(64\sh^{\vee})^2}{9\pi^6}\frac{[2\,3][4\,1]}{\langle 2\,3\rangle ^2\langle 4\,1\rangle ^2}\Big(\langle 2\,1\rangle\langle 3\,4\rangle+\langle 2\,4\rangle\langle 4\,1\rangle\Big)\,.
\end{align}
Observe that with $\Spec=\Z,2\Z+1$, the above vanish after Riemann regularization.
%%%%%%%%%%%%%%%%%%%%%%%%%%%%
\paragraph{Contribution coming from $:\tJ\tJ:$.} We find that the contribution coming from the $:\tJ\tJ:$ operator at two loop is zero. Indeed, since this sector amounts to computing
\begin{align}\label{eq:JJ-JJ-2-loop}
    &\sum_{h,h'\in\Spec}\Big(\langle \nor\tJ[h]\tJ[-h](z_1)\nor\,\nor\tJ[h']\tJ[-h'](z_3)\nor\rangle+\sum_{h,h'\in\Spec}\langle \nor\tJ[h]\tJ[-h](z_2)\nor\,\nor\tJ[h']\tJ[-h'](z_4)\nor\rangle\Big)\,.
\end{align}
Symmetry of the 2-pt functions then forces $h=h'$. Using the definition of the normal ordered product, cf. \eqref{eq:normal-order-def}, we obtain %$\cT_{1,1}^{a_1a_2a_3a_4}=f^{a_1a_2e}f^{ea_3a_4}+f^{a_1a_3e}f^{ea_2a_4}$
\begin{align}
    \eqref{eq:JJ-JJ-2-loop}\sim 2\frac{[1\,2][3\,4]}{z_{12}z_{23}}\,,
\end{align}
up to some traces that we ignore. Therefore,
\begin{align}
    \left\langle \tJ_1\big[1,1\big] \tJ_2\big[1,1\big]\tJ_3\big[1,1\big]\tJ_4\big[1,1\big]\right\rangle\Big|_{\tJ\tJ\text{-sector}}\sim 2\tau_{\tJ}^2\frac{(32\sh^{\vee})^2}{\pi^6}\sum_{h\in\Spec}\frac{[1\,2][3\,4]}{\langle 1\,2\rangle \langle 2\,3\rangle}\,.
\end{align}

%%%%%%%%%

%%%%%%%%%%%%%%%%%%%
\paragraph{Contribution coming from $\tU$ sector.} %Recalling that $\lim_{z_a\rightarrow z_b}\langle \tU[a](z_a)\tU[b](z_b)\rangle=\delta_{a,b}$, cf. \eqref{eq:UU-VV-normalization}. 
Proceed similarly with the $\tJ$-sector, we obtain 
\begin{align}
    &\left\langle \tJ_1\big[1,1\big] \tJ_2\big[1,1\big]\tJ_3\big[1,1\big]\tJ_4\big[1,1\big]\right\rangle\Big|_{\tU\text{-sector}}\nn\\
    &=+c_{\mg}^2\frac{[1\,2][3\,4]}{\langle 1\,2\rangle\langle 3\,4\rangle}%\sum_{h\in\Spec}
    \Big(\frac{\langle 1\,3\rangle^2+2\langle 1\,3\rangle(\langle 1\,2\rangle-\langle 3\,4\rangle)-\langle 1\,2\rangle\langle 3\,4\rangle}{\langle 1\,2\rangle\langle 3\,4\rangle }\Big)\nn\\
    &\quad +c_{\mg}^2\frac{[1\,2][3\,4]}{\langle 1\,2\rangle\langle 3\,4\rangle}%\sum_{h\in\Spec}
    \Big(\frac{\langle 2\,4\rangle^2+2\langle 2\,4\rangle(\langle 2\,3\rangle-\langle 4\,1\rangle)-\langle 2\,3\rangle\langle 4\,1\rangle}{\langle 2\,3\rangle\langle 4\,1\rangle}\Big)\,.
\end{align}
\normalsize
after a short computation. Here, we have used the following useful relations
\begin{align}
    \frac{[1\,2][3\,4]}{\langle 1\,2\rangle\langle 3\,4\rangle}=\frac{[1\,3][4\,2]}{\langle 1\,3\rangle\langle 4\,2\rangle}=\frac{[1\,4][2\,3]}{\langle 1\,4\rangle\langle 2\,3\rangle}\,,
\end{align}
which can be deduced from Schouten identities.

%%%%%%%%%%%%%%%%%%%%%%%%%%%%
\paragraph{Contribution coming from $\tV$.} Finally, we compute the two-loop contributions coming from the $\tV$-sector. We get
\begin{align}
    &\left\langle \tJ_1\big[h_1,1\big] \tJ_2\big[h_2,1\big]\tJ_3\big[h_3,1\big]\tJ_4\big[h_4,1\big]\right\rangle\Big|_{\tV\text{-sector}}=c_{\mg}^2\frac{[1\,2][3\,4]}{z_{12}z_{34}}%\Bigg(\frac{\big([1\,2]+[3\,4])^{h_1+h_2-h_3-h_4}}{(h_1+h_2-h_3-h_4)!}+\frac{\big([1\,2]+[3\,4])^{h_3+h_4-h_1-h_2}}{(h_3+h_4-h_1-h_2)!}\Bigg)\nn\\
    +c_{\mg}^2\frac{[2\,3][4\,1]}{z_{23}z_{41}}\,.%\Bigg(\frac{\big([2\,3]+[4\,1])^{h_2+h_3-h_4-h_1}}{(h_2+h_3-h_4-h_1)!}+\frac{\big([2\,3]+[4\,1])^{h_4+h_1-h_2-h_3}}{(h_4+h_1-h_2-h_3)!}\Bigg)
\end{align}
\normalsize
As a result
\begin{align}
    \left\langle \tJ_1\big[1;1\big] \tJ_2\big[1;1\big] \tJ_3\big[1;1\big] \tJ_4\big[1;1\big] \right\rangle\Big|_{\tV\text{-sector}}=2c^2_{\mg}\frac{[1\,2][3\,4]}{\langle 1\,2\rangle\langle 3\,4\rangle}\,,
\end{align}

\footnotesize
\bibliography{twistor.bib}
\bibliographystyle{JHEP-2}

\end{document}